\documentclass{article}
\usepackage{mathrsfs, amsmath, amssymb, amsthm, geometry, graphicx, tikz-cd, float, xcolor,soul,bbm}
\usepackage[backend=biber, sorting=nyt, style=alphabetic]{biblatex}
\usetikzlibrary{decorations.pathreplacing,arrows,arrows.meta,fit,calc,cd,decorations.pathmorphing,decorations.shapes,backgrounds,intersections,decorations,decorations.markings,through,external}
\tikzset{vertex/.style = {shape=circle,draw,fill=black,inner sep=0pt,minimum size=5pt}}
\tikzset{edge/.style = {->,> = latex', bend right}}
\tikzset{
	super thick/.style={line width=3pt}
}
\tikzstyle{knot}=[preaction={super thick, white, draw}]
\tikzset{
    quadruple/.style args={[#1] in [#2] in [#3] in [#4]}{
        #1,preaction={preaction={preaction={draw,#4},draw,#3}, draw,#2}
    }
}
\tikzstyle{shaded}=[fill=red!10!blue!20!gray!30!white]
\tikzstyle{unshaded}=[fill=white]
\tikzstyle{empty box}=[circle, draw, thick, fill=white, opaque, inner sep=2mm]
\tikzstyle{annular}=[scale=.7, inner sep=1mm, baseline]
\tikzstyle{rectangular}=[scale=.75, inner sep=1mm, baseline=-.1cm]
\tikzstyle{mid>}=[decoration={markings, mark=at position 0.53 with {\arrow{>}}}, postaction={decorate}]
\tikzstyle{mid<}=[decoration={markings, mark=at position 0.5 with {\arrow{<}}}, postaction={decorate}]
\tikzstyle{over}=[double, draw=white, super thick, double=]
\tikzstyle{box} = [rectangle,draw,rounded corners=5pt,very thick]

\usepackage{custom-macros}
\usepackage{hyperref}

\title{Compactification of quasi-local algebras on the lattice}
\author{Jun Ikeda}
\date{\today}

\begin{document}

\maketitle

\begin{abstract}
    We introduce a compactification construction for abstract quasi-local C*-algebras over countable metric spaces equipped with an isometric group action which is functorial with respect to bounded spread isomorphisms. 
    In $1$D, the construction recovers Ocneanu’s Tube algebra for fusion spin chains, and provides a canonical bridge between infinite-volume observables and observables with periodic boundary conditions.
    We exploit this connection to derive an obstruction for the implementability of such topological symmetries as Kramers-Wannier type dualities on symmetric subalgebras. 
\end{abstract}

\tableofcontents

\section{Introduction}
Locality is an organizing principle of quantum field theory. 
In algebraic quantum field theory (AQFT), this principle is encoded in the structure of a net of local algebras \cite{Haa96}, which assigns the algebra of physical observables to each region, satisfying axioms that reflect locality (\autoref{def:quasilocal-staralgebra}).
Concretely, it is required that (\ref{it:quasilocal-staralgebra-preservescolimits}) observables in bounded regions are ``enough data,'' (\ref{it:quasilocal-staralgebra-localcommutativity}) observables in distant regions commute, and (\ref{it:quasilocal-staralgebra-algebraichaagduality}) (weak) algebraic Haag duality is satisfied (\cite{Haa96, NS97}).
The quasi-local algebra \cite{Naa17} is then obtained as the inductive limit of local algebras supported on bounded regions. 
\footnote{In this work, we use the words net of algebra and quasi-local algebra interchangeably.}

A bounded-spread isomorphism (locality-preserving map) between quasi-local algebras formalizes the idea of unitary operation that moves information only a finite distance \cite{SW04}. 
It maps local observables in a region to observables supported in its uniformly bounded neighborhood (\autoref{def:quasilocal-staralgebra-boundedspread}). 
An isomorphism with this property (and whose inverse also has finite spread) is called a bounded-spread isomorphism, and an automorphism with this property on a tensor product quasi-local algebra is the standard definition of a quantum cellular automaton (QCA) \cite{FH20}.
It is worth mentioning that there is another interpretation of such a map; it is a symmetry of the moduli space of local Hamiltonians (i.e., Hamiltonians with all terms supported on uniformly bounded regions \cite{ZCZW19, LDOV21, AFM20, EF23}). 
This gives a more meaningful notion of equivalence of physical theories than using arbitrary isomorphisms of quasi-local algebras.  

Quasi-local algebras and bounded-spread isomorphisms between them arise naturally when considering symmetric subalgebras of a tensor-product quasi-local algebra: for the usual concrete spin chain $A =\colim_I \bigotimes_{i\in I}M_d(\mathbb{C})$, an action of a finite group $G$ (or, more generally, an action encoded by a Hopf algebra or by an MPO realisation of a fusion category $\mathcal{C}$) singles out the invariant subalgebra $A^{G}$ (resp.\ $A^{\mathcal{C}}$) of operators commuting with the action \cite{PhysRevLett93070601}. 
These symmetric subalgebras are the physically relevant local observables for $G$- or $\mathcal{C}$-symmetric Hamiltonians.
Dualities are exactly the bounded spread isomorphisms between the symmetric subalgebras.
Crucially, such dualities need not extend to QCAs of the ambient tensor-product algebras (for example, Kramers-Wannier duality \cite{KW41}), and therefore are hidden in a sense.

A particularly rich class of examples in $(1+1)$-dimension, \textit{fusion spin chain}, arises from unitary fusion categories. 
It generalizes the boundary algebras of toric code \cite{KITAEV20032, Naa11}, quantum double \cite{Naa15}, and Levin-Wen model \cite{LevinWen05, kawagoe2024levinwengaugetheoryentanglement}.
Fix a unitary fusion category $\caC$ and a chosen generating object $X$.
Informally, one places a copy of the ``spin'' labeled by $X$ at each lattice site on a $1$D lattice, with observables on an interval $I$ defined to be the endomorphism algebra of the monoidal product of all spins $\End_\caC (X^{\otimes (\# I)})$ (\autoref{def:fusionspinchainalgebra}). 
Taking the inductive limit, we obtain a quasi-local algebra of a fusion spin chain.

In \cite{jones2024dhrbimodulesquasilocalalgebras}, Jones gave an invariant of abstract quasi-local algebras with respect to bounded-spread isomorphisms: the braided tensor category of DHR bimodules $_\caA \DHR_\caA$. 
It is an operator-algebraic object that encodes an emergent bulk topological field theory: the anyon types, fusion rules, and braiding \cite{evans2025operatoralgebraicapproachfusion}. 
Specializing to fusion spin chains, it is braided equivalent to the Drinfeld center of the underlying fusion category. 
Indeed, as fusion spin chains are the boundary algebras of Levin-Wen, the DHR category reflects the bulk topological order holographically.
It is recognized that given a unitary modular tensor category (UMTC) $\caZ$, not all braided categorical symmetries are implemented on-site in an arbitrary bulk microscopic lattice model whose universality class is described by $\caZ$ \cite{tu2025anomaliesglobalsymmetrieslattice}. 
Restricting to the boundary algebra, this motivates the following problem:
\begin{question}\label{question:bulktqft-symmetry-implementability}
    Given a fusion spin chain, which symmetries of the emergent bulk topological theory (braided autoequivalences of the Drinfeld center) can be implemented by a quantum cellular automaton on the chain?
\end{question}

In the physics literature, instead of taking the inductive limit of algebras of local observables, it is common to assume periodic boundary conditions on the underlying geometry. 
Such a ``compactified'' correspondence to a quasi-local algebra of a fusion spin chain is Ocneanu's Tube algebra \cite{Mug03} (\autoref{def:fusioncategory-induction}).
To benefit from analysis in the thermodynamic limit and under periodic boundary conditions, it is important to precisely understand the relationship between the quasi-local algebra and its compactified version. 
So far, it is not clear how to compactify an abstract quasi-local algebra.

\begin{question}\label{question:boundary-quasilocal-recovery}
    Given a quasi-local algebra, is there a natural notion of compactification? 
    How much information is preserved or lost after compactification?
\end{question}

In this work, we give partial answers to these two problems by proposing a functorial construction of the \textit{compactified algebra} with respect to a group of isometries $G$ on the metric space (\autoref{def:DTcptalgebra}). 
The idea is to start from a quasi-local algebra over a countable space $L$ and compactify it to the fundamental domain such that its local behaviour is preserved.
\footnote{Strictly speaking, the fundamental domain is not always compact in all dimensions.}
Note that bounded-spread isomorphisms do not preserve the class of quasi-local algebras defined by unitary fusion categories.
In other words, a non-fusion categorical quasi-local algebra can be bounded-spread isomorphic to a fusion categorical one.
Our construction, therefore, does not rely on the data of unitary fusion category; instead, we define an algebra by assuming stronger locality assumptions on quasi-local algebras and extracting local generators and local relations.
We also mention that our definition of the compactified algebra does not depend on the dimension of the underlying lattice. 
We thus expect that this will be a component of future studies of general $(d+1)$-dimensional AQFT. 

Our first result, which addresses \autoref{question:boundary-quasilocal-recovery}, is that an infinite sequence of compactified algebras can partially recover the quasi-local algebras over an infinite space (\autoref{thm:asymptotic-cptalgebra-functor}).
For example, in the case of a fusion spin chain in $(1+1)$-dimension, this is an infinite sequence of Tube algebras with increasing circumference, and it, in a sense, completely recovers the isomorphism class of quasi-local algebra (\autoref{cor:fusionspinchainalgebra-asymptotic-cptalgebra}).
\footnote{More precisely, our functorial compactification is faithful.}
This is more subtle than one might expect because, given two bounded regions $\Lambda \subseteq \Delta$, there is no natural embedding of $\Tube_\caC (X^{\otimes (\# \partial \Lambda)})$, the Tube algebra assigned to the boundary $\partial \Lambda$, into $\Tube_\caC (X^{\otimes (\# \partial \Delta)})$, and therefore we cannot simply take the inductive limit $\colim_{k} \Tube_\caC (X^{\otimes k})$.

Our second result is a partial answer to \autoref{question:bulktqft-symmetry-implementability}.
Specializing the compactification to fusion spin chains, for each bounded spread automorphism $\alpha$, we obtain $\Tube^k (\alpha)$, an automorphism of the Tube algebra (\autoref{cor:fusionspinchainalgebra-DTcptalgebra-boundedspread}).
By investigating the relationship between the symmetry of the bulk TQFT (braided autoequivalence of $\caZ (\caC)$) induced by $\alpha$ and a functor induced by $\Tube^k (\alpha^{-1})^*$, we obtain a certain stabilization property of such autoequivalences (\autoref{cor:fusionspinchainalgebra-obstruction}).
\begin{customcor}{4.9}
    Let $\caA = \caA(\caC, X)$ be a fusion spin chain algebra with $\caC$ a unitary fusion category and $X$ a strong tensor generator with $n \ge 1$.
    Let $\alpha: \caA \to \caA$ a $\ZZ$-equivariant bounded-spread isomorphism with spread $s$.
    Then, for any $k > \max \{14 n + 4 s, 6 n + 12 s\}$, $I(X^{\otimes k}) \in \caZ (\caC)$ is fixed by $\DHR(\alpha)$ up to isomorphism and $\Inv (\caZ (\caC))$ orbit.
\end{customcor}
For example, we prove obstructions to the implementability of Kramers-Wannier type dualities of abelian anyon theories \cite{KW41} on lattice models, such as $e-m$ swap in $\ZZ / 2 \ZZ$-gauge theory (\autoref{ex:fusionspinchainalgebra-dhr-picard-rep-abelian}).
\begin{customex}{4.10}[$\ZZ / 2 \ZZ$ case]
    Let $a, b \in \NN$ and $a \neq b, d:= a + b$, $H = \diag (I_{a \times a}, -I_{b \times b}) \in M_{d} (\CC)$.
    For each finite interval $I \subseteq \ZZ$, let
    \begin{align}
        \caA_I := \setbuilder{M \in (M_d (\CC))^{\otimes |I|}}{H^{\otimes |I|} M \paren{H^{\otimes |I|}}^{-1} = M }
    \end{align}
    be the symmetric subalgebra.
    Then, $e-m$ swap in the emergent TQFT is not induced by quantum cellular automata on $\caA := \colim_{I} \caA_I$. 
    % Let $\caA_{I} = M_d (\CC)^{\otimes |I|}$
\end{customex}
We also show an example of the criterion applied on a nonabelian theory, namely, $D_4$ (\autoref{ex:fusionspinchainalgebra-dhr-picard-rep-nonabelian}).

In Section \ref{sec:preliminary}, we review the definition and properties of quasi-local *-algebras, fusion spin chain algebras, and related concepts. 
In Section \ref{subsec:localgenerationpresentation}, we define two additional conditions on quasi-local algebras: \textit{local generation} and \textit{local presentation.} 
Under these assumptions, we define in Section \ref{subsec:compactifiedalgebra} the \textit{compactified algebra.} 
We first define a parameter-dependent version and use the infinite sequence and categorical colimit technique to remove the dependence. 
In Section \ref{subsec:fusionspinchaincompactified}, we apply the general compactification developed in the previous section to fusion spin chains and verified that it matches the usual Tube algebra. 
We then study its implications for the implementability of the symmetries of bulk TQFT on the lattice model in Section \ref{subsec:obstruction}.

\section*{Acknowledgements}
The author is deeply grateful to the Dale and Suzanne Burger Fellowship and Caltech SURF Program 2025 for their generous support of the summer research. 

The author also thanks Corey Jones for invaluable mentorship and guidance, and his students for helpful instruction and sharing their expertise.

Furthermore, the author wishes to thank Nori Ohata, Akihiko Iyoda, Yoshihiro Takahashi, Shinji Kaino, the Abe Ryo Foundation, Ineko Shimizu, Reo Yanagi, Taisei Nakata, and Takashi Nakayama for their support regarding relevant experiences.

\section{Preliminary notations, definitions, and properties}\label{sec:preliminary}
Throughout the paper, let $(L, d)$ denote a countable metric space with bounded geometry (i.e., for all $D \ge 0$ there exists $N_D \ge 0$ such that for any ball $B$ of diameter at most $D$, $\#B \le N_D$). 
In particular, the metric topology is discrete.
Let $B_r (x) = \{y \in L \mid d(x, y) < r\}$.

For $U \subseteq L$, let
\begin{align}
    \frP   (U) =& \Open(U) = \{V \subseteq U\}  \\
    \frF   (U) =& \{V \in \frP (U); \# V < \infty\} \\
    \frB   (U) =& \{B_r (x) \subseteq U; x \in U, r > 0\} \\
    \frP_D (U) =& \frF_{D} (U) = \{V \in \frP (U); \diam (V) \le D\} \\ 
    \frB_D (U) =& \{B_r (x) \in \frB (U); 2 r \le D\}. 
\end{align}

The base field of algebras is always taken to be $\CC$ unless otherwise specified.
Let $\stAlg$ denote the category of *-algebras over $\CC$. 
Given a *-algebra $\caA$, let $\stAlg (\subseteq \caA)$ denote the category of the poset of *-subalgebras of $\caA$ with morphisms being the inclusions of *-subalgebras.

Given two algebras $R \subseteq S$ and a two-sided ideal $J \subseteq S$, let $(J)^c_R := J \cap R$ denote the contraction of $J$ to a two-sided ideal of $R$.
If $I \subseteq R$, let $(I)^e_S := S I S$ denote the extension of $I$ to a two-sided ideal of $S$.

\subsection{Quasi-Local *-Algebras}
Recall the definition of quasi-local algebra. 
It is common to call the precosheaf a net of local algebras and its inductive limit a quasi-local algebra.
In this paper, we call the whole structure a quasi-local algebra.
\begin{definition}[\cite{Naa17}]\label{def:quasilocal-staralgebra}
    A \underline{quasi-local *-algebra} $\caA$ over $L$ consists of a *-algebra $\caA_L$ and a precosheaf over $L$
    \begin{align}
        \caA_{(\cdot)}: \frP (L) &\to \stAlg (\subseteq \caA_L) \\
        U &\mapsto \caA_U \\
        (V \subseteq U) &\mapsto (\iota_{V \subseteq U}: \caA_V \hookrightarrow \caA_U)
    \end{align}
    such that 
    \begin{enumerate}
        \item \label{it:quasilocal-staralgebra-preservescolimits}
            \underline{finite determination}:
            it preserves colimits of the form $U = \colim_{V \in \frF (U)} V$ for all $U \in \frP (L)$.
            That is, $\caA_U$ is the smallest *-subalgebra of $\caA$ containing $\caA_{V}$ for all $V \in \frF (U)$.
            \begin{align}
                \caA_U = \caA_{\bigcup_{V \in \frF (U)} V} = \caA_{\colim_{V \in \frF (U)} V} = \colim_{V \in \frF (U)} \caA_{V} = \bigcup_{V \in \frF (U)} \caA_{V}.
            \end{align}
        \item \label{it:quasilocal-staralgebra-localcommutativity}
            \underline{$\delta$-local commutativity}
            \footnote{
                This is stronger than the one in \cite{Naa17}, which reads $\ker (\caA_U * \caA_V \to \caA_{U \cup V}) \supseteq [\caA_U, \caA_V]$.
                Our extra assumption excludes the pathological example where $\caA_U$ and $\caA_V$ are two mutually commuting algebras but $\caA_U \otimes \caA_V$ does not inject into $\caA_{U \cup V}$ for sufficiently distant $U, V$.
                Intuitively, observables in distant regions should be \textit{unentangled}.
            }:
            there exists $\delta \ge 0$ such that for any $U, V \in \frP (L)$ with $d(U, V) > \delta$, $\ker (\caA_U * \caA_V \to \caA_{U \cup V}) = [\caA_U, \caA_V]$.
        \item \label{it:quasilocal-staralgebra-algebraichaagduality}
            \underline{$(R, \gamma)$-(weak) algebraic Haag duality \cite{Haa96, NS97}}: 
            there exists $R, \gamma \ge 0$ such that for every $F \in \caB (L)$ with $\diam(F) \ge R$, $Z_{\caA}(\caA_{F^c}) \subseteq \caA_{F^{+\gamma}}$ where $Z_{\caA}(\cdot)$ denotes the centralizer in $\caA$.
    \end{enumerate}
\end{definition}
\begin{remark}\label{rem:quasilocal-cstaralgebra}
    We mainly consider quasi-local *-algebras, but with modifications, the definitions and some of the results can be adapted to quasi-local C*-algebras, which are more commonly used \cite{Naa17, jones2024dhrbimodulesquasilocalalgebras}. 
    One advantage of working in the C*-setting is that it contains local observables with tails that are spread over space.  
    The definition reads as follows.
    Let $\CstAlg$ denote the category of C*-algebras over $\CC$, and let $\CstAlg(\subseteq\caA)$ denote the poset of C*-subalgebras of $\caA$ (with inclusions as morphisms).
    A \underline{quasi-local C*-algebra} $\caA$ over $L$ consists of a C*-algebra $\caA_L$ and a precosheaf over $L$ 
    \begin{align}
        \caA_{(\cdot)}: \frP (L) &\to \CstAlg (\subseteq \caA_L) \\
        U &\mapsto \caA_U \\
        (V \subseteq U) &\mapsto (\iota_{V \subseteq U}: \caA_V \hookrightarrow \caA_U)
    \end{align}
    such that
    \begin{enumerate}
        \item \label{it:quasilocal-cstaralgebra-preservescolimits}
            \underline{finite determination}:
            it preserves colimits of the form $U = \colim_{V \in \frF (U)} V$ for all $U \in \frP (L)$.
            That is, $\caA_U$ is the smallest C*-subalgebra of $\caA$ containing $\caA_{V}$ for all $V \in \frF (U)$.
            \begin{align}
                \caA_U = \caA_{\bigcup_{V \in \frF (U)} V} = \caA_{\colim_{V \in \frF (U)} V} = \colim_{V \in \frF (U)} \caA_{V} = \overline{\bigcup_{V \in \frF (U)} \caA_{V}}^{\norm{\cdot}}.
            \end{align}
        \item \label{it:quasilocal-cstaralgebra-localcommutativity}
            \underline{$\delta$-local commutativity}:
            there exists $\delta \ge 0$ such that for any $U, V \in \frP (L)$ with $d(U, V) > \delta$, $\ker (\caA_U * \caA_V \to \caA_{U \cup V}) = [\caA_U, \caA_V]$.
        \item \label{it:quasilocal-cstaralgebra-algebraichaagduality}
            \underline{$(R, \gamma)$-(weak) algebraic Haag duality}: 
            there exists $R, \gamma \ge 0$ such that for every $F \in \caB (L)$ with $\diam(F) \ge R$, $Z_{\caA}(\caA_{F^c}) \subseteq \caA_{F^{+\gamma}}$ where $Z_{\caA}(\cdot)$ denotes the centralizer in $\caA$.
    \end{enumerate}
\end{remark}

Establishing the definition of algebras, we now consider morphisms between quasi-local *-algebras to form a category. 
Arbitrary *-algebra homomorphisms between quasi-local *-algebras do not necessarily make sense physically, as any locality-preserving symmetry or discrete time evolution should not propagate information an infinite distance in a single time step.
This observation motivates the following definition. 

\begin{definition}[\cite{SW04}]\label{def:quasilocal-staralgebra-boundedspread}
    Let $\caA, \caB$ be quasi-local *-algebras over $L$ and $\alpha: \caA_L \to \caB_L$ be a *-algebra homomorphism.
    $\alpha$ is a \underline{bounded-spread *-algebra homomorphism} if there exists $s \ge 0$ such that for every $U \in \frF (L)$, 
    \begin{align} \label{eq:quasilocal-staralgebra-boundedspread-condition}
        \alpha|_{\caA_U}: \caA_U \to \caB_{U^{+s}}, \quad U^{+s} := \{ x \in L \mid d(x, U) \le s \}.
    \end{align}
    In other words, $\alpha_{(\cdot)}: \caA_{(\cdot)} \to \caB_{(\cdot)^{+s}}$ defines a natural transformation by restriction $\alpha_U = \alpha|_{\caA_U}$.

    $\alpha$ is a \underline{bounded-spread isomorphism} if $\alpha$ is an isomorphism and if both $\alpha, \alpha^{-1}$ have bounded spread.
    Furthermore, if $\caA = \caB$, $\alpha$ is a \underline{quantum cellular automaton (QCA)}.
\end{definition}

If $\alpha$ is an inner automorphism (i.e., $\alpha(\cdot) = \paren*{\prod_i v_i} (\cdot) \paren*{\prod_i v_i^*}$ for some unitaries $v_i \in \caA_{U_i}$ and disjoint family $\{U_i\}$), then this is often called a finite-depth quantum circuit (FDQC). 
It is shown that the subgroup of FDQC in the group of QCA induces trivial braided autoequivalences of the bulk TQFT \cite{jones2024dhrbimodulesquasilocalalgebras}.

Often, we consider algebras of observables and discrete time evolution that ``look the same everywhere.''
For example, when $L = \ZZ$, many examples of interest have translation covariance. 

\begin{definition}\label{def:quasilocal-staralgebra-gcovariant}
    Let $\caA$ be a quasi-local *-algebra over $L$. 
    Given a subgroup $G$ of isometry group $\Iso(L, d)$, $\caA$ is \underline{$G$-covariant} if there is a group homomorphism $\phi: G \to \Aut (\caA_L)$ such that $\phi_g (\caA_U) = \caA_{gU}$ for any $g \in G$ and $U \in \frP (L)$.

    Let $\alpha: \caA \to \caB$ be a bounded-spread *-algebra homomorphism of spread $s \ge 0$ between two $G$-covariant quasi-local *-algebras over $L$ $(\caA, \phi^\caA), (\caB, \phi^\caB)$.
    $\alpha$ is \underline{$G$-equivariant} if for any $g \in G, U \in \frP(L)$, the following diagram commutes:
    \begin{equation}\begin{tikzcd}
        \caA_U \arrow[r, "\alpha_U"] \arrow[d, "\phi^\caA_g"] & \caB_{U^{+s}} \arrow[d, "\phi^\caB_g"] \\
        \caA_{gU} \arrow[r, "\alpha_{gU}"'] & \caB_{(gU)^{+s}}
    \end{tikzcd}\end{equation}
\end{definition}

$G$-covariance is an additional structure and not a property of $\caA$. 
For example, given a QCA $\alpha$ on $\caA$ with bounded-spread $s=0$, pre-composition or post-composition with $\phi_g$ yields a new $G$-covariance structure $\phi'_g$.

With the quasi-local algebras and bounded-spread homomorphisms defined, we form a category. 

\begin{definition}\label{def:quasilocal-staralgebra-cat-gcovariantcat}
    Let $\QLstAlg_{L}$ denote the category of quasi-local *-algebras over $L$ with objects: quasi-local *-algebras $\caA$ over $L$, and morphisms: bounded-spread *-algebra homomorphisms $\alpha$.

    Let $\QLstAlg_{L, G}$ denote the category of $G$-covariant quasi-local *-algebras over $L$ with objects: $G$-covariant quasi-local *-algebras $(\caA, \phi)$ over $L$, morphisms: $G$-equivariant bounded-spread *-algebra homomorphisms $\alpha$.
\end{definition}

It was realized in \cite{jones2024dhrbimodulesquasilocalalgebras} that DHR theory of superselection sectors formulated in the setting of quasi-local C*-algebras forms a braided tensor category.
Furthermore, a bounded-spread isomorphism induces a braided equivalence between the corresponding DHR categories.

For later adoption to *-algebra instead of C*-algebra, we define a pre-DHR bimodule, which is a pre-Hilbert module with the same localization property.
In other words, we remove the norm complete condition. 
For more detail, see \cite{Arveson1976Invitation}.
In our application to fusion spin chains, these two categories are equivalent. 

\begin{definition}[\cite{jones2024dhrbimodulesquasilocalalgebras}]\label{def:quasilocal-staralgebra-dhr}
    Let $\caA$ be a quasi-local C*-algebra (*-algebra) over $L$.
    A \underline{(pre-)DHR bimodule} is a (pre-)Hilbert $(\caA, \caA)$-bimodule $M$ such that for any sufficiently large $U \in \frF (L)$, there is a right projective basis $\{b^U_i\}$ such that for all $a \in \caA_{U^c}, b^U_i \triangleleft a = a \triangleright b^U_i$.
    Let $_\caA \DHR_\caA$ ($_\caA \pDHR_\caA$) denote the full subcategory of the category of (pre-)Hilbert $(\caA, \caA)$-bimodules whose objects are (pre-)DHR bimodules and morphisms are adjointable bimodule maps.

    $_\caA \DHR_\caA$ has a canonical braiding. 
    If $\alpha: \caA \to \caB$ is a bounded-spread isomorphism, then $\DHR(\alpha) = (\alpha^{-1})^*: _\caA \DHR_\caA \to _\caB \DHR_\caB$ is a braided autoequivalence.
    $(\alpha^{-1})^*$ denotes the pullback functor of (pre-)Hilbert bimodules along $\alpha^{-1}$: for a (pre-)Hilbert $(\caA, \caA)$-bimodule $M$ with $\braket{\cdot}{\cdot}: M \times M \to \caA$, 
    $(\alpha^{-1})^* (M)$ has $M$ as the underlying $\CC$-vector space with the actions and the bilinear form defined by
    \begin{align}
        b \triangleright_{\alpha^{-1}} m &:= \alpha^{-1} (b) \triangleright m \\ 
        m \triangleleft_{\alpha^{-1}} b &:= m \triangleleft \alpha^{-1} (b) \\
        \braket{m}{m'}_{\alpha^{-1}} &:= \alpha (\braket{m}{m'})
    \end{align}
\end{definition}

Observe that we need the invertibility to pull back the bilinear form. 
This is why we mainly consider the core (maximal groupoid) of our category $\QLstAlg_L$. 

\subsection{Properties of Fusion Spin Chains}
Recall that a unitary fusion category is a finitely semisimple rigid unitary tensor category in which the monoidal unit is simple \cite{EGNO15}. 
It is a categorification of simple algebras, and it generalizes representation categories of a Hopf algebra, while still qualifying to be ``symmetries'' in a broad sense \cite{FMT22}. 
Fusion spin chain is a particularly well-studied class of quasi-local algebra built from the input data of a unitary fusion category and a generating object, and it covers many important physical models.

\begin{definition}[\cite{jones2024dhrbimodulesquasilocalalgebras}]\label{def:fusionspinchainalgebra}
    Let $\caC$ be a unitary fusion category and $X \in \Ob(\caC)$ be a strong tensor generating object (i.e., there exists $n \ge 1$ such that for all $a \in \Irr(\caC)$ the multiplicity $\dim \Hom_\caC (a, X^{\otimes n})$ is at least $1$).
    Then, the \underline{fusion spin chain algebra} $\caA (\caC, X)$ is the quasi-local *-algebra over $\ZZ$ defined by
    \begin{align}
        \caA_I := \End_{\caC} (X^{\otimes \# I}), \quad I \in \frB (\ZZ)
    \end{align}
    where $\frB (\ZZ)$ is the set of finite intervals.
    There is an obvious inclusion $\End_{\caC} (X^{\otimes \# I}) \hookrightarrow \End_{\caC} (X^{\otimes \# J})$ for $I \subseteq J$.
    For $I = I_1 \sqcup \cdots \sqcup I_l \in \frF (\ZZ)$ with each $I_i \in \frB (\ZZ)$ and $d(I_i, I_j) > 1$, we define $\caA_I := \bigotimes_{i=1}^l \caA_{I_i}$.
    Thus for any $I, J \in \frF (\ZZ)$ with $d(I, J) > 1 =: \delta$, $\ker (\caA_I * \caA_J \to \caA_{I \cup J}) = [\caA_I, \caA_J]$. 
    $\caA$ satisfies algebraic Haag duality with $R = n$ and $\gamma = 0$.
\end{definition}

Observe that fusion spin chain algebras have an obvious translation covariance. 

\begin{proposition}\label{prop:fusionspinchainalgebra-zzcovariant}
    Let $\caA = \caA (\caC, X)$ be a fusion spin chain algebra over $\ZZ$ with $(X, n)$ a strong tensor generating object.
    Then, $\caA$ is $\ZZ$-covariant with the group homomorphism $\phi: \ZZ \to \Aut (\caA_L)$ defined by obvious translation isomorphism $\phi_g: \caA_I \to \caA_{I+g}$.
\end{proposition}

\begin{example}\label{ex:fusionspinchainalgebra-tc-levinwen}
    For example, the boundary algebra of a half plane in Kitaev's quantum double model \cite{KITAEV20032, Naa11, Naa15} is a fusion spin chain with $X = \bigoplus_{a \in \Irr(\caC)} a$ and $\caC = \Rep(G)$ for a finite group $G$. 
    More generally, Levin-Wen's string net model \cite{LevinWen05, kawagoe2024levinwengaugetheoryentanglement} has a fusion spin chain as the boundary algebra of a half plane with $X = \bigoplus_{a \in \Irr(\caC)} a$.
\end{example}

Recall the following standard isomorphism.

\begin{proposition}\label{prop:fusioncategory-factorization}
    Let $\caC$ be a unitary fusion category and $X, Y, Z, Z' \in \Ob(\caC)$.
    Then, if for all $s \in \Irr (\caC)$ $m_{Z,s} = \dim \caC (Z \to s) \ge 1$,
    \begin{align}
        \caC (Z \to Z') \otimes_{\End_\caC (Z)} \caC (X \to Y \otimes Z \otimes W) \cong \caC (X \to Y \otimes Z' \otimes W),
    \end{align}
    as $(\End_\caC (Y) \otimes \End_\caC (Z') \otimes \End_\caC (W), \End_\caC (X))$-bimodules.
\end{proposition}

Fusion spin chain algebras satisfy the following cosheaf-like property: any local algebra supported on a region can be factored into a tensor product of algebras supported on uniformly bounded regions, and the relations come from the intersection.
This follows directly from \autoref{prop:fusioncategory-factorization} and an isomorphism defined by a $6j$-symbol.

Later, we will generalize this property so that it becomes preserved by bounded-spread isomorphisms for the construction of the general compactified algebras.

\begin{corollary}\label{cor:fusionspinchainalgebra-factorization}
    Let $\caA = \caA (\caC, X)$ be a fusion spin chain algebra over $\ZZ$ with $(X, n)$ a strong tensor generating object.
    Then, for any $I, J \in \frF (\ZZ)$ with $\# (I \cap J) \ge n$, the multiplication maps 
    \begin{align}
        m_{I, J}: \caA_I \otimes_{\caA_{I \cap J}} \caA_J &\to \caA_{I \cup J} \\
        m_{J, I}: \caA_J \otimes_{\caA_{I \cap J}} \caA_I &\to \caA_{I \cup J}
    \end{align}
    are isomorphisms of $(\caA_I, \caA_J)$-bimodules and $(\caA_J, \caA_I)$-bimodules, respectively. 
    \begin{equation}
        \tikzmath{
            \draw[thick] (-.6,-1) -- (-.6,1);
            \draw[thick] (-.2,-1) -- (-.2,1);
            \draw[thick] (.2,-1) -- (.2,1);
            \draw[thick] (.6,-1) -- (.6,1);
            \roundNbox{fill=white}{(0,0)}{.3}{.6}{.6}{$\scriptstyle \varphi$}
        }
        =
        \sum_{i} 
        \tikzmath{
            \draw[thick] (-.6,-1) -- (-.6,1);
            \draw[thick] (-.2,-1) -- (-.2,1);
            \draw[thick] (.2,-1) -- (.2,1);
            \draw[thick] (.6,-1) -- (.6,1);
            \roundNbox{fill=white}{(-.2,.5)}{.3}{.3}{.3}{$a_I^i$}
            \roundNbox{fill=white}{(.2,-.5)}{.3}{.3}{.3}{$a_J^i$}
        }
        =
        \sum_{j}
        \tikzmath{
            \draw[thick] (-.6,-1) -- (-.6,1);
            \draw[thick] (-.2,-1) -- (-.2,1);
            \draw[thick] (.2,-1) -- (.2,1);
            \draw[thick] (.6,-1) -- (.6,1);
            \roundNbox{fill=white}{(.2,.5)}{.3}{.3}{.3}{$b_J^j$}
            \roundNbox{fill=white}{(-.2,-.5)}{.3}{.3}{.3}{$b_I^j$}
        }
    \end{equation}
\end{corollary}

Here, we use the graphical notation consistent with \cite{henriques2015categorifiedtracemoduletensor}.

Now, we are ready to define a Tube algebra of a fusion spin chain, which will be our prototype of compactification of algebras. 
Among equivalent definitions, a simple one uses a lax monoidal functor left adjoint to the forgetful functor $\caZ (\caC) \to \caC$.

\begin{definition}[\cite{Ocneanu1994Chirality}]\label{def:fusioncategory-induction}
    Let $\caC$ be a unitary fusion category and $\caZ (\caC)$ be its Drinfeld center.
    Then, there is an induction (trace) functor $I: \caC \to \caZ (\caC)$ left adjoint to $F: \caZ (\caC) \to \caC$ defined by $I(X) = \int^{Y \in \caC} Y \otimes X \otimes Y^\vee \cong \bigoplus_{s \in \Irr(\caC)} s \otimes X \otimes s^\vee$ with the canonical half braiding.

    Let $X \in \Ob(\caC)$. 
    Then, the \underline{Tube algebra} associated with $X$ is a finite dimensional C*-algebra defined by 
    \begin{align}
        \Tube_\caC (X) 
        :=& \End_{\caZ(\caC)} (I(X)) \\ 
        \cong& \caC (X \to F(I(X))) \\
        \cong& \bigoplus_{s \in \Irr (\caC)} \caC (X \to s \otimes X \otimes s^\vee)
    \end{align}
    Using the last expression, an element is represented by
    \begin{equation}
        \sum_{\color{red}{s} \normalcolor \in \Irr(\caC)}
        \tikzmath{
            \draw[thick] (0,0) -- (0,2);
            \draw[thick] (2,0) -- (2,2);
            \filldraw[thick, fill=white] (1,2) ellipse (1 and .15);
            \draw[thick, blue] (1,-.15) -- (1,1.85);
            \halfDottedEllipse{(0,0)}{1}{.15}
            \halfDottedEllipse[thick, red]{(0,2/2)}{1}{.15}
            \roundNboxEllipse[]{(1,2/2)}{1}{.15}{-120}{-60}{.5}{$\varphi$};
        }
    \end{equation}
    and the multiplication is given by tube stacking and resolving the simple strands by fusion relations.
\end{definition}

Again, our notation is consistent with \cite{henriques2015categorifiedtracemoduletensor}.

We recall a well-known tensor equivalence, which is useful in computing the Drinfeld center of a given unitary fusion category.
Since a Tube algebra is a bialgebra, $\Rep (\Tube_\caC (X^{\otimes n}))$ is a monoidal category with Day convolution product \cite{Day1970Closed}.

\begin{corollary}[\cite{Mug03, das2014drinfeldcenterplanaralgebra, Jones_2017}]\label{cor:fusioncategory-reptube-drinfeldcenter}
    Let $\caC$ be a unitary fusion category and $\caZ (\caC)$ be its Drinfeld center.
    Let $(X, n)$ be a strong generating object.
    Then, for any $k \ge n$, there is a canonical monoidal equivalence $\Rep (\Tube_\caC (X^{\otimes k})) \cong \caZ (\caC)$, where $\Rep (\Tube_\caC (X^{\otimes k}))$ is a category of right $\Tube_\caC (X^{\otimes k})$ modules equipped with Day convolution product. 
\end{corollary}

The next two statements describe the relationship between the Tube algebra and sufficiently small local algebras inside a quasi-local algebra.
First, we prove that such local algebras embed into the Tube algebra.
The idea is that a sufficiently small interval is contained in the boundary of a bounded region. 
We prove by a clean graphical calculus that reflects this idea.

\begin{lemma}\label{lem:fusioncategory-end-embed-tube-Xk}
    Let $\caC$ be a unitary fusion category and $\caZ (\caC)$ be its Drinfeld center. 
    Let $I \in \frB_{k-1} (\ZZ)$ be a finite interval.
    Then, there is a canonical embedding $\End_\caC (X^{\otimes \#(I)}) \hookrightarrow \Tube_\caC (X^{\otimes k})$ for any $X \in \Ob(\caC)$.
\end{lemma}
\begin{proof}
    Let 
    \begin{align}
        \tikzmath{
            \draw[thick, dashed] (-.5,0) -- (0.5, 0);
        }
        &= \id_{1_{\caC}} \\ 
        \tikzmath{
            \draw[thick, blue] (-.5,0) -- (0.5, 0);
        }
        &= \id_X \\ 
        \tikzmath{
            \draw[thick, red] (-.5,0) -- (0.5, 0);
        }
        &= \id_s, \quad s \in \Irr (\caC)
    \end{align}
    If $[I] = I / k \ZZ$ does not contain $\{k, 1\}$, 
    \begin{equation}
        \tikzmath{
            \draw[blue, thick] (-.6,-1) -- (-.6,1);
            \draw[blue, thick] (-.2,-1) -- (-.2,1);
            \draw[blue, thick] (.2,-1) -- (.2,1);
            \draw[blue, thick] (.6,-1) -- (.6,1);
            \roundNbox{fill=white}{(0,0)}{.3}{.2}{.2}{$\varphi$}
        }
        \;\mapsto\;
        \tikzmath{
            \draw[thick] (0,0) -- (0,2);
            \draw[thick] (2,0) -- (2,2);
            \filldraw[thick, fill=white] (1,2) ellipse (1 and .15);
            \draw[thick, blue] (.4,-.13) -- (0.4,1.87);
            \draw[thick, blue] (.8,-.15) -- (0.8,1.85);
            \draw[thick, blue] (1.2,-.15) -- (1.2,1.85);
            \draw[thick, blue] (1.6,-.13) -- (1.6,1.87);
            \halfDottedEllipse{(0,0)}{1}{.15}
            \halfDottedEllipse[thick, dashed]{(0,1)}{1}{.15}
            \roundNboxEllipse[]{(1,1)}{1}{.15}{-120}{-60}{.5}{$\varphi$};
        }
    \end{equation}
    If $[I]$ contains $\{k, 1\}$, 
    \begin{equation}
        \tikzmath{
            \draw[thick, blue] (-.6,-1) -- (-.6,1);
            \draw[thick, blue] (-.2,-1) -- (-.2,1);
            \draw[thick, blue] (.2,-1) -- (.2,1);
            \draw[thick, blue] (.6,-1) -- (.6,1);
            \roundNbox{fill=white}{(0,0)}{.3}{.2}{.2}{$\varphi$}
        }
        \;\mapsto\;
        \tikzmath{
            \draw[thick] (0,0) -- (0,2);
            \draw[thick] (2,0) -- (2,2);
            \filldraw[thick, fill=white] (1,2) ellipse (1 and .15);
            \roundNboxEllipse[dotted]{(1,1)}{1}{.15}{180}{0}{.5}{$\reflectbox{$\varphi$}$};
            \draw[thick, blue] (.4,-.13) -- (0.4,1.87);
            \draw[thick, blue] (.8,-.15) -- (0.8,1.85);
            \draw[thick, blue] (1.2,-.15) -- (1.2,1.85);
            \draw[thick, blue] (1.6,-.13) -- (1.6,1.87);
            \halfDottedEllipse{(0,0)}{1}{.15}
            \roundNboxEllipse[]{(1,1)}{1}{.15}{-180}{-120}{.5}{};
            \roundNboxEllipse[]{(1,1)}{1}{.15}{-60}{0}{.5}{};

        }
        \; = \; \sum_{\color{red}{s} \normalcolor \in \Irr (\caC)}
        \tikzmath{
            \draw[thick] (0,0) -- (0,2);
            \draw[thick] (2,0) -- (2,2);
            \filldraw[thick, fill=white] (1,2) ellipse (1 and .15);
            \draw[thick, blue] (.4,-.13) -- (0.4,1.87);
            \draw[thick, blue] (.8,-.15) -- (0.8,1.85);
            \draw[thick, blue] (1.2,-.15) -- (1.2,1.85);
            \draw[thick, blue] (1.6,-.13) -- (1.6,1.87);
            \halfDottedEllipse{(0,0)}{1}{.15}
            \draw[thick, dotted, red] ($ (0,1) + 2*(1,0)$) arc(0:180:{1} and {.15});
        	\draw[red, thick] (0,1) arc(-180:-140:{1} and {.15});
        	\draw[red, thick] (2,1) arc(0:-40:{1} and {.15});
            \roundNboxEllipse[]{(1,1)}{1}{.15}{-150}{-115}{.5}{};
            \roundNboxEllipse[]{(1,1)}{1}{.15}{-65}{-30}{.5}{};

        }
    \end{equation}
    where we used semisimplicity and suppressed the summation of pure tensors.
\end{proof}

The lemma above provides a clear intuiton for the simple strand around the tube: it glues the two ends. 
In fact, by semisimplicity, Yoneda lemma, and $6j$-symbol isomorphism, any $\phi \in \caC (Y \otimes Z \to Y \otimes Z)$ decomposes into
\begin{align}
    \tikzmath{
        \draw[thick] (-.4,-1) -- (-.4,1);
        \draw[thick] (.4,-1) -- (.4,1);
        \roundNbox{fill=white}{(0,0)}{.3}{.4}{.4}{$\scriptstyle \varphi$}
    } 
    = 
    \sum_{\color{red}{s} \normalcolor \in \Irr(\caC)}
    \tikzmath{
        \draw[thick] (-.4,-1) -- (-.4,-.5);
        \draw[thick] (-.4,.5) -- (-.4,1);
        \draw[thick] (.4,-1) -- (.4,-.5);
        \draw[thick] (.4,.5) -- (.4,1);
        \draw[thick, red] (0, -.5) -- (0,.5);
        \roundNbox{fill=white}{(0,.5)}{.3}{.4}{.4}{}
        \roundNbox{fill=white}{(0,-.5)}{.3}{.4}{.4}{}
    } 
    = 
    \sum_{\color{red}{s} \normalcolor \in \Irr(\caC)}
    \tikzmath{
        \draw[thick] (-.4,-1) -- (-.4,1);
        \draw[thick] (.4,-1) -- (.4,1);
        \draw[thick, red] (-.4, .5) -- (.4,-.5);
        \roundNbox{fill=white}{(-.4,.5)}{.3}{.2}{.2}{}
        \roundNbox{fill=white}{(.4,-.5)}{.3}{.2}{.2}{}
    } 
\end{align}

Next, we prove an isomorphism analogous to \autoref{cor:fusionspinchainalgebra-factorization}. 
With local algebra embeddings given in \autoref{lem:fusioncategory-end-embed-tube-Xk}, we can factor the Tube algebra into a tensor product of local subalgebras.

\begin{lemma}\label{lem:fusionspinchainalgebra-tube-factorization}
    Let $\caA = \caA (\caC, X)$ be a fusion spin chain algebra over $\ZZ$ with $(X, n)$ a strong tensor generating object.
    Then, for any $k \gg n$, choose any good cover of $\ZZ / k \ZZ$ by $I_1, I_2, ..., I_l$ such that $\# (I_i \cap I_{i+1 \mod l}) \ge n$ and $I_i \cap I_j = \emptyset$ for all $i, j$ such that $|i-j \mod l| > 1$. Then,  
    \begin{align}
        &\Tube_\caC (X^{\otimes k}) \\
        \cong& \left(\caA_{I_1} \otimesunder{\caA_{I_1 \cap I_2}} \cdots \otimesunder{\caA_{I_{l-2} \cap I_{l-1}}} \caA_{I_{l-1}}\right) \otimesunder{\caA_{I_{l-1} \cap I_l} \otimes \caA_{I_1 \cap I_0}^{\op}} \caA_{I_l},
    \end{align}
    where $I_0 := I_l - k$ and $\caA_{I_1 \cap I_0}$ acts on $\caA_{I_1}$ from the left by $g \acts f = (g \otimes \id_{X^{\# (I_1 \setminus I_0)}}) \circ f$ and on $\caA_{I_l}$ from the right by $f \acted g = f \circ (\id_{X^{\# (I_0 \setminus I_1)}} \otimes \phi_{+k} (g))$.
\end{lemma}
\begin{proof}
    Let 
    \begin{align}
        \tikzmath{
            \draw[very thick, blue] (-.5,0) -- (0.5, 0);
        }
        &= \id_{X ^ {\otimes n}} \\ 
        \tikzmath{
            \draw[thick, red] (-.5,0) -- (0.5, 0);
        }
        &= \id_s, \quad s \in \Irr (\caC)
    \end{align}
    \begin{align}
        \sum_{\color{red}{s} \normalcolor \in \Irr (\caC)}
        \tikzmath{
            \draw[thick] (0,0) -- (0,2);
            \draw[thick] (2,0) -- (2,2);
            \filldraw[thick, fill=white] (1,2) ellipse (1 and .15);
            \draw[very thick, blue] (.4,-.13) -- (0.4,1.87);
            \draw[very thick, blue] (.8,-.15) -- (0.8,1.85);
            \draw[very thick, blue] (1.2,-.15) -- (1.2,1.85);
            \draw[very thick, blue] (1.6,-.13) -- (1.6,1.87);
            \halfDottedEllipse{(0,0)}{1}{.15}
            \halfDottedEllipse[thick, red]{(0,1)}{1}{.15}
            \roundNboxEllipse[]{(1,1)}{1}{.15}{-140}{-40}{.5}{};
        }
        \; = \;
        \sum_{\color{red}{s} \normalcolor \in \Irr (\caC)}
        \tikzmath{
            \draw[thick] (0,0) -- (0,2);
            \draw[thick] (2,0) -- (2,2);
            \filldraw[thick, fill=white] (1,2) ellipse (1 and .15);
            \draw[very thick, blue] (.4,-.13) -- (0.4,1.87);
            \draw[very thick, blue] (.8,-.15) -- (0.8,1.85);
            \draw[very thick, blue] (1.2,-.15) -- (1.2,1.85);
            \draw[very thick, blue] (1.6,-.13) -- (1.6,1.87);
            \halfDottedEllipse{(0,0)}{1}{.15}
            \draw[thick, red] ($ (0,1.7)$) arc(-180:-130:{1} and {.15});
            \draw[thick, red] ($ (2,0.3)$) arc(0:-50:{1} and {.15});
            \draw[thick, red, dotted] (0,1.7) -- (2,0.3);
            \roundNboxEllipse[]{(1,1.7)}{1}{.15}{-140}{-115}{.25}{};
            \roundNboxEllipse[]{(1,1.35)}{1}{.15}{-140}{-95}{.25}{};
            \roundNboxEllipse[]{(1,1)}{1}{.15}{-110}{-70}{.25}{};
            \roundNboxEllipse[]{(1,0.65)}{1}{.15}{-85}{-40}{.25}{};
            \roundNboxEllipse[]{(1,0.3)}{1}{.15}{-65}{-40}{.25}{};
        } 
        \; = \;
        \tikzmath{
            \draw[thick] (0,0) -- (0,2);
            \draw[thick] (2,0) -- (2,2);
            \filldraw[thick, fill=white] (1,2) ellipse (1 and .15);
            \draw[very thick, blue] (.4,-.13) -- (0.4,1.87);
            \draw[very thick, blue] (.8,-.15) -- (0.8,1.85);
            \draw[very thick, blue] (1.2,-.15) -- (1.2,1.85);
            \draw[very thick, blue] (1.6,-.13) -- (1.6,1.87);
            \halfDottedEllipse{(0,0)}{1}{.15};
            \draw[very thick, dotted] (0,1.7 + 0.125) -- (2,0.3 + 0.125);
            \draw[very thick, dotted] (0,1.7 - 0.125) -- (2,0.3 - 0.125);
            \roundNboxEllipse[]{(1,1.7)}{1}{.15}{-180}{-115}{.25}{};
            \roundNboxEllipse[]{(1,1.35)}{1}{.15}{-140}{-95}{.25}{};
            \roundNboxEllipse[]{(1,1)}{1}{.15}{-110}{-70}{.25}{};
            \roundNboxEllipse[]{(1,0.65)}{1}{.15}{-85}{-40}{.25}{};
            \roundNboxEllipse[]{(1,0.3)}{1}{.15}{-65}{-0}{.25}{};
        } 
    \end{align}
    where we suppressed the summation of pure tensors.
\end{proof}

Lastly, it was proved in \cite{jones2024dhrbimodulesquasilocalalgebras} that in the case of a fusion spin chain, the DHR category is braided tensor equivalent to $\caZ (\caC)$. 

\begin{lemma}\label{lem:fusionspinchainalgebra-dhr-drinfeldcenter}\cite{jones2024dhrbimodulesquasilocalalgebras}
    Let $\caA = \caA (\caC, X)$ be a fusion spin chain algebra over $\ZZ$ with $(X, n)$ a strong tensor generating object.
    Let $\caZ (\caC)$ be the Drinfeld center.
    Then, we have canonical braided autoequivalence $\caZ (\caC) \to _\caA \DHR_\caA, (Z, \sigma) \mapsto M(Z, \sigma)$.
\end{lemma}

Combined with \autoref{def:quasilocal-staralgebra-dhr}, we obtain a braided autoequivalence of $\caZ (\caC)$ from a bounded-spread isomorphism. 

\section{Compactified Algebra by Local Presentation}\label{sec:compactification}
Now, we aim to give a general compactification of a quasi-local *-algebra which generalizes Tube algebras.  \autoref{lem:fusioncategory-end-embed-tube-Xk} and \autoref{lem:fusionspinchainalgebra-tube-factorization} suggest that the compactification should be generated by sufficiently small local subalgebras of the quasi-local algebra. 
 
\subsection{Local Generation and Presentation of Quasi-Local Algebras}\label{subsec:localgenerationpresentation}
We first generalize a locality property of fusion spin chain in \autoref{cor:fusionspinchainalgebra-factorization}, which is essential to make sense of the compactification. 
Namely, we define \textit{local generation} and \textit{local presentation.}

For a family $\Delta \subseteq \frF(L)$, let 
\begin{align}
    \bigast_\Delta \caA 
    &:= \bigast_\Delta \caA_{(\cdot)} := \bigast_{U_\alpha \in \Delta} \caA_{U_\alpha} \\
    (J)^c_\Delta &:= (J)^c_{\bigast_\Delta \caA} := J \cap \bigast_\Delta \caA \\
    (I)^e_\Delta &:= (I)^e_{\bigast_\Delta \caA} := \left(\bigast_\Delta \caA\right) I \left(\bigast_\Delta \caA\right) \\
    \sum_\Delta (\ker (\iota * \iota))^e_\Delta &:= \sum_\Delta (\ker (\iota_{(\cdot)} * \iota_{(\cdot)}))^e_\Delta \\
    &:= \sum_{U_\alpha, U_\beta \in \Delta} \left( \ker (\caA_{U_\alpha} * \caA_{U_\beta} \to \caA) \right)^{e}_{\bigast_\Delta \caA}.
\end{align}

\begin{definition}\label{def:quasilocal-staralgebra-locallygenerated}
    A quasi-local *-algebra $\caA$ over $L$ is \underline{strongly locally generated} if there exists $D \ge 0$ such that for every $U \in \frP (L)$, an open cover $\frF_{D} (U)$ makes the following sequence exact:
    \begin{equation}\begin{tikzcd}\label{eq:quasilocal-staralgebra-locallygenerated-strongsequence}
        \bigast_{\frF_D (U)} \caA \arrow[r, "\bigast \iota_{\alpha}"] & \caA_U \arrow[r] & 0,
    \end{tikzcd}\end{equation}
    where $\iota_{\alpha}: \caA_{U_{\alpha}} \hookrightarrow \caA_U$ is the inclusion.
    
    $\caA$ is \underline{(weakly) locally generated} if there exist $D, r \ge 0$ such that for every $U \in \frP (L)$, an open cover $\frF_{D} (U^{+r})$ makes the following sequence exact:
    \begin{equation}\begin{tikzcd}\label{eq:quasilocal-staralgebra-locallygenerated-weaksequence}
        \left(\bigast_{\frF_D (U^{+r})} \caA\right) \timesunder{\caA_{U^{+r}}} \caA_U \arrow[r, "\bigast \iota_{\alpha}"] & \caA_U \arrow[r] & 0,
    \end{tikzcd}\end{equation}
    where $\left(\bigast_{\frF_D (U^{+r})} \caA\right) \times_{\caA_{U^{+r}}} \caA_U$ is the fibered product.
    In other words, the map 
    \begin{align}
        \bigast_{\frF_D (U^{+r})} \caA \to \caA_{U^{+r}}
    \end{align}
    is surjective onto $\caA_U$ inside $\caA_{U^{+r}}$.
\end{definition}
Strong local generation is simpler and more intuitive; all elements are generated by elements supported on uniformly bounded regions.
However, it is not an appropriate property to consider in $\QLstAlg_L$ because it is not preserved under bounded-spread isomorphism. 
We need to allow elements in a controlled neighborhood to generate $\caA_U$.
As we will see later, weak local generation is indeed bounded-spread preserved.

Similarly, 

\begin{definition}\label{def:quasilocal-staralgebra-locallypresented}
    A quasi-local *-algebra $\caA$ over $L$ is \underline{strongly locally presented} if there exists $D \ge 0$ such that for every $U \in \frP (L)$, an open cover $\frF_{D} (U)$ makes the following sequence exact:
    \begin{equation}\begin{tikzcd}\label{eq:quasilocal-staralgebra-locallypresented-strongsequence}
        \sum_{\frF_D (U)} \left( \ker (\iota * \iota) \right)^{e}_{\frF_D (U)} \arrow[r] & \bigast_{\frF_D (U)} \caA \arrow[r, "\bigast \iota_{\alpha}"] & \caA_U \arrow[r] & 0.
    \end{tikzcd}\end{equation}

    $\caA$ is \underline{(weakly) locally presented} if there exist $T \ge D$ and $t \ge r \ge 0$ such that for every $U \in \frP (L)$, an open cover $\frF_{D} (U^{+r}), \frF_T (U^{+t})$ makes the following sequence exact:
    \begin{equation}\begin{tikzcd}\label{eq:quasilocal-staralgebra-locallypresented-weaksequence}
        \left(\sum_{\frF_T (U^{+t})} \left( \ker (\iota * \iota) \right)^{e}_{\frF_T (U^{+t})} \right)^c_{\frF_D (U^{+r})} \arrow[r]
        & \left(\bigast_{\frF_D (U^{+r})} \caA\right) \timesunder{\caA_{U^{+r}}} \caA_U \arrow[r, "\bigast \iota_{\alpha}"]
        & \caA_U \arrow[r] & 0,
    \end{tikzcd}\end{equation}
    where the first map is $a \mapsto (a, 0)$.
    This is well-defined since $a \in \sum_{\frF_T (U^{+t})} \left( \ker (\iota * \iota) \right)^{e}_{\frF_T (U^{+t})} $ maps to $0 \in \caA_U$ so that $(a, 0) \in \left(\bigast_{\frF_D (U^{+r})} \caA\right) \times_{\caA_{U^{+r}}} \caA_U$.
    In other words, 
    \begin{align}
        \sum_{\frF_T (U^{+t})} \left( \ker (\iota * \iota) \right)^{e}_{\frF_T (U^{+t})} \to \bigast_{\frF_T (U^{+t})} \caA
    \end{align} 
    is surjective onto $\ker \left(\bigast_{\frF_D (U^{+r})} \caA \to \caA_{U^{+r}} \right)$ inside $\bigast_{\frF_T (U^{+t})} \caA$.
\end{definition}

Observe that we let the relations between all pairs $V_1, V_2 \in \frF_T (U^{+t})$ generate the relations in $\caA_U$ so that $\diam (V_1 \cup V_2)$ is not uniformly bounded.
We need all such relations because, by definition, algebras over $\delta$-distant regions must commute, so $\ker (\bigast_{\frF_D (U^{+r})} \caA \to \caA_{U^{+r}})$ must contain $[\caA_{V_1}, \caA_{V_1}]$.
Recall, however, that we assumed, for all distant regions, there are no relations besides $[\caA_{V_1}, \caA_{V_2}]$. 
Thus, all nontrivial relations (that don't implement the distant commutativity) are generated by relations on regions with controlled diameters.

The condition of local presentation is analogous to cosheaf condition; however, assuming the exactness of the sequence for any cover trivializes the "entanglement" of the algebra. 

We show two elementary lemmas: that these properties are stable under increasing the parameters, and that they are preserved under bounded-spread isomorphism.
Although the notation is complicated, the idea is straightforward.
When increasing the parameters of local presentation, we always proceed by replacing each factor in a relation by generators on smaller regions. 

\begin{lemma}\label{lem:quasilocal-staralgebra-increaseparameters}
    Let $\caA$ be a quasi-local *-algebra over $L$.
    \begin{enumerate}
        \item \label{it:quasilocal-staralgebra-locallygenerated-increaseparameters}
        If $\caA$ is locally generated with parameters $D, r \ge 0$, then it is also locally generated with any $D' \ge D, r' \ge r$.
        \item \label{it:quasilocal-staralgebra-locallypresented-increaseparameters}
        If $\caA$ is locally presented with parameters $T \ge D \ge 0$ and $t \ge r \ge 0$, then it is also locally presented with any $D' \ge D, r' \ge r$, and $T' \ge \max \{D' + 2r, T\}, t' \ge t + r'$.
    \end{enumerate}
\end{lemma}
\begin{proof}
    \begin{enumerate}
        \item 
        If $D' \ge D, r' \ge r$, the surjectivity follows from
        \begin{equation}\begin{tikzcd}
            \left(\bigast_{\frF_D (U^{+r})} \caA\right) \timesunder{\caA_{U^{+r}}} \caA_U \arrow[r, two heads] \arrow[d, hookrightarrow] & \caA_U \\
            \left(\bigast_{\frF_{D'} (U^{+r'})} \caA\right) \timesunder{\caA_{U^{+r'}}} \caA_U \arrow[ur]
        \end{tikzcd}\end{equation}

        \item
        If $D' \ge D, r' \ge r$, and $T' \ge \max \{D' + 2r, T\}, t' \ge t + r'$, the exactness at $\caA_U$ (\autoref{eq:quasilocal-staralgebra-locallypresented-weaksequence}) is as above.
        Now we show the exactess at $\left(\bigast_{\frF_{D'} (U^{+r'})} \caA\right) \times_{\caA_{U^{+r'}}} \caA_U$ (\autoref{eq:quasilocal-staralgebra-locallypresented-weaksequence}).
        In general, for any algebras and maps $R \to S' \hookleftarrow S$, $\ker (R \times_{S'} S \to S) \cong \ker (R \to S')$ via projection.
        Thus, it suffices to prove that
        \begin{align}
            \ker \left(\bigast_{\frF_{D'} (U^{+r'})} \caA \to \caA_{U^{+r'}}\right) \subseteq \sum_{\frF_{T'} (U^{+t'})} \left( \ker (\iota * \iota) \right)^{e}_{\frF_{T'} (U^{+t'})}.
        \end{align}
        Take any words $\sum_i w_i \in \ker \left(\bigast_{\frF_{D'} (U^{+r'})} \caA \to \caA_{U^{+r'}}\right)$, with each factor of $w_i = a_i^{V_{i,1}} \cdots a_i^{V_{i,l_i}}$ are $a_i^{V_{i,j}} \in \caA_{V_{i,j}}$ for $V_{i,j} \in \frF_{D'} (U^{+r'})$.
        Consider one of the factors $a_i^{V_{i,j}} \in \caA_{V_{i,j}}$.
        Since $\caA$ is locally generated with $(D, r)$, $\bigast_{\frF_{D} (V_{i,j}^{+r})} \caA \to \caA_{V_{i,j}^{+r}}$ is surjective onto $\caA_{V_{i,j}}$.
        Thus, there exists $\sum_k u_k \in \bigast_{\frF_{D} (V_{i,j}^{+r})} \caA$ such that
        \begin{align}\label{eq:quasilocal-staralgebra-locallypresented-increaseparameters-relation1}
            a_i^{V_{i,j}} - \sum_k u_k 
            \in& \left( \ker (\caA_{V_{i,j}^{+r}} \ast \caA_{V_{i,j}} \to \caA_{V_{i,j}^{+r}}) \right)^{e}_{(\{V_{i,j}^{+r}, V_{i,j}\} \cup \frF_D (V_{i,j}^{+r}))} \\
            +& \sum_{\substack{U_{\alpha} \in \\ \frF_{D} (V_{i,j}^{+r})}} \left( \ker (\caA_{U_{\alpha}} \ast \caA_{V_{i,j}^{+r}} \to \caA_{V_{i,j}^{+r}}) \right)^{e}_{\{(V_{i,j}^{+r}, V_{i,j}\} \cup \frF_D (V_{i,j}^{+r}))} \\
            \subseteq& \sum_{\frF_{T'} (U^{+t'})} \left( \ker (\iota * \iota) \right)^{e}_{\frF_{T'} (U^{+t'})},
        \end{align}
        Here, we used the fact that $V_{i,j}^{+r}, V_{i,j} \in \frF_{D' + 2r} (U^{+r'+r}) \subseteq \frF_{T'} (U^{+t'})$ and $\frF_{D} (V_{i,j}^{+r}) \subseteq \frF_{T'} (U^{+t'})$.
        Up to the ideal on the right-hand side, we can replace $a_i^{V_{i,j}}$ with $\sum_k u_k$.
        Collecting all such relations, we obtain
        \begin{align}
            \sum_i w_i' \in& \bigast_{\substack{V_{\mu} \in \\ \frF_{D'} (U^{+r'})}} \bigast_{\substack{U_{\alpha} \in \\ \frF_{D} (V_{\mu}^{+r})}} \caA_{U_{\alpha}} \\
            \subseteq& \bigast_{\frF_D (U^{+r+r'})} \caA
        \end{align}
        such that
        \begin{align}
            \sum_i w_i' - \sum_i w_i \in& \sum_{\frF_{T'} (U^{+t'})} \left( \ker (\iota * \iota) \right)^{e}_{\frF_{T'} (U^{+t'})}.
        \end{align}
        Since $\sum_i w_i' \mapsto 0 \in \caA_{U^{+r'+r}}$ and $\caA$ is locally presented with $(D, r)$, $(T, t)$, we have
        \begin{align}
            \sum_i w_i' \in& \ker \left(\bigast_{\frF_D (U^{+r+r'})} \caA \to \caA_{U^{+r'+r}}\right) \\
            \subseteq& \sum_{\frF_{T} (U^{+t+r'})} \left( \ker (\iota * \iota) \right)^{e}_{\frF_{T} (U^{+t+r'})} \\ 
            \subseteq& \sum_{\frF_{T'} (U^{+t'})} \left( \ker (\iota * \iota) \right)^{e}_{\frF_{T'} (U^{+t'})},
        \end{align}
        completing the proof.
    \end{enumerate}
\end{proof}

\begin{lemma}\label{lem:quasilocal-staralgebra-boundedspreadpreserved}
    Let $\alpha: \caA \to \caB$ be a bounded-spread isomorphism between quasi-local *-algebras over $L$ with spread $s$ and $\beta := \alpha^{-1}$.
    \begin{enumerate}
        \item If $\caA$ is locally generated, then so is $\caB$.
        \item If $\caA$ is locally presented, then so is $\caB$.
    \end{enumerate}
\end{lemma}
\begin{proof}\begin{enumerate}
\item 
    If $\caA$ is locally generated with parameters $D, r \ge 0$, then $\caB$ is locally generated with parameters $D + 2 s,  r + 2s$.
    For any $U \in \frP (L)$, 
    \begin{align}
        \caB_U 
        \subseteq& \alpha (\caA_{U^{+s}}) \\
        \subseteq& \alpha \left(\bigast_{\frF_{D} (U^{+s+r})} \caA \right) \\
        \subseteq& \bigast \iota \left(\bigast_{\frF_{D + 2s} (U^{+2s+r})} \caB\right),
    \end{align}
    where the last containment follows from
    \begin{equation}\begin{tikzcd}
        \bigast_{\frF_{D} (U^{+s+r})} \caA \arrow[r, "\bigast \iota"] \arrow[d, "\bigast \alpha", hookrightarrow] & \caA_{U^{+s+r}} \arrow[d, "\alpha", hookrightarrow] \\
        \bigast_{\frF_{D + 2s} (U^{+2s+r})} \caB \arrow[r, "\bigast \iota"] & \caB_{U^{+2s+r}}.
    \end{tikzcd}\end{equation}
    $\bigast \alpha: \bigast_{\frF_{D} (U^{+s+r})} \caA \to \bigast_{\frF_{D + 2s} (U^{+2s+r})} \caB$ is induced by $\caA_{U_\alpha} \to \caB_{U_\alpha^{+s}}$.

\item 
    If $\caA$ is locally presented with parameters $T \ge D \ge 0$ and $t \ge r \ge 0$, by \autoref{lem:quasilocal-staralgebra-increaseparameters}, $\caA$ is also locally presented with parameters $D_A = D+4s, r_A = r+3s$ and some $T_A, t_A \ge 0$ (namely, any $T_A \ge \max \{D_A + 2r, T\}, t_A \ge t + r_A$).
    Then, $\caB$ is locally presented with parameters $D_B = D + 2s, r_B = r + 2s$ and $T_B = T_A + 2s, t_B = t_A + s$.
    Observe that $t_B \ge t + r_A + s = t + 4s \ge r + 4s = r_B + 2s$ and $T_B \ge D_A + 2r + 2s = D + 6s + 2r \ge D_B + 4s$, which will be used later. 

    The exactness at $\caB_U$ is as above. 
    Now, we show that 
    \begin{align}
        \ker \left(\bigast_{\frF_{D_B} (U^{+r_B})} \caB \to \caB _{U^{+r_B}}\right)
        \subseteq \sum_{\frF_{T_B} (U^{+t_B})} \left( \ker (\iota * \iota) \right)^{e}_{\frF_{T_B} (U^{+t_B})}.
    \end{align}
    
    Consider the following diagram
    \begin{equation}\begin{tikzcd}
        0 
        \arrow[r]
        & \caK^\caB_{D_B, r_B}
        \arrow[r] \arrow{d}{\bigast \beta}[swap]{\rotatebox{90}{$\sim$}}
        & \bigast_{\frF_{D_B} (U^{+r_B})} \caB 
        \arrow[r] \arrow{d}{\bigast \beta}[swap]{\rotatebox{90}{$\sim$}}
        & \caB_{U^{+r_B}}
        \arrow{d}{\beta}[swap]{\rotatebox{90}{$\sim$}} \\
        0
        \arrow[r]
        & (\bigast \beta) (\caK^\caB_{D_B, r_B})
        \arrow[r] \arrow[d, hookrightarrow]
        & \bigast_{\frF_{D_B} (U^{+r_B})} \beta (\caB)
        \arrow[r] \arrow[d, hookrightarrow]
        & \beta (\caB_{U^{+r_B}}) 
        \arrow[d, hookrightarrow] \\
        0
        \arrow[r]
        & \caK^\caA_{D_A, r_A}
        \arrow[r]
        & \bigast_{\frF_{D_A} (U^{+r_A})} \caA
        \arrow[r]
        & \caA_{U^{+r_A}},
    \end{tikzcd}\end{equation}
    where $\caK^\caB_{D_B, r_B} = \ker \left(\bigast_{\frF_{D_B} (U^{+r_B})} \caB \to \caB_{U^{+r_B}}\right)$ and $\caK^\caA_{D_A, r_A} = \ker \left(\bigast_{\frF_{D_A} (U^{+r_A})} \caA \to \caA_{U^{+r_A}}\right)$ and the inclusion $\bigast_{\frF_{D_B} (U^{+r_B})} \beta (\caB) \hookrightarrow \bigast_{\frF_{D_A} (U^{+r_A})} \caA$ is induced by $\beta (\caB_{V_{\mu}}) \hookrightarrow \caA_{V_{\mu}^{+s}}$.
    ($V_\mu \in \frF_{D_B} (U^{+r_B})$ implies $V_\mu^{+s} \in \frF_{D_B + 2s} (U^{+r_B + s}) = \frF_{D_A} (U^{+r_A})$.)
    By assumption,
    \begin{align}
        (\bigast \beta) (\caK^\caB_{D_B, r_B})
        \subseteq&
        \caK^\caA_{D_A, r_A} \\
        \subseteq&
        \sum_{\frF_{T_A} (U^{+t_A})} \left( \ker (\iota * \iota) \right)^{e}_{\frF_{T_A} (U^{+t_A})}.
    \end{align}
    Next, consider $\bigast \alpha: \bigast_{\frF_{T_A} (U^{+t_A})} \caA \to \bigast_{\frF_{T_A} (U^{+t_A})} \alpha (\caA)$, and the latter has an inclusion $\bigast_{\frF_{T_A} (U^{+t_A})} \alpha (\caA) \hookrightarrow \bigast_{\frF_{T_B} (U^{+t_B})} \caB$ induced by $\alpha (\caA_{W_\lambda}) \hookrightarrow \caB_{W_\lambda^{+s}}$.
    (Again, $W_\lambda \in \frF_{T_A} (U^{+t_A})$ implies $W_\lambda^{+s} \in \frF_{T_A + 2s} (U^{+t_A + s}) = \frF_{T_B} (U^{+t_B})$.)
    We also note that 
    \begin{equation}\begin{tikzcd}
        \bigast_{\frF_{D_B} (U^{+r_B})} \caB
        \arrow{r}{\bigast \beta}[swap]{\sim} 
        & \bigast_{\frF_{D_B} (U^{+r_B})} \beta (\caB)
        \arrow[r, hookrightarrow]
        & \bigast_{\frF_{D_A} (U^{+r_A})} \caA
        \arrow[dl, hookrightarrow] \\
        & \bigast_{\frF_{T_A} (U^{+t_A})} \caA
        \arrow{r}{\bigast \alpha}[swap]{\sim}
        & \bigast_{\frF_{T_A} (U^{+t_A})} \alpha (\caA)
        \arrow[r, hookrightarrow]
        & \bigast_{\frF_{T_B} (U^{+t_B})} \caB,
    \end{tikzcd}\end{equation}
    is not an inclusion $\bigast_{\frF_{D_B} (U^{+r_B})} \caB \hookrightarrow \bigast_{\frF_{T_B} (U^{+t_B})} \caB$ induced by $\caB_{V_\mu} \stackrel{\id}{\to} \caB_{V_\mu} \hookrightarrow \bigast_{\frF_{T_B} (U^{+t_B})} \caB$.
    Instead, it is induced by $\caB_{V_\mu} \hookrightarrow \caB_{V_\mu^{+2s}} \hookrightarrow \bigast_{\frF_{T_B} (U^{+t_B})} \caB$.
    
    We have
    \begin{align}
        (\bigast \alpha) (\bigast \beta) (\caK^\caB_{D_B, r_B})
        \subseteq&
        \sum_{\frF_{T_A} (U^{+t_A})} \left( (\bigast \alpha) (\ker (\iota * \iota)) \right)^{e}_{\alpha, \frF_{T_A} (U^{+t_A})} \\
        \subseteq&
        \sum_{\frF_{T_A} (U^{+t_A})} \left( (\bigast \alpha) (\ker (\iota * \iota)) \right)^{e}_{\frF_{T_B} (U^{+t_B})},
    \end{align}
    where $\left(\cdot\right)^{e}_{\alpha, \frF_{T_A} (U^{+t_A})}$ is short for $\left(\cdot\right)^{e}_{\left(\bigast_{\frF_{T_A} (U^{+t_A})} \alpha (\caA)\right)}$, an extension of ideal.

    Since for each $W_\lambda, W_\kappa \in \frF_{T_A} (U^{+t_A})$ and $\iota_\lambda: \caA_{W_\lambda} \to \caA_{U^{+t_A}}$, $\iota_\kappa: \caA_{W_\kappa} \to \caA_{U^{+t_A}}$, 
    \begin{equation}\begin{tikzcd}
        \ker (\iota_\lambda * \iota_\kappa) \arrow[r] \arrow{d}{\bigast \alpha}[swap]{\rotatebox{90}{$\sim$}}
        & \caA_{W_\lambda} \ast \caA_{W_\kappa} \arrow[r, "\iota_\lambda * \iota_\kappa"] \arrow{d}{\bigast \alpha}[swap]{\rotatebox{90}{$\sim$}}
        & \caA_{U^{+t_A}} \arrow{d}{\alpha}[swap]{\rotatebox{90}{$\sim$}} \\
        (\bigast \alpha) (\ker (\iota_\lambda * \iota_\kappa)) \arrow[r] \arrow[d, hookrightarrow]
        & \alpha (\caA_{W_\lambda}) \ast \alpha (\caA_{W_\kappa}) \arrow[r] \arrow[d, hookrightarrow]
        & \alpha (\caA_{U^{+t_A}}) \arrow[d, hookrightarrow] \\
        \ker \left(\caB_{W_\lambda^{+s}} * \caB_{W_\kappa^{+s}} \to \caB_{U^{+t_B}}\right) \arrow[r]
        & \caB_{W_\lambda^{+s}} \ast \caB_{W_\kappa^{+s}} \arrow[r] 
        & \caB_{U^{+t_B}},
    \end{tikzcd}\end{equation}
    \begin{align}
        \sum_{\frF_{T_A} (U^{+t_A})} \left( (\bigast \alpha) (\ker (\iota * \iota)) \right)^{e}_{\frF_{T_B} (U^{+t_B})}
        \subseteq&
        \sum_{\frF_{T_B} (U^{+t_B})} \left( \ker (\iota * \iota) \right)^{e}_{\frF_{T_B} (U^{+t_B})},
    \end{align}

    Finally, we show $((\bigast \alpha) (\bigast \beta) - \id)(\caK^\caB_{D_B, r_B})$ is contained in the ideal.
    For any $\sum_i w_i \in \caK^\caB_{D_B, r_B} \subseteq \bigast_{\frF_{D_B} (U^{+r_B})} \caB$, let $b_i^{V_{i,j}} \in \caB_{V_{i,j}}$ be a factor of $w_i = b_1^{V_{i,1}} \cdots b_i^{V_{i,l_i}}$ for $V_{i,j} \in \frF_{D_B} (U^{+r_B})$.
    \begin{align}
        (\bigast \alpha) (\bigast \beta) (b_i^{V_{i,j}}) - b_i^{V_{i,j}} \in \ker \left(\caB_{V_{i,j}^{+2s} } * \caB_{V_{i,j}} \to \caB_{U^{+r_B+2s}}\right).
    \end{align}
    Therefore,
    \begin{align}
        (\bigast \alpha) (\bigast \beta) \left(\sum_i w_i\right) - \sum_i w_i 
        \in&
        \sum_{\substack{V_{\mu} \in \\ \frF_{D_B} (U^{+r_B})}} \left( \ker (\caB_{V_{\mu}^{+2s}} * \caB_{V_{\mu}} \to \caB_{U^{+r_B+2s}}) \right)^{e}_{\frF_{D_B + 4s} (U^{+r_B + 2s})} \\
        \subseteq&
        \sum_{\frF_{T_B} (U^{+t_B})} \left( \ker (\iota * \iota) \right)^{e}_{\frF_{T_B} (U^{+t_B})},
    \end{align}
    where we used $t_B \ge r_B + 2s$ and $T_B \ge D_B + 4s$.
\end{enumerate}\end{proof}

\subsection{Compactified Algebra}\label{subsec:compactifiedalgebra}
We now define the compactified algebra, which is the main object of study in this paper.
The idea is to take all local generators (all elements over regions of a uniformly bounded size) and impose local relations, distant commutativity relations, and the relations that identify algebras on the same orbit under the group action. 

We set up the notations for the relations.
Recall that relation between distant regions $V_1, V_2$ such that $d(V_1, V_2) > \delta$ is exactly $[\caA_{V_1}, \caA_{V_2}]$.
To define a compactified algebra, the relations between such pairs play a different role from those between $V_1, V_2$ such that $d(V_1, V_2) \le \delta$.
Namely, we need to exclude from the distant commutativity relation the pairs that have intersecting $G$-orbit; otherwise, $\caA_U$ embedded in the compactification commutes with $\caA_{g U}$, which will be identified with $\caA_U$ itself.

\begin{definition}\label{def:relations}
    Let $\caA$ be a quasi-local *-algebra over $L$ with $\delta$ as in \ref{it:quasilocal-staralgebra-localcommutativity}.
    For $U \in \frP (L)$ and $T \ge 0$, we define ideals of $\bigast_{\frF_{T} (U)} \caA$ as follows:
    \begin{align}
        \caI^{\loc}_{\delta,T} (U) 
        &:= \sum_{\substack{U_{\alpha}, U_{\beta} \in  \frF_{T} (U) \\ d (U_{\alpha}, U_\beta) \le \delta}} \left( \ker (\iota_\alpha * \iota_\beta) \right)^{e}_{\frF_T (U)}, \\
        \caI^{\dist}_{\delta,T} (U) 
        &:= \sum_{\substack{U_{\alpha}, U_{\beta} \in  \frF_{T} (U) \\ d (U_{\alpha}, U_\beta) > \delta}} \left( \ker (\iota_\alpha * \iota_\beta) \right)^{e}_{\frF_T (U)}, \\
        &:= \sum_{\substack{U_{\alpha}, U_{\beta} \in  \frF_{T} (U) \\ d (U_{\alpha}, U_\beta) > \delta}} \left( [\caA_{U_{\alpha}}, \caA_{U_{\beta}}] \right)^{e}_{\frF_T (U)}.
    \end{align}
    We call $\caI^{\loc}_{\delta,T} (U)$ the \underline{local relations} and $\caI^{\dist}_{\delta,T} (U)$ the \underline{distant relations}.

    Let $(\caA, \delta)$ be $G$-covariant with $\phi: G \to \Aut (\caA_L)$.
    For $U \in \frP (L)$ and $T \ge 0$,
    \begin{align}
        \caI^{\Gdist{G}}_{\delta, T} (U) &:= \sum_{\substack{U_{\alpha}, U_{\beta} \in  \frF_{T} (U) \\ d (G U_{\alpha}, G U_\beta) > \delta}} \left( [\caA_{U_{\alpha}}, \caA_{U_{\beta}}] \right)^e _{\frF_T (U)} \subseteq \caI^{\dist}_{\delta, T} (U), \\
        \caI^{\Gorb{G}}_{T} (U) &:= \sum_{\substack{U_{\alpha} \in \\ \frF_{T} (U)}} \sum_{g \in G} \left( (\id - \phi_g) (\caA_{U_{\alpha}}) \right)^e _{\frF_T (U)},
    \end{align}
    where $G U_{\alpha}$ denotes the union of the orbit under the action of $G$.
    We call $\caI^{\Gdist{G}}_{\delta, T} (U)$ the \underline{$G$-distant relations} and $\caI^{\Gorb{G}}_{T} (U)$ the \underline{$G$-orbital relations}.
\end{definition} 

We first define a parameter-dependent compactification, namely the parameters of local generation and local presentation $(D, T)$. 
The map $j_U$ below generalizes the embeddings of the local algebras into a Tube algebra \autoref{lem:fusioncategory-end-embed-tube-Xk} as proved later in Section \ref{subsec:fusionspinchaincompactified}.
Although whether this map is injective in general is a complicated combinatorial question, we expect that in all well-behaved examples it is. 

\begin{definition}\label{def:DTcptalgebra}
    Let $(\caA, \delta)$ be a $G$-covariant locally presented quasi-local *-algebra with parameters $D, T$. 
    (i.e., there exists $r, t$ such that $\caA$ is locally generated with parameters $D, r$ and locally presented with parameters $T, t$.)
    We define the \underline{$(D, T)$-compactified algebra} as
    \begin{align}
        \Cpt_{D, T}^G (\caA, \delta) := \bigast_{\frF_{D} (L)} \caA \Big/ \left(\caI^{\loc}_{\delta, T} (L) + \caI^{\Gdist{G}}_{\delta, T} (L) + \caI^{\Gorb{G}}_{T} (L)\right)^c _{\frF_D (L)}.
    \end{align}
    For any $U \in \frF_D (L)$, there is a natural *-algebra map $j_U: \caA_U \to \Cpt_{D, T}^G (\caA, \delta)$. 
\end{definition}

\begin{remark}\label{rem:DTcptalgebra-canonicalj-Gdist}
    For any $G$-distant pair $U, V \in \frF_D (L)$ (i.e., $d(G U, G V) > \delta$), there is a canonical map $\caA_{U \cup V} \cong \caA_U \otimes \caA_V \to \Cpt_{D,T}^G (\caA, \delta)$.
    We denote this map by $j_{U \cup V} := j_U \otimes j_V$.
\end{remark}

If the parameters are sufficiently small compared to the size of the fundamental domain $L / G$, we expect that the $(D, T)$-compactified algebra is stabilized under changing the parameters. 
To prove the stabilization under certain conditions, we need the notion of free action at scale $T$. 

\begin{definition}\label{def:scaleT-freeaction}
    Let $T \ge 0$ and $G \subseteq \Iso(L, d)$.
    We say that $G$ acts \underline{freely at scale $T$} if for all $U \in \frF_T (L)$, $g U \cap U \neq \emptyset$ implies $g = \id_L$.
\end{definition}

For example, $k \ZZ \subseteq \Iso(\ZZ, d)$ acts freely at scale $k-1$. 

With this notion, we prove that the compactified algebra does not depend on the choice of $\delta$ as long as it is sufficiently small compared to the scale of action.

\begin{lemma}\label{lem:DTcptalgebra-increase-d}
    Let $(\caA, \delta)$ be a $G$-covariant locally presented quasi-local *-algebra with parameters $D, T$.
    Let $\delta' \ge \delta$.
    If $G$ acts freely at scale $2 T + 2 \delta + \delta'$, there is a canonical isomorphism $\Cpt_{D, T}^G (\caA, \delta) \cong \Cpt_{D, T}^G (\caA, \delta')$ compatible with the maps $j_U: \caA_U \to \Cpt_{D, T}^G (\caA, \delta), j'_U: \caA_U \to \Cpt_{D, T}^G (\caA, \delta')$ for any $U \in \frF_D (L)$.
\end{lemma}
\begin{proof}
    Clearly $\caI^{\loc}_{\delta, T} (L) \subseteq \caI^{\loc}_{\delta', T} (L)$ and $\caI^{\Gdist{G}}_{\delta, T} (L) \supseteq \caI^{\Gdist{G}}_{\delta', T} (L)$.
    Now, we show $\caI^{\Gdist{G}}_{\delta, T} (L) \subseteq \caI^{\loc}_{\delta', T} (L) + \caI^{\Gorb{G}}_{T} (L)$.
    \begin{align}
        \caI^{\Gdist{G}}_{\delta, T} (L) 
        =& \sum_{\substack{U_{\alpha}, U_{\beta} \in  \frF_{T} (L) \\ d (G U_{\alpha}, G U_\beta) > \delta}} \left( [\caA_{U_{\alpha}}, \caA_{U_{\beta}}] \right)^e _{\frF_T (L)} \\
        \subseteq& \sum_{\substack{U_{\alpha}, U_{\beta} \in  \frF_{T} (L) \\ \delta' < d (G U_{\alpha}, G U_\beta)}} \left( [\caA_{U_{\alpha}}, \caA_{U_{\beta}}] \right)^e _{\frF_T (L)} \\
        +& \sum_{\substack{U_{\alpha}, U_{\beta} \in  \frF_{T} (L) \\ \delta < d (G U_{\alpha}, G U_\beta) \le \delta'}} \left( [\caA_{U_{\alpha}}, \caA_{U_{\beta}}] \right)^e _{\frF_T (L)}
    \end{align}
    The first term is contained in $\caI^{\Gdist{G}}_{\delta', T} (L)$. 
    For the second term, take some $U_\alpha, U_\beta \in \frF_T (L)$ and $g \in G$ such that $d (U_\alpha, g U_\beta) \le \delta'$.
    Then, 
    \begin{align}
        \left( [\caA_{U_{\alpha}}, \caA_{U_{\beta}}] \right)^e _{\frF_T (L)} 
        \subseteq& ((\id - \phi_g) (\caA_{U_\beta}))^e_{\frF_T (L)} + \left( [\caA_{U_{\alpha}}, \caA_{g U_\beta}] \right)^e _{\frF_T (L)} \\
        \subseteq& \caI^{\Gorb{G}}_{T} (L) + \caI^{\loc}_{\delta', T} (L).
    \end{align}
    Thus, $\caI^{\Gdist{G}}_{\delta, T} (L) \subseteq \caI^{\Gdist{G}}_{\delta', T} (L) + \caI^{\loc}_{\delta', T} (L) + \caI^{\Gorb{G}}_{T} (L)$, therefore $\caI^{\loc}_{\delta, T} (L) + \caI^{\Gdist{G}}_{\delta, T} (L) + \caI^{\Gorb{G}}_{T} (L) \subseteq \caI^{\loc}_{\delta', T} (L) + \caI^{\Gdist{G}}_{\delta', T} (L) + \caI^{\Gorb{G}}_{T} (L)$.

    Next, we show that $\caI^{\loc}_{\delta', T} (L) \subseteq \caI^{\loc}_{\delta, T} (L) + \caI^{\Gdist{G}}_{\delta, T} (L)$.
    \begin{align}
        \caI^{\loc}_{\delta', T} (L) 
        =& \sum_{\substack{U_{\alpha}, U_{\beta} \in  \frF_{T} (L) \\ d (U_{\alpha}, U_\beta) \le \delta}} \left( \ker (\iota_\alpha * \iota_\beta) \right)^{e}_{\frF_T (L)} \\
        +& \sum_{\substack{U_{\alpha}, U_{\beta} \in  \frF_{T} (L) \\ \delta < d (U_{\alpha}, U_\beta) \le \delta'}} \left( \ker (\iota_\alpha * \iota_\beta) \right)^{e}_{\frF_T (L)}.
    \end{align}
    The first term is contained in $\caI^{\loc}_{\delta, T} (L)$.
    For the second term, we use \autoref{claim:scaleT-freeaction-separation} with $\delta = \delta, D = 2 T + \delta'$ to conclude 
    \begin{align}
        \sum_{\substack{U_{\alpha}, U_{\beta} \in  \frF_{T} (L) \\ \delta < d (U_{\alpha}, U_\beta) \le \delta'}} \left( \ker (\iota_\alpha * \iota_\beta) \right)^{e}_{\frF_T (L)} 
        \subseteq& \sum_{\substack{U_{\alpha}, U_{\beta} \in  \frF_{T} (L) \\ d (G U_{\alpha}, G U_\beta) > \delta}} \left( [\caA_{U_{\alpha}}, \caA_{U_{\beta}}] \right)^e _{\frF_T (L)} \\
        \subseteq& \caI^{\Gdist{G}}_{\delta, T} (L).
    \end{align}
\end{proof}

\begin{claim}\label{claim:scaleT-freeaction-separation}
    Let $U_1, U_2 \in \frF(L)$ such that $\diam (U_1 \cup U_2) \le D$ and $d(U_1, U_2) > \delta$.
    If $G$ acts freely at scale $D + 2 \delta$, then $d (G U_1, G U_2) > \delta$.
\end{claim}
\begin{proof}
    Suppose there exists $g \neq \id_L$ such that $d (U_1, g U_2) \le \delta$.
    Then, 
    \begin{align}
        &U_1 \cap g (U_2^{+\delta}) \neq \emptyset \\
        \Rightarrow& (U_1 \cup (U_2^{+\delta})) \cap g (U_1 \cup (U_2^{+\delta})) \neq \emptyset \\
        \Rightarrow& G \text{ does not act freely at scale } D + 2 \delta,
    \end{align}
    where $\diam (U_1 \cup (U_2^{+\delta})) \le \diam (U_1 \cup U_2)^{+\delta} \le D + 2 \delta$.
\end{proof}

Next, we prove that the compactified algebra stabilizes with respect to $D$ if it is sufficiently small compared to the scale of action. 
The basic idea of the proof is similar to \autoref{lem:quasilocal-staralgebra-increaseparameters}: replacing each factor with a sum of products of locally supported elements.

\begin{lemma}\label{lem:DTcptalgebra-increase-parameters}
    Let $(\caA, \delta)$ be a $G$-covariant locally presented quasi-local *-algebra.
    Let $(D, T), (D', T')$ be parameters such that $D' \ge D, T' \ge T$.
    Then,
    \begin{enumerate}
        \item There is a canonical *-homomorphism $\Cpt_{D, T}^G (\caA, \delta) \to \Cpt_{D', T'}^G (\caA, \delta)$ compatible with the maps $j_U: \caA_U \to \Cpt_{D, T}^G (\caA, \delta), j'_U: \caA_U \to \Cpt_{D', T'}^G (\caA, \delta)$ for any $U \in \frF_D (L) \subseteq \frF_{D'} (L)$.
        \item If $T = T'$, then the map is injective. \footnote{
            One may expect that local presentation makes the map an isomorphism for a sufficiently large $T' \ge T$. 
            However, this reduces to an involved combinatorial question. 
            The local presentation implies $(\caI^{\loc}_{\delta, T'} (L) + \caI^{\dist}_{\delta, T'} (L))^c_{\frF_D (L)} = \ker \paren*{\bigast_{\frF_D (L)} \caA \to \caA_L} = (\caI^{\loc}_{\delta, T} (L) + \caI^{\dist}_{\delta, T} (L))^c_{\frF_D (L)}$, or all the relations are already generated by the relations from algebras of diameter at most $T$. 
            For the desired isomorphism, we need a similar equality, but with $\caI^{\dist}$ replaced by $\caI^{\Gdist{G}}$. \label{ft:DTcptalgebra-increase-parameters-Tstabilization}}
        \item If there exists $t'$ such that $\caA$ is locally presented with parameters $(D', T', t')$ and $G$ acts freely at scale $D' + 2 t' + 2 \delta$, then the map is surjective.
    \end{enumerate}
\end{lemma}
\begin{proof}\begin{enumerate}
    \item Consider a map 
    \begin{equation}\begin{tikzcd}
        \bigast_{\frF_{D} (L)} \caA \arrow[r, "\bigast \iota"] &\bigast_{\frF_{D'} (L)} \caA \arrow[r, two heads] & \Cpt_{D', T'}^G (\caA, \delta) 
    \end{tikzcd}\end{equation} 
    Then, 
    \begin{align}
        &\ker \left(\bigast_{\frF_{D} (L)} \caA \to \Cpt_{D', T'}^G (\caA, \delta)\right) \\
        =& \left(\caI^{\loc}_{\delta, T'} (L) + \caI^{\Gdist{G}}_{\delta, T'} (L) + \caI^{\Gorb{G}}_{T'} (L)\right)^c _{\frF_{D'} (L)} \cap \bigast_{\frF_{D} (L)} \caA \\
        \supseteq& \left(\caI^{\loc}_{\delta, T} (L) + \caI^{\Gdist{G}}_{\delta, T} (L) + \caI^{\Gorb{G}}_{T} (L)\right)^c _{\frF_D (L)}.
    \end{align}

    \item If $T = T'$, then the containment above is an equality.
    \item Suppose $\caA$ is locally presented with parameters $(D, r, T, t)$ and $(D', r', T', t')$ and $G$ acts freely at scale $D' + 2 t' + 2 \delta$.
    
    Take any letter $a \in \caA_V \subseteq \bigast_{\frF_{D'} (L)} \caA$ with $V \in \frF_{D'} (L)$. 
    Then, there exists $\sum_k u_k \in \bigast_{\frF_{D} (V^{+r})} \caA \subseteq \bigast_{\frF_D (L)} \caA$ such that
    \begin{align}
        a - \sum_k u_k 
        \in& \ker \left(\caA_{V} \ast \bigast_{\frF_D (V^{+r})} \caA \to \caA_{V^{+r}}\right) \\ 
        \subseteq& \ker \left(\bigast_{\frF_{D'} (V^{+r'})} \caA \to \caA_{V^{+r'}}\right),
    \end{align}
    which is contained in 
    \begin{align}
        &\sum_{\frF_{T'} (V^{+t'})} \left( \ker (\iota * \iota) \right)^{e}_{\frF_{T'} (V^{+t'})} \\
        =& \caI^{\loc}_{\delta, T'} (V^{+t'}) + \caI^{\dist}_{\delta, T'} (V^{+t'}) \\
        =& \caI^{\loc}_{\delta, T'} (V^{+t'}) + \caI^{\Gdist{G}}_{\delta, T'} (V^{+t'}),
    \end{align}
    where, for the last equality, we used the \autoref{claim:scaleT-freeaction-separation} with $\delta = \delta, D = \diam (V^{+t'}) = D' + 2 t'$.
\end{enumerate}\end{proof}

The next lemma claims that a $(D, T)$-compactified algebra is partially functorial, giving a well-defined map under conditions on the spread and the parameters of local presentation. 
However, it is not a functor because well-definedness with a fixed group $G$ depends on the spread of the homomorphism compared to the parameters $D, T$. 
This is the motivation for the later definition of an infinite sequence of compactified algebra \autoref{def:asymptotic-DTcptalgebra}.

\begin{lemma}\label{lem:DTcptalgebra-boundedspread}
    Let $(\caA, \delta_A), (\caB, \delta_B)$ be $G$-covariant locally presented quasi-local *-algebras with parameters $(D_A, T_A), (D_B, T_B)$.
    Let $(\alpha, s): \caA \to \caB$ be a $G$-equivariant bounded-spread homomorphism.
    Suppose $D_A + 2 s \le D_B, T_A + 2 s \le T_B$ and $\delta_A \le \delta_B + 2 s$.
    Then, we have a canonical *-homomorphism $\Cpt^G (\alpha)$ such that
    \begin{equation}\begin{tikzcd}
        \caA_U \arrow[r, "\alpha"] \arrow[d, "j_U"] & \caB_{U^{+s}} \arrow[d, "j'_U"] \\
         \Cpt_{D_A, T_A}^G (\caA, \delta_A) \arrow[r, "\Cpt^G (\alpha)"] & \Cpt_{D_B, T_B}^G (\caB, \delta_B)
    \end{tikzcd}\end{equation}
    In particular, if $D_A' \ge D_A, T_A' \ge T_A$, $D_B' \ge D_B, T_B' \ge T_B$, and $D_A' + 2 s \le D_B', T_A' + 2 s \le T_B'$, then 
    \begin{equation}\begin{tikzcd}
        \Cpt_{D_A, T_A}^G (\caA, \delta_A) \arrow[r, "\Cpt^G (\alpha)"] \arrow[d] & \Cpt_{D_B, T_B}^G (\caB, \delta_B) \arrow[d] \\ 
        \Cpt_{D_A', T_A'}^G (\caA, \delta_A) \arrow[r, "\Cpt^G (\alpha)"] & \Cpt_{D_B', T_B'}^G (\caB, \delta_B)
    \end{tikzcd}\end{equation}
\end{lemma}
\begin{proof}
    Consider a map
    \begin{equation}\begin{tikzcd}
        \bigast_{\frF_{D_A} (L)} \caA \arrow[r, "\bigast \alpha"] & \bigast_{\frF_{D_B} (L)} \caB \arrow[r, two heads] & \Cpt_{D_B, T_B}^G (\caB, \delta_B).
    \end{tikzcd}\end{equation}
    We will show that $\left(\bigast \alpha\right) \left(\caI^{\loc}_{\delta_A, T_A} (L) + \caI^{\Gdist{G}}_{\delta_A, T_A} (L) + \caI^{\Gorb{G}}_{T_A} (L)\right)^c _{\frF_{D_A} (L)}$ is contained in $\caI^{\loc}_{\delta_B, T_B} (L) + \caI^{\Gdist{G}}_{\delta_B, T_B} (L) + \caI^{\Gorb{G}}_{T_B} (L)$. 
    By construction, the commutativity of the diagrams is clear. 
    \begin{enumerate}
        \item $\left(\bigast \alpha\right) \left(\caI^{\loc}_{\delta_A, T_A} (L)\right)$ 
        \begin{align}
            \left(\bigast \alpha\right) \left(\caI^{\loc}_{\delta_A, T_A} (L)\right)
            =& \sum_{\substack{U_{\alpha}, U_{\beta} \in  \frF_{T_A} (L) \\ d (U_{\alpha}, U_\beta) \le \delta_A}} \left(\left(\bigast \alpha\right)\ker (\iota_\alpha * \iota_\beta) \right)^{e}_{\alpha, \frF_{T_A} (L)} \\
            \subseteq& \sum_{\substack{U_{\alpha}, U_{\beta} \in  \frF_{T_A} (L) \\ d (U_{\alpha}, U_\beta) \le \delta_A}} \left(\ker (\caB_{U_{\alpha}^{+s} } * \caB_{U_{\beta}^{+s}}) \right)^{e}_{\frF_{T_B} (L)} \\
            \subseteq& \caI^{\loc}_{\delta_B, T_B} (L),
        \end{align}
        
        \item $\left(\bigast \alpha\right) \left(\caI^{\Gdist{G}}_{\delta_A, T_A} (L)\right)$
        \begin{align}
            \left(\bigast \alpha\right) \left(\caI^{\Gdist{G}}_{\delta_A, T_A} (L)\right)
            =& \sum_{\substack{U_{\alpha}, U_{\beta} \in \frF_{T_A} (L) \\ \delta_B + 2 s < d (G U_{\alpha}, G U_\beta)}} \left(\left(\bigast \alpha\right) [\caA_{U_{\alpha}}, \caA_{U_{\beta}}] \right)^{e}_{\alpha, \frF_{T_A} (L)} \\
            +& \sum_{\substack{U_{\alpha}, U_{\beta} \in \frF_{T_A} (L) \\ \delta_A < d (G U_{\alpha}, G U_\beta) \le \delta_B + 2 s}} \left(\left(\bigast \alpha\right) \ker (\iota_\alpha * \iota_\beta) \right)^{e}_{\alpha, \frF_{T_A} (L)}
        \end{align}
        For the first term,
        \begin{align}
            \sum_{\substack{U_{\alpha}, U_{\beta} \in \frF_{T_A} (L) \\ \delta_B + 2 s < d (G U_{\alpha}, G U_\beta)}} \left(\left(\bigast \alpha\right) [\caA_{U_{\alpha}}, \caA_{U_{\beta}}] \right)^{e}_{\alpha, \frF_{T_A} (L)} 
            \subseteq& \sum_{\substack{U_{\alpha}, U_{\beta} \in \frF_{T_A} (L) \\ \delta_B + 2 s < d (G U_{\alpha}, G U_\beta)}} \left([\caB_{U_{\alpha}^{+s}}, \caB_{U_{\beta}^{+s}}] \right)^{e}_{\frF_{T_B} (L)} \\
            \subseteq& \caI^{\Gdist{G}}_{\delta_B, T_B} (L).
        \end{align}
        For the second term, 
        \begin{align}
            &\sum_{\substack{U_{\alpha}, U_{\beta} \in \frF_{T_A} (L) \\ \delta_A < d (G U_{\alpha}, G U_\beta) \le \delta_B + 2 s}} \left(\left(\bigast \alpha\right) \ker (\iota_\alpha * \iota_\beta) \right)^{e}_{\alpha, \frF_{T_A} (L)} \\
            \subseteq& \sum_{\substack{U_{\alpha}, U_{\beta} \in \frF_{T_A} (L) \\ d (G U_{\alpha}, G U_\beta) \le \delta_B + 2 s}} \left(\ker (\caB_{U_{\alpha}^{+s}} * \caB_{U_{\beta}^{+s}}) \right)^{e}_{\frF_{T_B} (L)} \\
            \subseteq& \sum_{\substack{V_\mu, V_\nu \in \frF_{T_B} (L) \\ d (G V_\mu, G V_\nu) \le \delta_B}} \left(\ker (\caB_{V_\mu} * \caB_{V_\nu}) \right)^{e}_{\frF_{T_B} (L)}.
        \end{align}
        If $d(V_\mu, g V_\nu) \le \delta_B$, then 
        \begin{align}
            \left(\ker (\caB_{V_\mu} * \caB_{V_\nu})\right)^{e}_{\frF_{T_B} (L)}
            \subseteq& \left(\ker (\caB_{V_\mu} * \caB_{g V_\nu})\right)^{e}_{\frF_{T_B} (L)} + \left( (\id - \phi_g) (\caB_{V_\mu}) \right)^{e}_{\frF_{T_B} (L)} \\
            \subseteq& \caI^{\loc}_{\delta_B, T_B} (L) + \caI^{\Gorb{G}}_{T_B} (L).
        \end{align}

        \item $\left(\bigast \alpha\right) \left(\caI^{\Gorb{G}}_{T_A} (L)\right)$
        \begin{align}
            \left(\bigast \alpha\right) \left(\caI^{\Gorb{G}}_{T_A} (L)\right) 
            =& \sum_{\substack{U_{\alpha} \in \\ \frF_{T_A} (L)}} \sum_{g \in G} \left(\left(\bigast \alpha\right) (\id - \phi_g^{\caA}) (\caA_{U_{\alpha}}) \right)^{e}_{\alpha, \frF_{T_A} (L)} \\
            \subseteq& \sum_{\substack{U_{\alpha} \in \\ \frF_{T_A} (L)}} \sum_{g \in G} \left((\id - \phi_g^{\caB}) (\caB_{U_{\alpha}^{+s}}) \right)^{e}_{\frF_{T_B} (L)} \\
            \subseteq& \caI^{\Gorb{G}}_{T_B} (L).
        \end{align}
    \end{enumerate}
\end{proof}

To obtain a functorial parameter-independent compactification, we take an infinite family of group actions such that the fundamental domain grows to the whole space.
The idea is that if one takes a sequence of finite regions increasing in size to infinity, then the compactified algebras over the increasing fundamental domain should contain all the information of the original algebra.

\begin{definition}\label{def:asymptoticallyfreefamily}
    An infinite family $\{G_\lambda\}_{\lambda \in \Lambda}$ of subgroups of $G_{\lambda_0} \subseteq \Iso(L,d)$ is \underline{asymptotically free} if, for any $T \ge 0$, $\Lambda_{\le T} := \setbuilder{\lambda \in \Lambda}{G_\lambda \text{ does not act freely at scale } T} \subseteq \Lambda$ is finite.
\end{definition}

For example, $\{k \ZZ\}_{k \in \NN}$ acting on $(\ZZ, d)$ by translation is an asymptotically free family. 
Observe that the family of fundamental domains $\{\ZZ / k \ZZ\}$ has an increasing circumference. 

\begin{definition}\label{def:asymptotic-DTcptalgebra}
    Let $(\caA, \delta)$ be a $G_{\lambda_0}$-covariant locally presented quasi-local *-algebra with parameters $D, T$.
    Let $\{G_\lambda\}_{\lambda \in \Lambda}$ be an asymptotically free family of subgroups of $G_{\lambda_0} \subseteq \Iso(L, d)$.
    We define the \underline{asymptotic $(D, T)$-compactified algebra} as follows:
    \begin{align}
        \Cpt_{D,T}^\Lambda (\caA, \delta) := \prod_{\lambda \in \Lambda} \Cpt_{D,T}^{G_\lambda} (\caA, \delta) \Big/ \bigoplus_{\lambda \in \Lambda} \Cpt_{D,T}^{G_\lambda} (\caA, \delta).
    \end{align}
    For any $U \in \frF_D (L)$, there is a natural *-algebra map $j_U^\Lambda: \caA_U \to \Cpt_{D,T}^\Lambda (\caA, \delta)$.
\end{definition}

The next three corollaries follow straightforwardly from the corresponding lemma above for compactified algebra with a fixed group.
Note, however, that because we do not have stabilization with respect to $T$ \footref{ft:DTcptalgebra-increase-parameters-Tstabilization}, the asymptotic compactified algebra still depends on $T$.

\begin{corollary}\label{cor:asymptotic-DTcptalgebra-increase-d}
    Let $(\caA, \delta)$ be a $G_{\lambda_0}$-covariant locally presented quasi-local *-algebra with parameters $D, T$.
    Let $\delta' \ge \delta$.
    Let $\{G_\lambda\}_{\lambda \in \Lambda}$ be an asymptotically free family of subgroups of $G_{\lambda_0} \subseteq \Iso(L, d)$.
    Then, there is a canonical isomorphism $\Cpt_{D, T}^\Lambda (\caA, \delta) \cong \Cpt_{D, T}^\Lambda (\caA, \delta')$ compatible with the maps $j_U^\Lambda: \caA_U \to \Cpt_{D, T}^\Lambda (\caA, \delta), {j'}_U^\Lambda: \caA_U \to \Cpt_{D, T}^\Lambda (\caA, \delta')$ for any $U \in \frF_D (L)$.
    Thus, we denote by $\Cpt_{D, T}^\Lambda (\caA)$ a representative of the isomorphism class of $\Cpt_{D, T}^\Lambda (\caA, \delta)$ for any $\delta$ such that $\caA$ is locally presented with parameters $(D, T)$.
\end{corollary}
\begin{proof}
    By \autoref{lem:DTcptalgebra-increase-d}, for all but finitely many $G_\lambda$, we have a canonical isomorphism $\Cpt_{D, T}^{G_\lambda} (\caA, \delta) \cong \Cpt_{D, T}^{G_\lambda} (\caA, \delta')$.
    Thus, we have a canonical isomorphism $\Cpt_{D, T}^\Lambda (\caA, \delta) \cong \Cpt_{D, T}^\Lambda (\caA, \delta')$.
\end{proof}

\begin{corollary}\label{cor:asymptotic-DTcptalgebra-increase-parameters}
    Let $(\caA, \delta)$ be a $G_{\lambda_0}$-covariant locally presented quasi-local *-algebra.
    Let $(D, T), (D', T')$ be parameters such that $D' \ge D, T' \ge T$.
    Let $\{G_\lambda\}_{\lambda \in \Lambda}$ be an asymptotically free family of subgroups of $G_{\lambda_0} \subseteq \Iso(L, d)$.
    \begin{enumerate}
        \item There is a canonical surjective *-homomorphism $\Cpt_{D, T}^\Lambda (\caA) \twoheadrightarrow \Cpt_{D', T'}^\Lambda (\caA)$ compatible with the maps ${j}_U^\Lambda: \caA_U \to \Cpt_{D, T}^\Lambda (\caA, \delta), {j'}_U^\Lambda: \caA_U \to \Cpt_{D', T'}^\Lambda (\caA, \delta)$ for any $U \in \frF_D (L) \subseteq \frF_{D'} (L)$.
        \item If $T = T'$, then the map is an isomorphism. 
        Thus we denote by $\Cpt_T^\Lambda (\caA)$ a representative of the isomorphism class of $\Cpt_{D, T}^\Lambda (\caA)$ for any $D$ such that $\caA$ is locally presented with parameters $(D, T)$.
    \end{enumerate}
\end{corollary}
\begin{proof}\begin{enumerate}
    \item By \autoref{lem:DTcptalgebra-increase-parameters} (1), there is a canonical *-homomorphism $\Cpt_{D, T}^{G_\lambda} (\caA, \delta) \to \Cpt_{D', T'}^{G_\lambda} (\caA, \delta)$ for each $\lambda \in \Lambda$.
    Thus we have a *-homomorphism $\Cpt_{D, T}^\Lambda (\caA) \to \Cpt_{D', T'}^\Lambda (\caA)$.

    Let $t'$ be a parameter such that $\caA$ is locally presented with parameters $(D', T', t')$.
    Since $\Lambda_{\le D' + 2 t' + 2 \delta}$ is finite, by \autoref{lem:DTcptalgebra-increase-parameters} (3), the map $\Cpt_{D, T}^{G_\lambda} (\caA, \delta) \to \Cpt_{D', T'}^{G_\lambda} (\caA, \delta)$ is surjective for all but finitely many $\lambda \in \Lambda$.
    Thus, $\Cpt_{D, T}^\Lambda (\caA) \to \Cpt_{D', T'}^\Lambda (\caA)$ is surjective.

    \item If $T = T'$, then by \autoref{lem:DTcptalgebra-increase-parameters} (2), the map $\Cpt_{D, T}^{G_\lambda} (\caA) \to \Cpt_{D', T'}^{G_\lambda} (\caA)$ is injective for all $\lambda \in \Lambda$.
    Thus $\Cpt_{D, T}^\Lambda (\caA) \to \Cpt_{D', T'}^\Lambda (\caA)$ is injective.
\end{enumerate}\end{proof}

\begin{corollary}\label{cor:asymptotic-DTcptalgebra-boundedspread}
    Let $(\caA, \delta_A), (\caB, \delta_B)$ be $G_{\lambda_0}$-covariant locally presented quasi-local *-algebras with parameters $(D_A, T_A), (D_B, T_B)$.
    Let $(\alpha, s): \caA \to \caB$ be a $G_{\lambda_0}$-equivariant bounded-spread homomorphism.
    Suppose $T_A + 2s \le T_B$ and that there exists $D_A, D_B$ such that $D_A + 2 s \le D_B$ and $\caA$ is locally presented with parameters $(D_A, T_A)$ and $\caB$ is locally presented with parameters $(D_B, T_B)$.
    Let $\{G_\lambda\}_{\lambda \in \Lambda}$ be an asymptotically free family of subgroups of $G_{\lambda_0} \subseteq \Iso(L, d)$.
    Then, we have a canonical *-homomorphism $\Cpt^G (\alpha)$ such that
    \begin{equation}\begin{tikzcd}
        \caA_U \arrow[r, "\alpha"] \arrow[d, "{j}_U^\Lambda"] & \caB_{U^{+s}} \arrow[d, "{j'}_U^\Lambda"] \\
         \Cpt_{T_A}^\Lambda (\caA) \arrow[r, "\Cpt^\Lambda (\alpha)"] & \Cpt_{T_B}^\Lambda (\caB)
    \end{tikzcd}\end{equation}
    In particular, if $T_A' \ge T_A, T_B' \ge T_B$ such that $T_A' + 2 s \le T_B'$ and there exists $D_A', D_B'$ such that $D_A' + 2 s \le D_B'$, then
    \begin{equation}\begin{tikzcd}
        \Cpt_{T_A}^\Lambda (\caA) \arrow[r, "\Cpt^\Lambda (\alpha)"] \arrow[d, two heads] & \Cpt_{T_B}^\Lambda (\caB) \arrow[d, two heads] \\ 
        \Cpt_{T_A'}^\Lambda (\caA) \arrow[r, "\Cpt^\Lambda (\alpha)"] & \Cpt_{T_B'}^\Lambda (\caB)
    \end{tikzcd}\end{equation}
\end{corollary}
\begin{proof}
    It follows from \autoref{lem:DTcptalgebra-boundedspread}.
\end{proof}

Finally, we make it fully parameter-independent by taking a colimit with respect to $T$.

\begin{definition}\label{def:asymptotic-cptalgebra}
    Let $\caA$ be a $G_{\lambda_0}$-covariant locally presented quasi-local *-algebra.
    Let $\{G_\lambda\}_{\lambda \in \Lambda}$ be an asymptotically free family of subgroups of $G_{\lambda_0} \subseteq \Iso(L, d)$.
    We define the \underline{asymptotic compactified algebra} as follows:
    \begin{align}
        \Cpt^\Lambda (\caA) := \colim_T \Cpt_T^\Lambda (\caA).
    \end{align}
\end{definition}

We note that it is not possible to define a functor by taking a colimit with respect to $T$ first and taking an infinite sequence indexed by $\lambda$. 
In fact, for a fixed group action $G_{\lambda}$, $D, T$ need to be sufficiently small compared to the scale of the $G_{\lambda}$ action. 

As a limit of $j_U: \caA_U \to \Cpt^G_{D, T} (\caA, \delta)$ which models the local algebra embedding into Tube algebra (\autoref{lem:fusioncategory-end-embed-tube-Xk}), we have the following natural *-algebra map.
Again, it is in general not injective.
Yet, we expect it to be injective in practical cases.

\begin{lemma}\label{lem:asymptotic-cptalgebra-embedding}
    Let $\caA$ be a $G_{\lambda_0}$-covariant locally presented quasi-local *-algebra.
    Let $\{G_\lambda\}_{\lambda \in \Lambda}$ be an asymptotically free family of subgroups of $G_{\lambda_0} \subseteq \Iso(L, d)$.
    Then, we have a canonical *-homomorphism
    \begin{align}
        \caA \to \Cpt^\Lambda (\caA).
    \end{align}
\end{lemma}
\begin{proof}
    For every $U \in \frF (L)$, there is some $T_U$ such that $\caA$ is locally presented with $(\diam U, T_U)$ so that there is $\caA_U \to \Cpt_{T_U}^\Lambda (\caA) \to \Cpt^{\Lambda} (\caA)$.
    This map is independent of the choice of $T_U$.
    Since $\caA = \colim_{U \in \frF (L)} \caA_U$ and the family of maps $\caA_U \to \Cpt^{\Lambda} (\caA)$ is compatible with the inductive system, we have a canonical *-homomorphism $\caA \to \Cpt^\Lambda (\caA)$.
\end{proof}

Finally, we prove the functoriality of the compactification $\caA \mapsto \Cpt^\Lambda (\caA)$.
If we let the target category be $\stAlg$, we do not have a hope for it to be full or essentially surjective. 
The faithfulness follows if the natural map $j: \caA \to \Cpt^\Lambda (\caA)$ in \autoref{lem:asymptotic-cptalgebra-embedding} is indeed an embedding. 
We will see this in the fusion spin chain case in Section \ref{sec:fusionspinchain}. 

Thus, the general compactification partially remembers the quasi-local algebra (observables in the infinite volume limit), and in the fusion case, one can completely recover it. 

\begin{theorem}\label{thm:asymptotic-cptalgebra-functor}
    Let $\{G_\lambda\}_{\lambda \in \Lambda}$ be an asymptotically free family of subgroups of $G_{\lambda_0} \subseteq \Iso(L, d)$.
    Then, $\Cpt^\Lambda$ is a functor from a full subcategory of $\QLstAlg_{L, G_{\lambda_0}}$ consisting of locally presented quasi-local *-algebras to the category of *-algebras.
\end{theorem}
\begin{proof}
    Let $\alpha: \caA \to \caB$ be a $G_{\lambda_0}$-equivariant bounded-spread homomorphism with spread $s$. 
    For any sufficiently large $T_A$, $\caB$ is locally presented with $T_B \ge T_A + 2s$ and admits a corresponding $D_B$ such that $D_B \ge D_A + 2s$ where $\caA$ is locally presented with parameters $(D_A, T_A)$.
    Thus, there is a map $\Cpt_{T_A}^\Lambda (\caA) \to \Cpt_{T_B}^\Lambda (\caB)$ by \autoref{cor:asymptotic-DTcptalgebra-boundedspread}.
    Composing with $\Cpt_{T_B}^\Lambda (\caB) \to \Cpt^\Lambda (\caB)$, we have a map $\Cpt_{T_A}^\Lambda (\caA) \to \Cpt^\Lambda (\caB)$ for all sufficiently large $T_A$.
    Thus, there is a map $\Cpt^\Lambda (\alpha): \Cpt^\Lambda (\caA) \to \Cpt^\Lambda (\caB)$. 
    Since every construction is canonical, it is functorial.
\end{proof}

The next lemma shows that pre-DHR bimodules of locally presented quasi-local *-algebra can also be compactified.
It is unclear how to interpret the compactified bimodule in the general setting. 
We will see in the fusion spin chain case that this functor factors the equivalence $_\caA \DHR_\caA \cong \Rep (\Tube_\caC (X^{\otimes k}))$ (\autoref{thm:fusionspinchainalgebra-DTIcptmod}).
While the construction of the left module seems to be arbitrary, it has a natural interpretation in terms of the tube picture in the case of fusion spin chains. 

\begin{lemma}\label{lem:DTIcptmod}
    Let $(\caA, \delta)$ be a $G$-covariant locally presented quasi-local *-algebra over $L$ with $(R, \gamma)$ the constants of algebraic Haag duality and $(D, T), (D', T')$ the parameters of local presentation.
    Let $U \in \frB (L)$ be a ball such that $\diam (U) \ge R$ so that $\caA_{U^{-\delta}} \subseteq Z_{\caA_L} (\caA_{U^{c}}) \subseteq \caA_{U^{+\gamma}}$.

    If $D, D'$ satisfies $D' \ge 2 T + \diam (U) + 4 \delta + 4 \gamma$ and $G$ acts freely at scale $2 T + \delta + \diam(U)$, there is a functor from a full subcategory of $_\caA \pDHR_\caA$ localizable in $U$ to the category of pre-Hilbert $(\Cpt_{D, T}^G (\caA, \delta), \Cpt_{D', T'}^G (\caA, \delta))$-bimodules. 
    \begin{align}
        (_\caA \pDHR_\caA)_U \to _{\Cpt_{D, T}^G (\caA, \delta)} \pHMod_{\Cpt_{D', T'}^G (\caA, \delta)}
    \end{align}

    When $\Cpt_{D, T}^G (\caA, \delta)$ and $\Cpt_{D', T'}^G (\caA, \delta)$ are canonically isomorphic 
    \footnote{
        The condition $D' \ge 2 T + \diam (U) + 4 \delta + 4 \gamma$ forces $T' > T$. 
        Therefore, $\Cpt_{D, T}^G (\caA, \delta) \cong \Cpt_{D', T'}^G (\caA, \delta)$ does not happen as a consequence of \autoref{lem:DTcptalgebra-increase-parameters}.
        In practice, we expect that many examples of $\caA$ makes the map $\Cpt_{D, T}^G (\caA, \delta) \to \Cpt_{D', T'}^G (\caA, \delta)$ injective for $T' \ge T$.
    \label{ft:DTIcptmod-TTprime}}, this is a monoidal functor. 
    In particular, $_\caA \pDHR_\caA$ acts on $\Rep (\Cpt_{D', T'}^G (\caA, \delta))$ from right.
\end{lemma}
\begin{proof}
    Let $\caT = \Cpt_{D, T}^G (\caA, \delta)$ and $\caT' = \Cpt_{D', T'}^G (\caA, \delta)$ for brevity.
    For $M \in {}_\caA \pDHR_\caA$ localizable in $U$, let $M_U := \setbuilder{x \in M}{\text{for all} \; a \in \caA_{U^c}, ax = xa} = \sum_i b_i^U Z_{\caA_L} (\caA_{U^{c}})$, where $\{b^U_i\}$ is a right projective basis localized in $U$.
    Then, $M_U$ is a $(Z_{\caA_L} (\caA_{U^{c}}), Z_{\caA_L} (\caA_{U^{c}}))$ bimodule: for $x \in M_U, a \in Z_{\caA_L} (\caA_{U^{c}})$, $a x$ commutes with $\caA_{U^c}$. 
    Let $\braket*{\cdot}{\cdot}_{M_U}: M_U \times M_U \to \caA_L$ denote the form inherited from $M$. 
    Then, $\braket*{x}{y}_{M_U} \in Z_{\caA_L} (\caA_{U^c}) \subseteq \caA_{U^{+\gamma}}$.

    By assumption $D' \ge \diam (U^{+\gamma})$, there is an algebra map $Z_{\caA_L} (\caA_{U^{c}}) \hookrightarrow \caA_{U^{+\gamma}} \to \caT'$.
    Define a pre-Hilbert right $\caT'$-module
    \begin{align}
        \widetilde{M_U} := M_U \boxtimes_{Z_{\caA_L} (\caA_{U^{c}})} \caT'.
    \end{align}
    If $f: M \to N$ is an adjointable bimodule map, $f(M_U) \subseteq N_U$. 
    Thus, we have $\widetilde{f_U} := f_U \boxtimes \id_{\caT'}: \widetilde{M_U} \to \widetilde{N_U}$, which is adjointable right $\caT'$ module map.

    We introduce a left $\caT$-action on $\widetilde{M_U}$ and show $\widetilde{f_U}$ respects it. 
    It suffices to define a left $\caA_V$ action for all $V \in \frF_D (L)$ that is invariant under $\caI^{\loc}_{\delta, T} (V) + \caI^{\Gdist{G}}_{\delta, T} (V) + \caI^{\Gorb{G}}_{T} (V)$.
    For later convenience, we define a left $\caA_V$ action for all $V \in \frF_{2 T + \delta} (L)$. 
    Let $a \in \caA_V$ and $V \in \frF_{2 T + \delta} (L)$. 
    Choose $g \in G$ such that $d(g V, U) = \min_h d(h V, U)$.
    For $\sum_\mu m_\mu \boxtimes t_\mu \in \widetilde{M_U}$, 
    \begin{align}
        a \triangleright \paren*{\sum_\mu m_\mu \boxtimes t_\mu} 
        := \sum_{i, \mu} b_i \boxtimes j \paren*{\braket*{b_i}{\phi_g (a) m_\mu}_{M}} t_\mu,
    \end{align}
    where, $\braket*{b_i}{\phi_g (a) m_\mu}_{M} \subseteq Z_{\caA_L} (\caA_{(U \cup g V^{+\delta})^c}) \subseteq \caA_{(U \cup g V^{+\delta})^{+\gamma}} = \caA_{U^{+\gamma} \cup g V^{+\delta+\gamma}}$. 
    If $d(U^{+\gamma}, g V^{+\delta+\gamma}) > \delta$, by \autoref{rem:DTcptalgebra-canonicalj-Gdist}, there is a canonical map $j: \caA_{U^{+\gamma}} \otimes \caA_{g V^{+\delta+\gamma}} \to \caT'$ induced by $j_{U^{+\gamma}}: \caA_{U^{+\gamma}} \to \caT', j_{g V^{+\delta+\gamma}}: \caA_{g V^{+\delta+\gamma}} \to \caT'$.
    If $d(U^{+\gamma}, g V^{+\delta+\gamma}) \le \delta$, there is $j: \caA_{U^{+\gamma} \cup g V^{+\delta+\gamma}} \to \caT'$ since $D' \ge 2 T + \diam (U) + 4 \delta + 4 \gamma \ge \diam (U^{+\gamma} \cup g V^{+\delta+\gamma})$.

    \begin{enumerate}
        \item It is independent of the choice of $g \in G$. \\ 
        Suppose $g \neq h$ minimizes $d(g V, U) = d(h V, U)$. 
        Then, $d(g V, U) = d(h V, U) > 0$; otherwise $(g V \cup U) \cap (h g^{-1})(g V \cup U) = \emptyset$ but $G$ acts freely at scale $2 T + \delta + \diam(U) \ge \diam (g V \cup U)$.
        Then, 
        \begin{align}
            \sum_{i, \mu} b_i \boxtimes j \paren*{\braket*{b_i}{\phi_g (a) m_\mu}_{M}} t_\mu
            =& \sum_{i, \mu} b_i \boxtimes j \paren*{\braket*{b_i}{m_\mu}_M} j_{g V} (\phi_g (a)) t_\mu \\
            =& \sum_{i, \mu} b_i \boxtimes j \paren*{\braket*{b_i}{m_\mu}_M} j_{h V} (\phi_h (a)) t_\mu \\
            =& \sum_{i, \mu} b_i \boxtimes j \paren*{\braket*{b_i}{\phi_h (a) m_\mu}_{M}} t_\mu
        \end{align}

        \item It is independent of the choice of projective basis. \\ 
        Let $\{b_i^U\}, \{c_j^U\}$ be two projective bases localized in $U$.
        \begin{align}
            \sum_{i, \mu} b_i \boxtimes j \paren*{\braket*{b_i}{\phi_g (a) m_\mu}_{M}} t_\mu 
            =& \sum_{i, j, \mu} c_j \braket*{c_j}{b_i}_{M_U} \boxtimes j \paren*{\braket*{b_i}{\phi_g (a) m_\mu}_{M}} t_\mu \\ 
            =& \sum_{i, j, \mu} c_j \boxtimes j \paren*{\braket*{c_j}{b_i}_{M_U}} j \paren*{\braket*{b_i}{\phi_g (a) m_\mu}_{M}} t_\mu \\
            =& \sum_{j, \mu} c_j \boxtimes j \paren*{\braket*{c_j}{\phi_g (a) m_\mu}_{M}} t_\mu
        \end{align}
    \end{enumerate}

    We induce an action of $\bigast_{\frF_D (L)} \caA$ and check that $\caI^{\loc}_{\delta, T} (L) + \caI^{\Gdist{G}}_{\delta, T} (L) + \caI^{\Gorb{G}}_{T} (L)$ acts trivially.
    By construction, $\caI^{\Gorb{G}}_{T} (L)$ acts trivially. 
    \begin{enumerate}
        \item $\caI^{\Gdist{G}}_{\delta, T} (L)$ \\ 
        Let $a b - b a \in [\caA_{V_1}, \caA_{V_2}]$ for $V_1, V_2 \in \frF_T (L)$ such that $d(G V_1, G V_2) > \delta$.
        \begin{align}
            a \triangleright \paren*{b \triangleright \paren*{\sum_\mu m_{\mu} \boxtimes t_\mu}} 
            &= a \triangleright \paren*{\sum_{\mu, j} b_j \boxtimes j \paren*{\braket*{b_j}{\phi_h (b) m_{\mu}}_M} t_\mu} \\
            &= \sum_{\mu, j, k} b_k \boxtimes j \paren*{\braket*{b_k}{\phi_g (a) b_j}_M} j \paren*{\braket*{b_j}{\phi_h (b) m_{\mu}}_M} t_\mu \\
            &= \sum_{\mu, k} b_k \boxtimes j \paren*{\braket*{b_k}{\phi_g (a) \phi_h (b) m_{\mu}}_M} t_\mu \\
            &= \sum_{\mu, k} b_k \boxtimes j \paren*{\braket*{b_k}{\phi_h (b) \phi_g (a) m_{\mu}}_M} t_\mu \\
            &= b \triangleright \paren*{a \triangleright \paren*{\sum_\mu m_{\mu} \boxtimes t_\mu}} 
        \end{align}

        \item $\caI^{\loc}_{\delta, T} (L)$ \\ 
        Let $V_1, V_2 \in \frF_T (L)$ such that $d (V_1, V_2) \le \delta$. 
        Then, we have a well-defined action of $\caA_{V_1 \cup V_2}$ on $\widetilde{M_U}$ which factors $\caA_{V_1}, \caA_{V_2}$ actions. 
        Thus, any $\sum_i w_i \in \ker (\caA_{V_1} \ast \caA_{V_2} \to \caA_{V_1 \cup V_2})$ has trivial action.
    \end{enumerate}

    Finally, we show that $\widetilde{f_U}$ respects the left $\caT$-action. 
    \begin{align}
        \widetilde{f_U} \paren*{a \triangleright \paren*{\sum_\mu m_\mu \boxtimes t_\mu}} 
        =& \widetilde{f_U} \paren*{\sum_{\mu, j} b_j \boxtimes j \paren*{\braket*{b_j}{\phi_g (a) m_\mu}_M} t_{\mu}} \\
        =& \sum_{\mu, j} f_U (b_j) \boxtimes j \paren*{\braket*{b_j}{\phi_g(a) m_\mu}_M} t_{\mu} \\
        =& \sum_{\mu, j, k} c_k \braket*{c_k}{f_U (b_j)}_{N_U} \boxtimes j \paren*{\braket*{b_j}{\phi_g(a) m_\mu}_M} t_{\mu} \\
        =& \sum_{\mu, j, k} c_k \boxtimes j \paren*{\braket*{f^* (c_k)}{b_j}_M \braket*{b_j}{\phi_g(a) m_\mu}_M} t_{\mu} \\
        =& \sum_{\mu, j, k} c_k \boxtimes j \paren*{\braket*{c_k}{\phi_g(a) f_U(m_\mu)}_N} t_{\mu} \\
        =& a \triangleright \paren*{\sum_\mu f_U (m_\mu) \boxtimes t_{\mu}}.
    \end{align}

    Now, suppose $\caT \cong \caT'$.
    Let $\eta_{M, N}: \widetilde{(M \boxtimes N)_U} \to \widetilde{M_U} \boxtimes \widetilde{N_U}$ be defined by 
    \begin{align}
        \eta_{M, N} \paren*{b_i \boxtimes c_j \boxtimes 1} := (b_i \boxtimes 1) \boxtimes (c_j \boxtimes 1),
    \end{align}
    where $\{b_i\}, \{c_j\}$ are right projective bases of $M, N$ localized in $U$.
    Note that since $\{b_i \boxtimes c_j\}$ is a right projective basis of $M \boxtimes N$ localized in $U$, $(M \boxtimes N)_U = \sum_{i, j} (b_i \boxtimes c_j) Z_{\caA_L} (\caA_{U^{c}})$. 
    It is isometric on the projective basis $\{b_i \boxtimes c_j \boxtimes 1\}$ of $\widetilde{(M \boxtimes N)_U}$ and thus extends to a well-defined injective map by right $\caT$ linearity.
    Surjectivity follows from observing that any $(m \boxtimes t) \boxtimes (n \boxtimes s) \in \widetilde{M_U} \boxtimes \widetilde{N_U}$ is expressed as $(m \boxtimes 1) \boxtimes t \triangleright (n \boxtimes s) = (m \boxtimes 1) \boxtimes \paren*{\sum_\nu n_\nu \boxtimes s_\nu}$ for some $n_\nu \in N_U, s_\nu \in \caT'$.
    We show that this respects the left $\caT$ action. 
    \begin{align}
        a \triangleright \eta_{M, N} \paren*{(b_i \boxtimes c_j) \boxtimes t} 
        =& a \triangleright (b_i \boxtimes 1) \boxtimes (c_j \boxtimes t) \\ 
        =& \sum_{i'} \paren*{b_{i'} \boxtimes \braket*{b_{i'}}{\phi_{g}(a) b_i}_{M}} \boxtimes \paren*{c_j \boxtimes t} \\ 
        =& \sum_{i', j'} \paren*{b_{i'} \boxtimes 1} \boxtimes \paren*{c_{j'} \boxtimes \braket*{c_{j'}}{\braket*{b_{i'}}{\phi_{g}(a) b_i}_{M} c_j}_{N} t} \\ 
        =& \sum_{i', j'} \paren*{b_{i'} \boxtimes 1} \boxtimes \paren*{c_{j'} \boxtimes \braket*{b_{i'} \boxtimes c_{j'}}{\phi_{g}(a) b_i \boxtimes c_j}_{M \boxtimes N} t} \\ 
        =& \eta_{M, N} \paren*{\sum_{i', j'} b_{i'} \boxtimes c_{j'} \boxtimes \braket*{b_{i'} \boxtimes c_{j'}}{\phi_{g}(a) b_i \boxtimes c_j}_{M \boxtimes N} t} \\ 
        =& \eta_{M, N} \paren*{a \triangleright \paren*{b_i \boxtimes c_j \boxtimes t}}
    \end{align}
    $\eta_{M, N}$ does not depend on the choice of projective bases. 
    Let $\{b_k'\}, \{c_l'\}$ be another choice of projective bases localized in $U$.
    Then, 
    \begin{align}
        \eta_{M, N} \paren*{b_k' \boxtimes c_l' \boxtimes 1} 
        =& \eta_{M, N} \paren*{\sum_{i} b_i \braket*{b_i}{b_k'}_{M_U} \boxtimes c_l' \boxtimes 1} \\
        =& \eta_{M, N} \paren*{\sum_{i, j} b_i \boxtimes c_j \braket*{c_j}{\braket*{b_i}{b_k'}_{M_U} c_l'}_{N_U} \boxtimes 1} \\
        =& \eta_{M, N} \paren*{\sum_{i, j} b_i \boxtimes c_j \boxtimes \braket*{c_j}{\braket*{b_i}{b_k'}_{M_U} c_l'}_{N_U}} \\
        =& \sum_{i, j} b_i \boxtimes 1 \boxtimes c_j \boxtimes \braket*{c_j}{\braket*{b_i}{b_k'}_{M_U} c_l'}_{N_U} \\
        =& \sum_{i} b_i \boxtimes 1 \boxtimes \braket*{b_i}{b_k'}_{M_U} c_l' \boxtimes 1 \\
        =& b_k' \boxtimes 1 \boxtimes c_l' \boxtimes 1.
    \end{align}
\end{proof}

\section{Fusion Spin Chain}\label{sec:fusionspinchain}
As mentioned above, the most general compactification lacks properties that we want it to have; sufficiently small local algebras do not inject into the compactified algebras, and $\Cpt_{D, T}^G (\caA, \delta)$ depends on $T$.
In this section, we apply the compactification to fusion spin chain algebras and show that it reduces to the usual Tube algebra. 
The significance of viewing a Tube algebra in this way is \autoref{cor:fusionspinchainalgebra-DTcptalgebra-boundedspread}, \autoref{cor:fusionspinchainalgebra-asymptotic-cptalgebra}, and \autoref{lem:fusionspinchainalgebra-DHR-Tube-boundedspread}. 

\autoref{cor:fusionspinchainalgebra-DTcptalgebra-boundedspread} states that a bounded-spread homomorphism $\alpha$ induces a map between Tube algebras over the boundary of sufficiently large bounded regions. 
A good visualization of this induced map is the following.
\begin{align}
    \sum_{\color{red}{s} \normalcolor \in \Irr (\caC)}
    \tikzmath{
        \draw[thick] (0,0) -- (0,2);
        \draw[thick] (2,0) -- (2,2);
        \filldraw[thick, fill=white] (1,2) ellipse (1 and .15);
        \draw[thick, blue] (.4,-.13) -- (0.4,1.87);
        \draw[thick, blue] (.8,-.15) -- (0.8,1.85);
        \draw[thick, blue] (1.2,-.15) -- (1.2,1.85);
        \draw[thick, blue] (1.6,-.13) -- (1.6,1.87);
        \halfDottedEllipse{(0,0)}{1}{.15}
        \halfDottedEllipse[thick, red]{(0,1)}{1}{.15}
        \roundNboxEllipse[]{(1,1)}{1}{.15}{-120}{-60}{.5}{$\varphi$};
    }
    \;\mapsto\;
    \sum_{\color{red}{s} \normalcolor \in \Irr (\caC)}
    \tikzmath{
        \draw[thick] (0,0) -- (0,2);
        \draw[thick] (2,0) -- (2,2);
        \filldraw[thick, fill=white] (1,2) ellipse (1 and .15);
        \draw[thick, blue] (.4,-.13) -- (0.4,1.87);
        \draw[thick, blue] (.8,-.15) -- (0.8,1.85);
        \draw[thick, blue] (1.2,-.15) -- (1.2,1.85);
        \draw[thick, blue] (1.6,-.13) -- (1.6,1.87);
        \halfDottedEllipse{(0,0)}{1}{.15}
        \halfDottedEllipse[thick, red]{(0,1)}{1}{.15}
        \roundNboxEllipse[]{(1,1)}{1}{.15}{-150}{-30}{.5}{$\alpha(\varphi)$};
    }
\end{align}
Any element in a Tube algebra is a product of such local patches (\autoref{lem:fusionspinchainalgebra-tube-factorization}) so that as long as $+s$-neighborhood of local patches is still sufficiently small compared to $k$, it is well-defined. 
By defining this map from the general compactification applied to fusion spin chains, the independence of local decomposition follows automatically.

\autoref{cor:fusionspinchainalgebra-asymptotic-cptalgebra} says that the infinite sequence of Tube algebras $\prod_k \Tube_\caC (X^{\otimes k}) \Big/ \bigoplus_k \Tube_\caC (X^{\otimes k})$ recovers all information of the infinite volume limit, in contrast to \autoref{thm:asymptotic-cptalgebra-functor}, which achieved only partial recoveries. 
(More precisely, we will see that the compactification functor is faithful when restricted to quasi-local algebras bounded-spread isomorphic to fusion spin chains. 
If we define a bounded-spread homomorphism between two infinite sequences of Tube algebras to be a map induced by a bounded-spread homomorphism between quasi-local algebras, then quasi-local algebras are isomorphic if and only if the asymptotic Tube algebras are isomorphic.)

\autoref{lem:fusionspinchainalgebra-DHR-Tube-boundedspread} is a specialization of pre-DHR bimodule compactification. 
By relating $\DHR(\alpha)$, the braided autoequivalence between DHR categories, and $\Tube^k (\alpha^{-1})^*$, the pullback functor by the induced map between Tube algebras, we develop an obstruction theory in Section \ref{subsec:obstruction}.

\subsection{Characterization of Compactification in Fusion Spin Chain}\label{subsec:fusionspinchaincompactified}
First, we check that a fusion spin chain algebra is indeed (strongly) locally generated and (strongly) locally presented. 
\begin{lemma}[\cite{jones2024dhrbimodulesquasilocalalgebras}]\label{lem:fusionspinchainalgebra-locallygenerated}
    Let $\caA = \caA (\caC, X)$ be a fusion spin chain algebra over $\ZZ$ with $(X, n)$ a strong tensor generating object.
    Then, $\caA$ is locally generated with $D = n+1$ and $r = 0$.
\end{lemma}
\begin{proof}
    Take a finite set $I = I_1 \sqcup \cdots \sqcup I_l \in \frF (\ZZ)$, where each $I_i$ is a finite interval and $d(I_i, I_j) > 1$ for $i \neq j$.
    Then, $\caA_I = \bigotimes_i \caA_{I_i}$ by definition, and therefore it suffices to prove the claim for any finite interval $I \in \frB (\ZZ)$.

    By $\ZZ$-covariance, we may assume $I = \{1, ..., k\}$.
    If $\diam I = k \le n+1$, 
    \begin{align}
        \bigast_{\frF_{n+1} (I)} \caA \twoheadrightarrow \caA_I \ 
    \end{align}
    since $I \in \frF_{n+1} (I)$.

    Suppose $\diam (I) = k \ge n+1$.
    Let 
    \begin{align}
        I_1 = \{1, ..., n+1\},
        I_2 = \{2, ..., n+2\}, 
        ..., 
        I_{k-n} = \{k-n, ..., k\} \in \caF_{n+1} (I).
    \end{align}
    By repeatedly applying \autoref{cor:fusionspinchainalgebra-factorization},
    \begin{align}
        \caA_I \cong \caA_{I_1} \otimes_{\caA_{I_1 \cap I_2}} \caA_{I_2} \otimes_{\caA_{I_2 \cap I_3}} \cdots \otimes_{\caA_{I_{k-n-1} \cap I_{k-n}}} \caA_{I_{k-n}}.
    \end{align}
    Thus, any $a \in \caA_I$ can be written as a linear sum of elements of the form $a_1 \cdots a_{k-n}$, where $a_i \in \caA_{I_i}$, implying that $\bigast_i \caA_{I_i} \to \caA_I$ is surjective.
\end{proof}

\begin{lemma}\label{lem:fusionspinchainalgebra-locallypresented}
    Let $\caA = \caA (\caC, X)$ be a fusion spin chain algebra over $\ZZ$ with $(X, n)$ a strong tensor generating object.
    Then, $\caA$ is locally presented with any $D = T = 2n$, $r = t = 0$.
\end{lemma}
\begin{proof}
    The surjectivity of the map $\bigast_{\frF_D (I)} \caA \to \caA_I$ follows from \autoref{lem:fusionspinchainalgebra-locallygenerated} and \autoref{lem:quasilocal-staralgebra-increaseparameters} (\ref{it:quasilocal-staralgebra-locallygenerated-increaseparameters}).

    Take a finite set $I \in \frF (\ZZ)$. 
    We will prove that $\ker \left(\bigast_{\frF_D (I)} \caA \to \caA_I\right) \Big/ \left(\sum_{\frF_T (I)} \left(\ker (\iota * \iota) \right)^e_{\frF_D (I)} \right) = 0.$
    
    Take any $\sum_i w_i \in \ker \left(\bigast_{\frF_D (I)} \caA \to \caA_I\right)$ and let $w_i = a_i^{V_{i, 1}} \ast \cdots \ast a_i^{V_{i, l_i}}$ where $a_i^{V_{i, j}} \in \caA_{V_{i, j}}$ for some $V_{i, j} \in \frF_D (I)$.
    \begin{enumerate}
        \item 
        Then, by \autoref{lem:fusionspinchainalgebra-locallygenerated}, there exists $\sum_\rho b_{1, \rho} \ast \cdots \ast b_{k_\rho, \rho} \in \bigast_{\frB_{n+1} (V_{i, j})} \caA \mapsto a_i^{V_{i,j}} \in \caA_{V_{i, j}}$.
        Up to relations 
        \begin{align}
            \sum_{J \in \frB_{n+1} (V_{i,j})} \left(\ker \caA_{V_{i,j}} * \caA_{J} \to \caA_{V_{i, j}}\right)^{e}_{\frF_T (I)} \subseteq \sum_{\frF_T (I)} \left(\ker (\iota * \iota) \right)^e_{\frF_T (I)},
        \end{align}
        we assume without loss of generality that $w_i = a_i^{V_{i, 1}} \ast \cdots \ast a_i^{V_{i, l_i}}$ where each $V_{i, j} \in \frB_{n+1} (I)$ is an interval of size $n+1$.

        \item 
        Let $I = J_1 \sqcup \cdots \sqcup J_k$ where $d(J_\mu, J_\nu) > 1$ and $J_\mu$ is an interval. 
        Choose a cover $\{J_{\mu, \lambda} \in \frB_{2n} (J_\mu)\}_{\lambda = 1,...,q_\mu}$ of $J_\mu$ by subintervals of size at most $2 n$ such that $J_{\mu, \lambda} \cap J_{\mu, \kappa} \neq \emptyset \Rightarrow \# (J_{\mu, \lambda} \cap J_{\mu, \kappa}) \ge n$ and any interval $K \in \caB_{n+1} (J_{\mu})$ of size at most $n+1$ is contained in at least one of $\{J_{\mu, \lambda}\}$.
        \footnote{
            For example, if $J_\mu = \{1, ..., l_\mu\}$ and $l_\mu > 2 n$, one can take 
            \begin{gather*}
                J_{\mu, 1} = \{1, ..., 2 n \}, 
                J_{\mu, 2} = \{n + 1, ..., 3 n \},
                ..., \\
                J_{\mu, q_\mu - 1} = \{n (q_\mu - 2) + 1, ..., n q_\mu\},
                J_{\mu, q_\mu} = \{n (q_\mu - 1) + 1, ..., l_\mu\} \in \frF_{2n}(J_{\mu}),
            \end{gather*}
            where $n q_\mu < l_\mu \le n (q_\mu + 1)$. 

        }
        Any factor $a_i^{V_{i, j}}$ on $V_{i, j}$ of size $n+1$ lies in at least one of $J_{\mu, \lambda}$. 
        Up to relations
        \begin{align}
            &\sum_{J_{\mu, \lambda}} \sum_{\substack{K \in \frB_{n+1} (J_{\mu,\lambda})}} \left(\ker \caA_{J_{\mu,\lambda}} * \caA_{K} \to \caA_{J_{\mu,\lambda}}\right)^e_{\frF_T (I)} \\
            \subseteq& \sum_{\frF_T (I)} \left(\ker (\iota * \iota) \right)^e_{\frF_T (I)},
        \end{align}
        we may further assume that every factor is in $\{J_{\mu, \lambda}\}_{\mu, \lambda}$. 
        $w_i = a_i^{V_{i, 1}} \ast \cdots \ast a_i^{V_{i, l_i}}$ where each $V_{i, j} \in \{J_{\mu, \lambda}\}_{\mu, \lambda}$.

        \item
        For neighboring factors $a_i^{V_{i, j}} \ast a_i^{V_{i, j+1}}$ in $w_i$, if $\min {V_{i,j}} \ge \min {V_{i,j+1}}$, we apply the following exchange rule:
        \begin{enumerate}
            \item If $V_{i, j}, V_{i, j+1}$ are disjoint, 
            $\ker (\caA_{V_{i, j}} \ast \caA_{V_{i, j+1}} \to \caA_{V_{i, j} \cup V_{i, j+1}}) = [\caA_{V_{i, j}}, \caA_{V_{i, j+1}}]$ is in $\sum_{\frF_T (I)} \left(\ker (\iota * \iota) \right)^e_{\frF_T (I)}$, so we may exchange them.
            \item If $V_{i, j} \cap V_{i, j+1} \neq \emptyset$, by \autoref{cor:fusionspinchainalgebra-factorization}, $\caA_{V_{i, j}} \otimes_{\caA_{V_{i, j} \cap V_{i, j+1}}} \caA_{V_{i, j+1}} \cong \caA_{V_{i, j+1}} \otimes_{\caA_{V_{i, j} \cap V_{i, j+1}}} \caA_{V_{i, j}}$.
            Thus, inside $\caA_{V_{i,j} \cup V_{i,j+1}}$, $a_i^{V_{i, j}} \cdot a_i^{V_{i, j+1}} = \sum_\pi b^{V_{i, j+1}}_\pi \cdot b^{V_{i, j}}_\pi$ for some $b^{V_{i, j+1}}_\pi \in \caA_{V_{i, j+1}}, b^{V_{i, j}}_\pi \in \caA_{V_{i, j}}$.
            Up to relations 
            \begin{align}
                &(\ker (\caA_{V_{i,j}} \ast \caA_{V_{i,j+1}} \to \caA_{I}))^e_{\frF_T (I)} \\
                \subseteq& \sum_{\frF_T (I)} \left(\ker (\iota * \iota) \right)^e_{\frF_T (I)},
            \end{align}
            we replace $a_i^{V_{i, j}} \ast a_i^{V_{i, j+1}}$ with $\sum_\pi b^{V_{i, j+1}}_\pi \ast b^{V_{i, j}}_\pi$.
        \end{enumerate}
        Therefore, we may assume that each $w_i$ is of the form $w_i = a_i^{J_{1,1}} \ast \cdots \ast a_i^{J_{k, q_k}}$ where $a_i^{J_{\mu, \lambda}} \in \caA_{J_{\mu, \lambda}}$.
        In other words, we have $\sum_i a_i^{J_{1,1}} \ast \cdots \ast a_i^{J_{k, q_k}}$ in the submodule of $\bigast_{\frF_D (I)} \caA$ isomorphic to $\left(\caA_{J_{1, 1}} \otimes \cdots \otimes \caA_{J_{1, l_1}}\right) \otimes \cdots \otimes \left(\caA_{J_{k, 1}} \otimes \cdots \otimes \caA_{J_{k, q_k}}\right)$.

        \item 
        By assumption that $\sum_i w_i$ is in the kernel, 
        \begin{align}
            &\sum_i a_i^{J_{1,1}} \otimes \cdots \otimes a_i^{J_{k, q_k}} \in \left(\caA_{J_{1, 1}} \otimes \cdots \otimes \caA_{J_{1, l_1}}\right) \otimes \cdots \otimes \left(\caA_{J_{k, 1}} \otimes \cdots \otimes \caA_{J_{k, q_k}}\right)\\ 
            \mapsto& \sum_i a_i^{J_{1,1}} \cdots a_i^{J_{k, q_k}} = 0 \in \caA_I \\
            \cong& \left(\caA_{J_{1, 1}} \otimesunder{\caA_{J_{1, 1} \cap J_{1, 2}}} \cdots \otimesunder{\caA_{J_{1, l_1 - 1} \cap J_{1, l_1}}} \caA_{J_{1, l_1}}\right) \otimes \cdots \otimes \left(\caA_{J_{k, 1}} \otimesunder{\caA_{J_{k, 1} \cap J_{k, 2}}} \cdots \otimesunder{\caA_{J_{k, q_k - 1} \cap J_{k, q_k}}} \caA_{J_{k, q_k}}\right).
        \end{align}
        Thus, 
        \begin{align}
            &\sum_i a_i^{J_{1,1}} \otimes \cdots \otimes a_i^{J_{k, q_k}} \\
            \in& \spanv_\CC \setbuilder{\cdots \otimes (a x \otimes a' - a \otimes x a') \otimes \cdots}{a \in \caA_{J_{\mu, \lambda}}, b \in \caA_{J_{\mu, \lambda + 1}}, x \in \caA_{J_{\mu, \lambda} \cap J_{\mu, \lambda + 1}}
            }
        \end{align}
        As a submodule of $\bigast_{\frF_D (I)} \caA$, this is contained in 
        \begin{align}
            &\sum_{\mu=1}^k \sum_{\lambda = 1}^{q_\mu-1} \left(\ker (\caA_{J_{\mu, \lambda}} \ast \caA_{J_{\mu, \lambda + 1}} \to \caA_I) \right)^e_{\frF_T (I)} \\
            \subseteq& \sum_{\frF_T (I)} \left(\ker (\iota * \iota) \right)^e_{\frF_T (I)}.
        \end{align}
    \end{enumerate}
\end{proof}

We obtain a $(D, T)$-compactified algebra as follows. 
Observe that under conditions on $k$, the parameter dependence is resolved up to isomorphism, and this is the first of two important simplifications in the fusion case. 

\begin{theorem}\label{thm:fusionspinchainalgebra-DTcptalgebra}
    Let $\caA = \caA (\caC, X)$ be a fusion spin chain algebra over $\ZZ$ with $(X, n)$ a strong tensor generating object.
    Let $D, T \ge 3 n$ be parameters of a local presentation with $r = 0, t = 0$.
    Then, when $k > 2T + \delta + 2$, we have a canonical isomorphism
    \begin{align}
        \Cpt_{D, T}^{k \ZZ} (\caA, \delta) \cong \Tube_\caC (X^{\otimes k}).
    \end{align} 
    If $D' \ge D, T' \ge T$, it factors through $\Cpt_{D, T}^{k \ZZ} (\caA, \delta) \to \Cpt_{D', T'}^{k \ZZ} (\caA, \delta)$.
\end{theorem}
\begin{proof}
    We will first show that if $S \ge 3n$, $k > 2S + 1$, $\Cpt_{S, S}^{k \ZZ} (\caA, 1) \cong \Tube_\caC (X^{\otimes k})$. 
    Once we prove it, assuming $k > 2 T + \delta + 2$, $\Tube_\caC (X^{\otimes k}) \cong \Cpt_{T, T}^{k \ZZ} (\caA, 1) \cong \Cpt_{T, T}^{k \ZZ} (\caA, \delta) \cong \Cpt_{D, T}^{k \ZZ} (\caA, \delta)$ by \autoref{lem:DTcptalgebra-increase-d} and \autoref{lem:DTcptalgebra-increase-parameters}.

    Let $[\cdot]: \ZZ \to \ZZ / k \ZZ$ be the quotient map.
    Then, by \autoref{lem:fusioncategory-end-embed-tube-Xk}, there is a canonical up to isomorphism embedding $\caA_V \hookrightarrow \Tube_\caC (X^{\otimes})$ for each $V \in \frF_S (\ZZ)$.
    We obtain a *-homomorphism $\bigast_{\frF_S (\ZZ)} \caA \to \Tube_\caC (X^{\otimes k})$.
    By \autoref{lem:fusionspinchainalgebra-tube-factorization}, this is surjective.

    We will show that 
    \begin{align}
        \ker \left(\bigast_{\frF_S (\ZZ)} \caA \to \Tube_\caC (X^{\otimes k})\right) = \left(\caI^{\loc}_{1, S} (\ZZ) + \caI^{\Gdist{k \ZZ}}_{1, S} (\ZZ) + \caI^{\Gorb{k \ZZ}}_S (\ZZ)\right)^c_{\frF_S (\ZZ)}.
    \end{align}
    Observe that $\caI^{\Gorb{k \ZZ}}_S (\ZZ), \caI^{\Gdist{k \ZZ}}_{1, S} (\ZZ)  \subseteq \ker \left(\bigast_{\frF_S (\ZZ)} \caA \to \Tube_\caC (X^{\otimes k})\right)$ by construction.
    Take any $U, V \in \frF_S (\ZZ)$ with $d(U, V) \le 1$.
    Then, $W = U \cup V \in \frF_{2 S + 1} (\ZZ)$, so that there is a canonical embedding $\caA_W \hookrightarrow \Tube_\caC (X^{\otimes k})$ by \autoref{lem:fusioncategory-end-embed-tube-Xk}, which factors $\caA_U \hookrightarrow \Tube_\caC (X^{\otimes k})$ and $\caA_V \hookrightarrow \Tube_\caC (X^{\otimes k})$.
    Thus, any $\sum_i w_i \in \ker (\caA_U \ast \caA_V \to \caA_{\ZZ}) \subseteq \caI^{\loc}_{1, S} (\ZZ)$ is mapped to $0 \in \Tube_\caC (X^{\otimes k})$. 

    Now, we show 
    \begin{align}
        \ker \left(\bigast_{\frF_S (\ZZ)} \caA \to \Tube_\caC (X^{\otimes k})\right) \Big/ \left(\caI^{\loc}_{1, S} (\ZZ) + \caI^{\Gdist{k \ZZ}}_{1, S} (\ZZ) + \caI^{\Gorb{k \ZZ}}_S (\ZZ)\right)^c_{\frF_S (\ZZ)} = 0.
    \end{align}
    \begin{enumerate}
        \item
        Take any $\sum_i w_i \in \ker \left(\bigast_{\frF_S (\ZZ)} \caA \to \Tube_\caC (X^{\otimes k})\right)$ and let $w_i = a_i^{V_{i, 1}} \ast \cdots \ast a_i^{V_{i, l_i}}$ where $a_i^{V_{i, j}} \in \caA_{V_{i, j}}$ for some $V_{i, j} \in \frF_S (\ZZ)$.
        Then, by \autoref{lem:fusionspinchainalgebra-locallygenerated}, there exists $\sum_\rho b_{1, \rho} \ast \cdots \ast b_{k_\rho, \rho} \in \bigast_{\frB_{n+1} (V_{i, j})} \caA \mapsto a_i^{V_{i,j}} \in \caA_{V_{i, j}}$.
        Up to relations 
        \begin{align}
            \sum_{J \in \frB_{n+1} (V_{i,j})} \left(\ker \caA_{V_{i,j}} * \caA_{J} \to \caA_{V_{i, j}}\right)^{e}_{\frF_S (\ZZ)} \subseteq \caI^{\loc}_{1, S} (\ZZ),
        \end{align}
        we assume without loss of generality that $w_i = a_i^{V_{i, 1}} \ast \cdots \ast a_i^{V_{i, l_i}}$ where each $V_{i, j} \in \frB_{n+1} (\ZZ)$ is an interval of size $n+1$.

        \item 
        Any $V_{i,j} \in \frB_{n+1} (\ZZ)$ can be translated by some $n k \in k \ZZ$ so that $(V_{i, j} + n k) \subseteq J := \{1, ..., k, ..., k+n\}$. 
        Up to relations $\caI^{\Gorb{k \ZZ}}_{n+1} (\ZZ) \subseteq \caI^{\Gorb{k \ZZ}}_S (\ZZ)$, we replace $a_i^{V_{i, j}}$ with $\phi_{+n k} (a_i^{V_{i, j}})$.
        Thus, we may assume that $V_{i,j} \in \setbuilder{\{i+1,...,i+n+1\}}{0 \le i \le k-1}$.

        \item 
        Choose a cover $\{J_{\lambda} \in \frB_{3n} (J)\}_{\lambda = 1,...,q}$ of $J$ by subintervals of size at most $3 n$ such that $[J_{\lambda}] \subseteq \ZZ / k \ZZ$ is contractible, $[J_{\lambda}] \cap [J_{\kappa}] \neq \emptyset \Rightarrow$ $([J_{\lambda}] \cap [J_{\kappa}])$ is contractible and $\# ([J_{\lambda}] \cap [J_{\kappa}]) \ge n$, and any contractible $[K] \subseteq \ZZ / k \ZZ$ of size at most $n+1$ is contained in at least one of $[J_\lambda]$.
        \footnote{
            For example, since $k \ge 3 n$,
            \begin{gather*}
                J_{1} = \{1,...,2n\}, 
                J_{2} = \{n+1,...,3n\}, 
                ..., \\
                J_{q - 1} = \{(q-2)n+1,...,qn\}, 
                J_{q} = \{n(q-1)+1,...,k,k+1,...,k+n\} \in \frF_{3n} (J)
            \end{gather*}
            where $n q \le k < n (q+1)$.
            Note that we assume $k > 2 S \ge 6 n$, so this cover satisfies the conditions.
        }
        Any factor $a_i^{V_{i, j}}$ on $V_{i,j} \in \setbuilder{\{i+1,...,i+n+1\}}{0 \le i \le k-1}$ lies in at least one of $J_{\lambda}$. 
        Up to relations
        \begin{align}
            &\sum_{J_{\lambda}} \sum_{\substack{K \in \frB_{n+1} (J_{\lambda})}} \left(\ker \caA_{J_{\lambda}} * \caA_{K} \to \caA_{J_{\lambda}}\right)^e_{\frF_S (\ZZ)} 
            \subseteq \caI^{\loc}_{1, S} (\ZZ),
        \end{align}
        we may further assume that every factor is in $\{J_{\lambda}\}_{\lambda}$. 
        $w_i = a_i^{V_{i, 1}} \ast \cdots \ast a_i^{V_{i, l_i}}$ where each $V_{i, j} \in \{J_{\lambda}\}_{\lambda}$.
        
        \item 
        For neighboring factors $a_i^{V_{i, j}} \ast a_i^{V_{i, j+1}}$ in $w_i$, if $\min {V_{i,j}} \ge \min {V_{i,j+1}}$, we apply the following exchange rule:
        \begin{enumerate}
            \item If $V_{i, j}, V_{i, j+1}$ are $k \ZZ$-disjoint, $\ker (\caA_{V_{i, j}} \ast \caA_{V_{i, j+1}} \to \caA_{V_{i, j} \cup V_{i, j+1}}) = [\caA_{V_{i, j}}, \caA_{V_{i, j+1}}]$ is in $\caI^{\Gdist{k\ZZ}}$, so we may exchange them.

            \item If $[V_{i, j}] \cap [V_{i, j+1}] \neq \emptyset$, let $+ m k \in k \ZZ$ such that $(V_{i,j} + m k) \cap V_{i,j+1} \neq \emptyset$. 
            By \autoref{cor:fusionspinchainalgebra-factorization}, $\caA_{(V_{i, j} + m k)} \otimes_{\caA_{(V_{i, j} + m k) \cap V_{i, j+1}}} \caA_{V_{i, j+1}} \cong \caA_{V_{i, j+1}} \otimes_{\caA_{(V_{i, j} + m k) \cap V_{i, j+1}}} \caA_{(V_{i, j} + m k)}$.
            Thus, inside $\caA_{(V_{i,j} + m k) \cup V_{i,j+1}}$, $\phi_{+mk} (a_i^{V_{i, j}}) \cdot a_i^{V_{i, j+1}} = \sum_\pi b^{V_{i, j+1}}_\pi \cdot \phi_{+mk}(b^{V_{i, j}}_\pi)$ for some $b^{V_{i, j+1}}_\pi \in \caA_{V_{i, j+1}}, \phi_{+mk}(b^{V_{i, j}}_\pi) \in \caA_{(V_{i, j}+mk)}$.
            Up to relations 
            \begin{align}
                &(\ker (\caA_{(V_{i,j}+mk)} \ast \caA_{V_{i,j+1}} \to \caA_{L}))^e_{\frF_S (\ZZ)} + \caI^{\Gorb{k\ZZ}}_{S} (\ZZ) \subseteq \caI^{\loc}_{1, S} (\ZZ) + \caI^{\Gorb{k\ZZ}}_{S} (\ZZ)
            \end{align}
            we replace $a_i^{V_{i, j}} \ast a_i^{V_{i, j+1}}$ with $\sum_\pi b^{V_{i, j+1}}_\pi \ast b^{V_{i, j}}_\pi$.
        \end{enumerate}
        Therefore, we may assume that each $w_i$ is of the form $w_i = a_i^{J_{1}} \ast \cdots \ast a_i^{J_{q}}$ where $a_i^{J_{\lambda}} \in \caA_{J_{\lambda}}$.
        In other words, we have $\sum_i a_i^{J_{1}} \ast \cdots \ast a_i^{J_{q}}$ in the submodule of $\bigast_{\frF_S (\ZZ)} \caA$ isomorphic to $\caA_{J_1} \otimes \cdots \otimes \caA_{J_q}$.

        \item 
        By assumption that $\sum_i w_i$ is in the kernel, 
        \begin{align}
            &\sum_i a_i^{J_1} \otimes \cdots \otimes a_i^{J_q} \in \caA_{J_{1}} \otimes \cdots \otimes \caA_{J_{q}}\\ 
            \mapsto& \sum_i a_i^{J_1} \cdots a_i^{J_q} = 0 \in \Tube_{\caC} (X^{\otimes k}) \\
            \cong& \left(\caA_{J_1} \otimesunder{\caA_{J_1 \cap J_2}} \cdots \otimesunder{\caA_{J_{q-2} \cap J_{q-1}}} \caA_{J_{q-1}}\right) \otimesunder{\caA_{J_{q-1} \cap J_q} \otimes \caA_{J_1 \cap (J_q - k)}} \caA_{J_q},
        \end{align}
        where we use \autoref{lem:fusionspinchainalgebra-tube-factorization}.
        Thus, 
        \begin{align}
            &\sum_i a_i^{J_{1}} \otimes \cdots \otimes a_i^{J_{q}} \\
            \in& \spanv_\CC \setbuilder{\cdots \otimes (a x \otimes a' - a \otimes x a') \otimes \cdots}{a \in \caA_{J_{\lambda}}, b \in \caA_{J_{\lambda+1}}, x \in \caA_{J_{\lambda} \cap J_{\lambda + 1}}} \\ 
            +& \spanv_\CC \setbuilder{x a \otimes \cdots \otimes a' - a \otimes \cdots \otimes a' \phi_{+k}(x)}{a \in \caA_{J_1}, a' \in \caA_{J_q}, x \in \caA_{J_1 \cap (J_q - k)}} 
        \end{align}
        As a submodule of $\bigast_{\frF_D (I)} \caA$, this is contained in 
        \begin{align}
            &\sum_{\lambda = 1}^{q-1} \left(\ker (\caA_{J_{\lambda}} \ast \caA_{J_{\lambda + 1}} \to \caA_\ZZ) \right)^e_{\frF_S (\ZZ)} + \left(\ker (\caA_{J_{1}} \ast \caA_{(J_{q} - k)} \to \caA_\ZZ) \right)^e_{\frF_S (\ZZ)} + \caI^{\Gorb{k \ZZ}}_{S} (\ZZ) \\
            \subseteq& \caI^{\loc}_{0, S} (\ZZ) + \caI^{\Gorb{k \ZZ}}_{S} (\ZZ).
        \end{align}
    \end{enumerate}
\end{proof}

An important result that takes advantage of this stabilization is the following.
Although $\Tube^k$ is still not functorial, we obtain an isomorphism at the level of fixed group action $k \ZZ$ under assumptions, unlike in \autoref{lem:DTcptalgebra-boundedspread}.

\begin{corollary}\label{cor:fusionspinchainalgebra-DTcptalgebra-boundedspread}
    Let $\caA = \caA (\caC, X), \caB = \caA (\caD, Y)$ be fusion spin chain algebras over $\ZZ$ with $(X, n), (Y, m)$ strong tensor generating objects, respectively.
    Let $\alpha: \caA \to \caB$ be a $\ZZ$-equivariant bounded-spread homomorphism with spread $s$.
    For any $k > \max \{6 n + 4 s, 6 m\} + 3$, we have
    \begin{align}
        \Tube^k (\alpha): \Tube_\caC (X^{\otimes k}) \to \Tube_\caD (Y^{\otimes k})
    \end{align}
    If $\alpha$ is a $\ZZ$-equivariant bounded-spread isomorphism and $k > \max \{6 n + 12 s, 6 m + 8 s\} + 3$, $\Tube^k (\alpha)$ is an isomorphism.
\end{corollary}
\begin{proof}
    Assuming $k > \max \{6 n + 4 s, 6 m\} + 3$, there exists $D_A, T_A, \delta_A, D_B, T_B, \delta_B$ such that $D_A, T_A \ge 3n, D_B, T_B \ge 3m$; $k > 2 T_A + \delta_A + 2, k > 2 T_B + \delta_B + 2$; and $D_B \ge D_A + 2 s, T_B \ge T_B + 2 s$. 
    \footnote{
        For example, take $\delta_A = \delta_B = 1$, $D_A = T_A = 3 n$, $D_B = T_B = \max\{3 n + 2 s, 3 m\}$.
    }
    By \autoref{thm:fusionspinchainalgebra-DTcptalgebra} and \autoref{lem:DTcptalgebra-boundedspread}, we have 
    \begin{equation}\begin{tikzcd}
        \Tube_\caC (X^{\otimes k}) & \Cpt_{D_A, T_A}^{k \ZZ} (\caA, \delta_A) \arrow[l, "\cong"] \arrow[r, "\Cpt^{k\ZZ} (\alpha)"] & \Cpt_{D_B, T_B}^{k \ZZ} (\caB, \delta_B) \arrow[r, "\cong"] & \Tube_\caD (Y^{\otimes k}) 
    \end{tikzcd}\end{equation}
    which is independent of the choice of parameters.

    If $\alpha$ is a $\ZZ$-equivariant bounded-spread isomorphism and $k > \max \{6 n + 12 s, 6 m + 8 s\} + 3$, there exists $D_A$, $T_A$, $D'_A$, $T'_A$, $\delta_A$, $D_B$, $T_B$, $D'_B$, $T'_B$, $\delta_B$ such that $D_A, T_A, D'_A, T'_A \ge 3 n$, $D_B, T_B, D'_B, T'_B \ge 3 m$; $k > 2 T_A + \delta_A + 2, k > 2 T'_A + \delta_A + 2$, $k > 2 T_B + \delta_B + 2, k > 2 T'_B + \delta_B + 2$; and $D_B \ge D_A + 2 s, T_B \ge T_B + 2 s, D'_A \ge D_B + 2s, T'_A \ge T_B + 2 s, D'_B \ge D'_A + 2s, T'_B \ge T'_A + 2 s$.
    \footnote{ 
        For example, take $\delta_A = \delta_B = 1$, $D_A = T_A = 3 n$, $D_B = T_B = \max\{3 n + 2 s, 3 m\}$, $D_A' = T_A' = \max\{3 n + 4 s, 3 m + 2 s\}$, $D_B' = T_B' = \max\{3 n + 6 s, 3 m + 4 s\}$.
    }
    By \autoref{thm:fusionspinchainalgebra-DTcptalgebra} and \autoref{lem:DTcptalgebra-boundedspread}, we have
    \begin{equation}\begin{tikzcd}[row sep=large]
        \Tube_\caC (X^{\otimes k}) & \Cpt_{D_A, T_A}^{k \ZZ} (\caA, \delta_A) \arrow[l, "\cong"] \arrow[r, "\Cpt^{k\ZZ} (\alpha)"] & \Cpt_{D_B, T_B}^{k \ZZ} (\caB, \delta_B) \arrow[ld, "\Cpt^{k \ZZ}(\alpha^{-1})"] \arrow[r, "\cong"] & \Tube_\caD (Y^{\otimes k}) \\
        & \Cpt_{D'_A, T'_A}^{k \ZZ} (\caA, \delta_A) \arrow[lu, "\cong"] \arrow[r, "\Cpt^{k\ZZ} (\alpha)"] & \Cpt_{D'_B, T'_B}^{k \ZZ} \arrow[ru, "\cong"]
    \end{tikzcd}\end{equation}
\end{proof}

Meanwhile, the asymptotic compactified algebra has a rather simple description. 
An interpretation of this result is that $\Cpt^\NN$ preserves the information of the infinite volume limit.
Indeed, if we call $\prod_k \Tube^k (\alpha)$ a bounded-spread homomorphism between asymptotic Tube algebras, two quasi-local *-algebras are bounded-spread isomorphic if and only if their asymptotic Tube algebras are bounded-spread isomorphic. 
There is no convenient alternative description of $\{\prod_k \Tube^k (\alpha) \mid \alpha: \caA \to \caA \; \text{wtih bounded-spread}\}$.

Mathematically, it is to say $\Cpt^\NN$ is faithful if we restrict to $\QLstAlg_{\ZZ, \ZZ}^{\mathrm{fusion}} \subseteq \QLstAlg_{\ZZ, \ZZ}$, a full subcategory consisting of quasi-local algebras bounded-spread isomorphic to a fusion spin chain algebra.

\begin{corollary}\label{cor:fusionspinchainalgebra-asymptotic-cptalgebra}
    Let $\caA = \caA (\caC, X)$ be a fusion spin chain algebra over $\ZZ$ with $(X, n)$ a strong tensor generating object.
    $\{k \ZZ\}_{k \in \NN}$ is an asymptotically free family of subgroups of $\ZZ \subseteq \Iso(\ZZ, d)$.
    Then, for any $T \ge 3n$,
    \begin{align}
        \Cpt_{T}^\NN (\caA) \cong \Cpt^\NN (\caA) \cong \prod_k \Tube_\caC (X^{\otimes k}) \Big/ \bigoplus_k \Tube_\caC (X^{\otimes k}).
    \end{align}
    There is a natural embedding $\caA \hookrightarrow \prod_k \Tube_\caC (X^{\otimes k}) \Big/ \bigoplus_k \Tube_\caC (X^{\otimes k})$. 
    In particular, $\Cpt^\NN: \QLstAlg(\caA , \caA) \to \stAlg(\Cpt^\NN (\caA), \Cpt^\NN (\caA))$ is injective.
\end{corollary}

Next, we apply \autoref{lem:DTIcptmod} the compactification of a pre-DHR bimodule to the fusion case. 
(Here, the categories of pre-DHR bimodules and DHR bimodules are canonically equivalent.)
Recall that when $\caT \cong \caT'$ via the canonical map, the functor in \autoref{lem:DTIcptmod} is monoidal, which is the case for fusion spin chain.
We also see that there is a simpler description through equivalence $\caZ (\caC) \to _\caA \DHR_\caA$ with a graphical description.

\begin{theorem}\label{thm:fusionspinchainalgebra-DTIcptmod}
    Let $\caA = \caA(\caC, X)$ be a fusion spin chain algebra with $\caC$ a unitary fusion category and $X$ a strong tensor generator with $n \ge 1$ and $k > 14 n + 3$.
    Then, there is a monoidal functor $_\caA \DHR_\caA \to _{\Tube_\caC (X^{\otimes k})} \pHMod_{\Tube_\caC (X^{\otimes k})}$.

    Furthermore, the following commutes up to natural isomorphism
    \begin{equation}\begin{tikzcd}
        \caZ (\caC) \arrow[r] \arrow[bend right=20, swap]{rr}{\caZ (\caC) (I(X^{\otimes k}) \to I(X^{\otimes k}) \otimes -)}  & _\caA \DHR_\caA \arrow[r] & _{\Tube_\caC (X^{\otimes k})} \pHMod_{\Tube_\caC (X^{\otimes k})}
    \end{tikzcd}\end{equation}
    where $\caZ (\caC) \to _\caA \DHR_\caA$ is the equivalence in \cite{jones2024dhrbimodulesquasilocalalgebras}. 
    \begin{equation}
        \tikzmath{
            \draw[thick, knot, orange, decorate, decoration={snake, segment length=1mm, amplitude=.2mm}] (-.5,0) node[right, xshift=.2cm]{$(Z, \sigma)$} to[out=60,in=270] (0.5,1);
        }  
        \; \mapsto \;
        \braces*{\tikzmath{
            \node at (-1.4,0) {$\cdots$};
            \draw[very thick, blue] (-1,-1) -- (-1,1);
            \draw[very thick, blue] (-.6,-1) -- (-.6,1);
            \draw[very thick, blue] (-.2,-1) -- (-.2,1);
            \draw[very thick, blue] (.2,-1) -- (.2,1);
            \draw[very thick, blue] (.6,-1) -- (.6,1);
            \draw[very thick, blue] (1,-1) -- (,1);
            \node at (1.4,0) {$\cdots$};
            \draw[thick, knot, orange, decorate, decoration={snake, segment length=1mm, amplitude=.2mm}] (0.6,0) to[out=60,in=270] (0.9,1);
            \roundNbox{fill=white}{(0,0)}{.3}{.9}{.9}{}
        }}
        \; \mapsto \;
        \braces*{
        \sum_{\color{red}{s} \normalcolor \in \Irr (\caC)}
        \tikzmath{
            \draw[thick] (0,0) -- (0,2);
            \draw[thick] (2,0) -- (2,2);
            \filldraw[thick, fill=white] (1,2) ellipse (1 and .15);
            \draw[very thick, blue] (.4,-.13) -- (0.4,1.87);
            \draw[very thick, blue] (.8,-.15) -- (0.8,1.85);
            \draw[very thick, blue] (1.2,-.15) -- (1.2,1.85);
            \draw[very thick, blue] (1.6,-.13) -- (1.6,1.87);
            \halfDottedEllipse{(0,0)}{1}{.15};
            \halfDottedEllipse[red]{(0,1)}{1}{.15};
            \draw[thick, knot, orange, decorate, decoration={snake, segment length=1mm, amplitude=.2mm}] (1.8,0.8) node[right, xshift=.6cm, yshift=.1cm]{$(Z, \sigma)$} to[out=0,in=270] (3,2);
            \roundNboxEllipse[]{(1,1)}{1}{.15}{-140}{-40}{.5}{};
        }}
    \end{equation}
\end{theorem}
\begin{proof}
    Let 
    \begin{align}
        \tikzmath{
            \draw[very thick, blue] (-.5,0) -- (.5,0);
        } =& \id_{X^{\otimes n}} \\
        \tikzmath{
            \draw[thick, red] (-.5,0) -- (.5,0);
        } =& \id_{\color{red}{s}} \quad \color{red}{s} \normalcolor \in \Irr(\caC) \\
        \tikzmath{
            \draw[thick, knot, orange, decorate, decoration={snake, segment length=1mm, amplitude=.2mm}] (-.5,0) -- (.5,0);
        } =& \id_{\color{orange}{Z}} .
    \end{align}
    Take $I \in \frB (\ZZ)$ such that $\diam(I) \ge n$ and $\delta, D, T, D', T' \ge 3 n$ such that $k > 2 T + \delta + 2, k > 2 T' + \delta + 2$; $D' \ge 2 T + \diam (I), k > 2 T + \diam (I) + \delta$. 
    \footnote{
        For example, take $\diam (I) = n$, $\delta = 1$, $D = T = 3 n$, $D' = T' = 6 n + \diam (I)$.
    }
    Then, by \autoref{thm:fusionspinchainalgebra-DTcptalgebra} and \autoref{lem:DTIcptmod}, we have a monoidal functor 
    \begin{align}
        (_\caA \DHR_\caA)_I \to _{\Tube_\caC (X^{\otimes k})} \pHMod_{\Tube_\caC (X^{\otimes k})},
    \end{align}
    which is independent of the choice of parameters.
    By \cite{jones2024dhrbimodulesquasilocalalgebras}, any DHR bimodule of a fusion spin chain algebra is localizable in any interval of size at least $n$. 
    Thus, $(_\caA \DHR_\caA)_I = _\caA \DHR_\caA$.

    Recall that $M(\widetilde{Z}) \cong \colim_l \caC (X^{\otimes l} \otimes X^{\otimes (\#I)} \otimes X^{\otimes l} \to X^{\otimes l} \otimes X^{\otimes (\#I)} \otimes Z \otimes X^{\otimes l})$ with right $\caA_{I^{+l}}$ action by pre-composition and left $\caA_{I^{+l}}$ action by post-composing $f \in \caA_{I^{+l}}$ conjugated by $\sigma_{Z, X^{\otimes l}}$. 
    We have $M (Z, \sigma)_I \cong \sum_i b_i^I \caA_I$ for a projective basis localized in $I$.
    By explicitly choosing $b_i^I$ by semisimplicity, we obtain $M(Z, \sigma)_I \cong \caC (X^{\otimes (\#I)} \to X^{\otimes (\#I)} \otimes Z)$ as $(\caA_I, \caA_I)$-bimodules.
    \begin{align}
        \widetilde{M(Z, \sigma)_I} 
        \cong& \caC (X^{\otimes (\#I)} \to X^{\otimes (\#I)} \otimes Z) \otimes_{\caA_I} {\Tube_\caC (X^{\otimes k})} \\ 
        \cong& \bigoplus_{s \in \Irr (\caC)} \caC (X^{\otimes (\#I)} \to X^{\otimes (\#I)} \otimes Z) \otimes_{\caA_I} \caC (X^{\otimes k} \to s \otimes X^{\otimes k} \otimes s^{\vee}) \\ 
        \cong& \bigoplus_{s \in \Irr (\caC)} \caC (X^{\otimes k} \to s \otimes X^{\otimes a} \otimes X^{\otimes (\#I)} \otimes Z \otimes X^{\otimes b} \otimes s^{\vee}) 
    \end{align}
    \begin{align}
        \sum_{\color{red}{s} \normalcolor \in \Irr (\caC)}
        \tikzmath{
            \draw[thick] (0,0) -- (0,2);
            \draw[thick] (2,0) -- (2,2);
            \filldraw[thick, fill=white] (1,2) ellipse (1 and .15);
            \draw[very thick, blue] (.4,-.13) -- (0.4,1.87);
            \draw[very thick, blue] (.8,-.15) -- (0.8,1.85);
            \draw[very thick, blue] (1.2,-.15) -- (1.2,1.85);
            \draw[very thick, blue] (1.6,-.13) -- (1.6,1.87);
            \halfDottedEllipse{(0,0)}{1}{.15};
            \halfDottedEllipse[red]{(0,1)}{1}{.15};
            \roundNboxEllipse[]{(1,1)}{1}{.15}{-140}{-40}{.5}{};
            \draw[very thick, blue] (0.8,1.85) -- (0.8,4);
            \draw[very thick, blue] (1.2,1.85) -- (1.2,4);
            \draw[thick, knot, orange, decorate, decoration={snake, segment length=1mm, amplitude=.2mm}] (1.2,3) to[out=0,in=270] (2,4);
            \roundNbox{fill=white}{(1,3)}{.25}{.2}{.2}{};
        }
        \; \xmapsto{\sim} \; 
        \sum_{\color{red}{s} \normalcolor \in \Irr (\caC)}
        \tikzmath{
            \draw[very thick, blue] (.4,0) -- (0.4,4);
            \draw[very thick, blue] (.8,0) -- (0.8,4);
            \draw[very thick, blue] (1.2,0) -- (1.2,4);
            \draw[very thick, blue] (1.6,0) -- (1.6,4);
            \draw[thick, red] (1.6,1) to[out=0,in=270] (2.5,4);
            \draw[thick, red] (.4,1) to[out=180,in=270] (-.5,4);
            \roundNbox{fill=white}{(1,1)}{.25}{.5}{.5}{};
            \draw[thick, knot, orange, decorate, decoration={snake, segment length=1mm, amplitude=.2mm}] (1.2,3) to[out=60,in=270] (1.4,4);
            \roundNbox{fill=white}{(1,3)}{.25}{.2}{.2}{};
        }
        \; = \; 
        \sum_{\color{red}{s} \normalcolor \in \Irr (\caC)}
        \tikzmath{
            \draw[very thick, blue] (.4,0) -- (0.4,2);
            \draw[very thick, blue] (.8,0) -- (0.8,2);
            \draw[very thick, blue] (1.2,0) -- (1.2,2);
            \draw[very thick, blue] (1.6,0) -- (1.6,2);
            \draw[thick, red] (1.6,1) to[out=0,in=270] (2.5,2);
            \draw[thick, red] (.4,1) to[out=180,in=270] (-.5,2);
            \draw[thick, knot, orange, decorate, decoration={snake, segment length=1mm, amplitude=.2mm}] (1.2,1) to[out=60,in=270] (1.4,2);
            \roundNbox{fill=white}{(1,1)}{.25}{.5}{.5}{};
        }
    \end{align}
    Meanwhile, 
    \begin{align}
        \caZ (\caC) (I(X^{\otimes k}) \to I(X^{\otimes k}) \otimes \widetilde{Z}) 
        \cong& \bigoplus_{s \in \Irr (\caC)} \caC (X^{\otimes k} \to s \otimes X^{\otimes a} \otimes X^{\otimes (\#I)} \otimes X^{\otimes b} \otimes s^{\vee} \otimes Z) 
    \end{align}
    \begin{equation}
        \sum_{\color{red}{s} \normalcolor \in \Irr (\caC)}
        \tikzmath{
            \draw[thick] (0,0) -- (0,2);
            \draw[thick] (2,0) -- (2,2);
            \filldraw[thick, fill=white] (1,2) ellipse (1 and .15);
            \draw[very thick, blue] (.4,-.13) -- (0.4,1.87);
            \draw[very thick, blue] (.8,-.15) -- (0.8,1.85);
            \draw[very thick, blue] (1.2,-.15) -- (1.2,1.85);
            \draw[very thick, blue] (1.6,-.13) -- (1.6,1.87);
            \halfDottedEllipse{(0,0)}{1}{.15};
            \halfDottedEllipse[red]{(0,1)}{1}{.15};
            \draw[thick, knot, orange, decorate, decoration={snake, segment length=1mm, amplitude=.2mm}] (1.8,0.8) to[out=0,in=270] (3,2);
            \roundNboxEllipse[]{(1,1)}{1}{.15}{-140}{-40}{.5}{};
        }
        \; \xmapsto{\sim} \; 
        \sum_{\color{red}{s} \normalcolor \in \Irr (\caC)}
        \tikzmath{
            \draw[very thick, blue] (.4,0) -- (0.4,2);
            \draw[very thick, blue] (.8,0) -- (0.8,2);
            \draw[very thick, blue] (1.2,0) -- (1.2,2);
            \draw[very thick, blue] (1.6,0) -- (1.6,2);
            \draw[thick, red] (1.6,1) to[out=0,in=270] (2.5,2);
            \draw[thick, red] (.4,1) to[out=180,in=270] (-.5,2);
            \draw[thick, knot, orange, decorate, decoration={snake, segment length=1mm, amplitude=.2mm}] (1.2,.9) to[out=0,in=270] (3,2);
            \roundNbox{fill=white}{(1,1)}{.25}{.5}{.5}{};
        }
    \end{equation}
    Define a natural linear isomorphism 
    \begin{align}
        \Phi: &\bigoplus_{s \in \Irr (\caC)} \caC (X^{\otimes k} \to s \otimes X^{\otimes a} \otimes X^{\otimes (\#I)} \otimes Z \otimes X^{\otimes b} \otimes s^{\vee}) \\ 
        \to& \bigoplus_{s \in \Irr (\caC)} \caC (X^{\otimes k} \to s \otimes X^{\otimes a} \otimes X^{\otimes (\#I)} \otimes X^{\otimes b} \otimes s^{\vee} \otimes Z) \\ 
        &f \mapsto (\id_{s} \otimes \id_{X^{\otimes a + (\# I)}} \otimes \sigma_{Z, X^{\otimes b} \otimes s^\vee}) \circ f
    \end{align}
    The compatibility with the right $\Tube_\caC (X^{\otimes k})$ action is a direct consequence of graphical calculus.
    For left action, let $U \in \frF_D (\ZZ)$ and $a \in \caA_U$ and $f_{\color{red}{s}} = f^1_{\color{red}{s}} f^2_{\color{red}{s}}$ where $f^1_{\color{red}{s}} \in \caC (X^{\otimes (\# I)} \to X^{\otimes (\# I)} \otimes Z)$ and $f^2_{\color{red}{s}} \in \caC (X^{\otimes k} \to s \otimes X^{\otimes k} \otimes s^\vee)$.
    \begin{enumerate}
        \item if $[U] = u / k \ZZ$ does not contain $\{k, 1\}$, 
        \begin{align}
            &\Phi \paren*{a \triangleright \sum_{\color{red}{s} \normalcolor \in \Irr (\caC)} f_{\color{red}{s}}} 
            = \Phi \paren*{\sum_{\color{red}{s} \normalcolor \in \Irr (\caC)}  \sum_i b_i^I \braket{b_i^I}{a f^1_{\color{red}{s}}}_M f^2_{\color{red}{s}}}  \\
            =& \Phi \paren*{
            \sum_{\substack{\color{red}{s} \normalcolor \in \Irr (\caC) \\ b_i^I}}
            \tikzmath{
                \draw[very thick, blue] (.4,-1) -- (0.4,4);
                \draw[very thick, blue] (.8,-1) -- (0.8,4);
                \draw[very thick, blue] (1.2,-1) -- (1.2,4);
                \draw[very thick, blue] (1.6,-1) -- (1.6,4);
                \draw[thick, red] (1.6,-.5) to[out=45,in=270] (2.5,4);
                \draw[thick, red] (.4,-.5) to[out=135,in=270] (-.5,4);
                \draw[thick, knot, orange, decorate, decoration={snake, segment length=1mm, amplitude=.2mm}] (1.2,.5) to[out=60,in=270] (2,1.5) to[out=90,in=300] (1.2,2.5);
                \draw[thick, knot, orange, decorate, decoration={snake, segment length=1mm, amplitude=.2mm}] (1.2,3.5) to[out=80,in=270] (1.4,4);
                \roundNbox{fill=white}{(1,-.5)}{.25}{.5}{.5}{$f^2_{\color{red}{s}}$};
                \roundNbox{fill=white}{(1,.5)}{.25}{.2}{.2}{$f^1_{\color{red}{s}}$};
                \roundNbox{fill=white}{(1.2,1.5)}{.25}{.3}{.3}{$a$};
                \roundNbox{fill=white}{(1,2.5)}{.25}{.2}{.2}{${b_i^I}^\ast$};
                \roundNbox{fill=white}{(1,3.5)}{.25}{.2}{.2}{$b_i^I$};
            }}  
            =
            \sum_{\substack{\color{red}{s} \normalcolor \in \Irr (\caC) \\ b_i^I}}
            \tikzmath{
                \draw[very thick, blue] (.4,-1) -- (0.4,4);
                \draw[very thick, blue] (.8,-1) -- (0.8,4);
                \draw[very thick, blue] (1.2,-1) -- (1.2,4);
                \draw[very thick, blue] (1.6,-1) -- (1.6,4);
                \draw[thick, red] (1.6,-.5) to[out=45,in=270] (2.5,4);
                \draw[thick, red] (.4,-.5) to[out=135,in=270] (-.5,4);
                \draw[thick, knot, orange, decorate, decoration={snake, segment length=1mm, amplitude=.2mm}] (1.2,.5) to[out=60,in=270] (2,1.5) to[out=90,in=300] (1.2,2.5);
                \draw[thick, knot, orange, decorate, decoration={snake, segment length=1mm, amplitude=.2mm}] (1.2,3.5) to[out=0,in=270] (3,4);
                \roundNbox{fill=white}{(1,-.5)}{.25}{.5}{.5}{$f^2_{\color{red}{s}}$};
                \roundNbox{fill=white}{(1,.5)}{.25}{.2}{.2}{$f^1_{\color{red}{s}}$};
                \roundNbox{fill=white}{(1.2,1.5)}{.25}{.3}{.3}{$a$};
                \roundNbox{fill=white}{(1,2.5)}{.25}{.2}{.2}{${b_i^I}^\ast$};
                \roundNbox{fill=white}{(1,3.5)}{.25}{.2}{.2}{$b_i^I$};
            } \\ 
            =& 
            \sum_{\color{red}{s} \normalcolor \in \Irr (\caC)}
            \tikzmath{
                \draw[very thick, blue] (.4,-1) -- (0.4,2);
                \draw[very thick, blue] (.8,-1) -- (0.8,2);
                \draw[very thick, blue] (1.2,-1) -- (1.2,2);
                \draw[very thick, blue] (1.6,-1) -- (1.6,2);
                \draw[thick, red] (1.6,-.5) to[out=0,in=270] (2.5,2);
                \draw[thick, red] (.4,-.5) to[out=180,in=270] (-.5,2);
                \draw[thick, knot, orange, decorate, decoration={snake, segment length=1mm, amplitude=.2mm}] (1.2,.5) to[out=45,in=270] (3,2);
                \roundNbox{fill=white}{(1,-.5)}{.25}{.5}{.5}{$f^2_{\color{red}{s}}$};
                \roundNbox{fill=white}{(1,.5)}{.25}{.2}{.2}{$f^1_{\color{red}{s}}$};
                \roundNbox{fill=white}{(1.2,1.5)}{.25}{.3}{.3}{$a$};
            } 
            = a \Phi \paren*{\sum_{\color{red}{s} \normalcolor \in \Irr (\caC)} f_{\color{red}{s}}}
        \end{align}

        \item if $[U] = u / k \ZZ$ contains $\{k, 1\}$, 
        \begin{align}
            &\braket{b_i^I}{a f^1_{\color{red}{s}}}_M
            =
            \tikzmath{
                \draw[very thick, blue] (.4,1) -- (0.4,4);
                \draw[very thick, blue] (.8,1) -- (0.8,4);
                \draw[very thick, blue] (1.2,1) -- (1.2,4);
                \draw[very thick, blue] (1.6,1) -- (1.6,4);
                \draw[thick, blue] (2,1) -- (2,4);
                \draw[thick, knot, orange, decorate, decoration={snake, segment length=1mm, amplitude=.2mm}] (1.2,1.5) to[out=60,in=270] (2.4,2.5) to[out=90,in=300] (1.2,3.5);
                \roundNbox{fill=white}{(1,1.5)}{.25}{.2}{.2}{$f^1_{\color{red}{s}}$};
                \roundNbox{fill=white}{(1.6,2.5)}{.25}{.3}{.3}{$a$};
                \roundNbox{fill=white}{(1,3.5)}{.25}{.2}{.2}{${b_i^I}^\ast$};
                \roundNbox{fill=white}{(1,3.5)}{.25}{.2}{.2}{$b_i^I$};
            }
            =
            \tikzmath{
                \draw[very thick, blue] (.4,1) -- (0.4,4);
                \draw[very thick, blue] (.8,1) -- (0.8,4);
                \draw[very thick, blue] (1.2,1) -- (1.2,4);
                \draw[very thick, blue] (1.6,1) -- (1.6,4);
                \draw[thick, blue] (2,1) -- (2,4);
                \draw[thick, knot, orange, decorate, decoration={snake, segment length=1mm, amplitude=.2mm}] (1.2,1.5) to[out=60,in=270] (0.9,2.5) to[out=90,in=300] (1.2,3.5);
                \roundNbox{fill=white}{(1,1.5)}{.25}{.2}{.2}{$f^1_{\color{red}{s}}$};
                \roundNbox{fill=white}{(1.6,2.5)}{.25}{.3}{.3}{$a$};
                \roundNbox{fill=white}{(1,3.5)}{.25}{.2}{.2}{${b_i^I}^\ast$};
                \roundNbox{fill=white}{(1,3.5)}{.25}{.2}{.2}{$b_i^I$};
            } \in \caA_{I \cup U} 
            \hookrightarrow 
            \sum_{\color{red}{t} \normalcolor \in \Irr (\caC)}
            \tikzmath{
                \draw[very thick, blue] (.4,1) -- (0.4,4);
                \draw[very thick, blue] (.8,1) -- (0.8,4);
                \draw[very thick, blue] (1.2,1) -- (1.2,4);
                \draw[very thick, blue] (1.6,1) -- (1.6,4);
                \draw[thick, red] (1.6,2.5) to[out=0,in=270] (2.5,4);
                \draw[thick, red] (.4,2.5) to[out=180,in=270] (-.5,4);
                \draw[thick, knot, orange, decorate, decoration={snake, segment length=1mm, amplitude=.2mm}] (1.2,1.5) to[out=60,in=270] (0.9,2.5) to[out=90,in=300] (1.2,3.5);
                \roundNbox{fill=white}{(1,1.5)}{.25}{.2}{.2}{$f^1_{\color{red}{s}}$};
                \roundNbox{fill=white}{(1.4,2.5)}{.25}{.1}{.1}{};
                \roundNbox{fill=white}{(.4,2.5)}{.25}{0}{0}{};
                \roundNbox{fill=white}{(1,3.5)}{.25}{.2}{.2}{${b_i^I}^\ast$};
                \roundNbox{fill=white}{(1,3.5)}{.25}{.2}{.2}{$b_i^I$};
            }
        \end{align}
        Thus,
        \begin{align}
            &\Phi \paren*{a \triangleright \sum_{\color{red}{s} \normalcolor \in \Irr (\caC)} f_{\color{red}{s}}} 
            = \Phi \paren*{\sum_{\color{red}{s} \normalcolor \in \Irr (\caC)}  \sum_i b_i^I \braket{b_i^I}{a f^1_{\color{red}{s}}}_M f^2_{\color{red}{s}}}  \\
            =& \Phi \paren*{
            \sum_{\substack{\color{red}{s, t, u} \normalcolor \in \Irr (\caC) \\ \tikzmath{
                \fill[red] (.3,.3) circle (.05cm);
                \draw[thick, red] (.3,.3) -- (.3,.45);
                \draw[thick, red] (.3,.3) -- (.15,.15);
                \draw[thick, red] (.3,.3) -- (.45,.15);
            } \\ b_i^I}}
            \tikzmath{
                \draw[very thick, blue] (.4,-1) -- (0.4,4);
                \draw[very thick, blue] (.8,-1) -- (0.8,4);
                \draw[very thick, blue] (1.2,-1) -- (1.2,4);
                \draw[very thick, blue] (1.6,-1) -- (1.6,4);
                \draw[thick, red] (1.6,-.5) to[out=45,in=270] (2.5,3);
                \draw[thick, red] (.4,-.5) to[out=135,in=270] (-.5,3);
                \draw[thick, red] (1.6,1.5) to[out=0,in=270] (2.5,3);
                \draw[thick, red] (.4,1.5) to[out=180,in=270] (-.5,3);
                \draw[thick, red] (2.5,3) -- (2.5,4);
                \draw[thick, red] (-.5,3) -- (-.5,4);
                \node at (2.5,3) {$\color{red}{\bullet}$};
                \node at (-.5,3) {$\color{red}{\bullet}$};
                \draw[thick, knot, orange, decorate, decoration={snake, segment length=1mm, amplitude=.2mm}] (1.2,.5) to[out=60,in=270] (0.9,1.5) to[out=90,in=300] (1.2,2.5);
                \draw[thick, knot, orange, decorate, decoration={snake, segment length=1mm, amplitude=.2mm}] (1.2,3.5) to[out=80,in=270] (1.4,4);
                \roundNbox{fill=white}{(1,-.5)}{.25}{.5}{.5}{$f^2_{\color{red}{s}}$};
                \roundNbox{fill=white}{(1,.5)}{.25}{.2}{.2}{$f^1_{\color{red}{s}}$};
                \roundNbox{fill=white}{(1.4,1.5)}{.25}{.1}{.1}{};
                \roundNbox{fill=white}{(.4,1.5)}{.25}{0}{0}{};
                \roundNbox{fill=white}{(1,2.5)}{.25}{.2}{.2}{${b_i^I}^\ast$};
                \roundNbox{fill=white}{(1,3.5)}{.25}{.2}{.2}{$b_i^I$};
            }}  
            =
            \sum_{\substack{\color{red}{s, t, u} \normalcolor \in \Irr (\caC) \\ \tikzmath{
                \fill[red] (.3,.3) circle (.05cm);
                \draw[thick, red] (.3,.3) -- (.3,.45);
                \draw[thick, red] (.3,.3) -- (.15,.15);
                \draw[thick, red] (.3,.3) -- (.45,.15);
            } \\ b_i^I}}
            \tikzmath{
                \draw[very thick, blue] (.4,-1) -- (0.4,4);
                \draw[very thick, blue] (.8,-1) -- (0.8,4);
                \draw[very thick, blue] (1.2,-1) -- (1.2,4);
                \draw[very thick, blue] (1.6,-1) -- (1.6,4);
                \draw[thick, red] (1.6,-.5) to[out=45,in=270] (2.5,3);
                \draw[thick, red] (.4,-.5) to[out=135,in=270] (-.5,3);
                \draw[thick, red] (1.6,1.5) to[out=0,in=270] (2.5,3);
                \draw[thick, red] (.4,1.5) to[out=180,in=270] (-.5,3);
                \draw[thick, red] (2.5,3) -- (2.5,4);
                \draw[thick, red] (-.5,3) -- (-.5,4);
                \node at (2.5,3) {$\color{red}{\bullet}$};
                \node at (-.5,3) {$\color{red}{\bullet}$};
                \draw[thick, knot, orange, decorate, decoration={snake, segment length=1mm, amplitude=.2mm}] (1.2,.5) to[out=60,in=270] (0.9,1.5) to[out=90,in=300] (1.2,2.5);
                \draw[thick, knot, orange, decorate, decoration={snake, segment length=1mm, amplitude=.2mm}] (1.2,3.5) to[out=0,in=270] (3,4);
                \roundNbox{fill=white}{(1,-.5)}{.25}{.5}{.5}{$f^2_{\color{red}{s}}$};
                \roundNbox{fill=white}{(1,.5)}{.25}{.2}{.2}{$f^1_{\color{red}{s}}$};
                \roundNbox{fill=white}{(1.4,1.5)}{.25}{.1}{.1}{};
                \roundNbox{fill=white}{(.4,1.5)}{.25}{0}{0}{};
                \roundNbox{fill=white}{(1,2.5)}{.25}{.2}{.2}{${b_i^I}^\ast$};
                \roundNbox{fill=white}{(1,3.5)}{.25}{.2}{.2}{$b_i^I$};
            } \\ 
            =& 
            \sum_{\substack{\color{red}{s, t, u} \normalcolor \in \Irr (\caC) \\ \tikzmath{
                \fill[red] (.3,.3) circle (.05cm);
                \draw[thick, red] (.3,.3) -- (.3,.45);
                \draw[thick, red] (.3,.3) -- (.15,.15);
                \draw[thick, red] (.3,.3) -- (.45,.15);
            }}}
            \tikzmath{
                \draw[very thick, blue] (.4,-1) -- (0.4,3);
                \draw[very thick, blue] (.8,-1) -- (0.8,3);
                \draw[very thick, blue] (1.2,-1) -- (1.2,3);
                \draw[very thick, blue] (1.6,-1) -- (1.6,3);
                \draw[thick, red] (1.6,-.5) to[out=45,in=270] (2.5,2.5);
                \draw[thick, red] (.4,-.5) to[out=135,in=270] (-.5,2.5);
                \draw[thick, red] (1.6,1.5) to[out=0,in=270] (2.5,2.5);
                \draw[thick, red] (.4,1.5) to[out=180,in=270] (-.5,2.5);
                \draw[thick, red] (2.5,2.5) -- (2.5,3);
                \draw[thick, red] (-.5,2.5) -- (-.5,3);
                \node at (2.5,2.5) {$\color{red}{\bullet}$};
                \node at (-.5,2.5) {$\color{red}{\bullet}$};
                \draw[thick, knot, orange, decorate, decoration={snake, segment length=1mm, amplitude=.2mm}] (1.2,.5) to[out=30,in=270] (3,3);
                \roundNbox{fill=white}{(1,-.5)}{.25}{.5}{.5}{$f^2_{\color{red}{s}}$};
                \roundNbox{fill=white}{(1,.5)}{.25}{.2}{.2}{$f^1_{\color{red}{s}}$};
                \roundNbox{fill=white}{(1.4,1.5)}{.25}{.1}{.1}{};
                \roundNbox{fill=white}{(.4,1.5)}{.25}{0}{0}{};
            }
            = a \Phi \paren*{\sum_{\color{red}{s} \normalcolor \in \Irr (\caC)} f_{\color{red}{s}}}
        \end{align}
    \end{enumerate} 
\end{proof}

The next lemma states that the compactification of a DHR bimodule is compatible with the compactification of a bounded-spread isomorphism. 
We remark that this proof uses purely algebraic description of the DHR bimodule and the compactified bimodule and therefore works for an abstract quasi-local *-algebra, given that $\Cpt (\alpha): \Cpt_{D, T}^G (\caA, \delta) \to \Cpt_{D', T'}^G (\caA, \delta)$ is an isomorphism. 
Again, we take advantage of the stabilization of $\Cpt_{D,T}^{k \ZZ} (\caA, \delta)$ with respect to $T$ in the fusion spin chain case.

\begin{lemma}\label{lem:fusionspinchainalgebra-DHR-Tube-boundedspread}
    Let $\caA = \caA(\caC, X)$ be a fusion spin chain algebra with $\caC$ a unitary fusion category and $X$ a strong tensor generator with $n \ge 1$.
    Let $\alpha: \caA \to \caA$ a $\ZZ$-equivariant bounded-spread isomorphism with spread $s$, $k > \max \{14 n + 4 s, 6 n + 12 s\} + 3$.
    Then, the following commutes up to natural isomorphism
    \begin{equation}\begin{tikzcd}[column sep=huge]
        _\caA \DHR_\caA \arrow[r, "\DHR(\alpha)"] \arrow[d] & _\caA \DHR_\caA \arrow[d] \\ 
        _{\Tube_\caC (X^{\otimes k})} \pHMod_{\Tube_\caC (X^{\otimes k})} \arrow{r}[swap]{(\Tube^k (\alpha^{-1}))^*} & _{\Tube_\caC (X^{\otimes k})} \pHMod_{\Tube_\caC (X^{\otimes k})} 
    \end{tikzcd}\end{equation}
\end{lemma}
\begin{proof}
    Choose some $I \in \frB(\ZZ)$ such that $\diam (I) \ge n$ and $k > 12n + 2 \diam (I^{+s}) + 3$. 
    By \autoref{thm:fusionspinchainalgebra-DTIcptmod}, there are monoidal functors $M \mapsto \widetilde{M_I}$ and $M \mapsto \widetilde{M_{I^{+s}}}$, which are independent of choice of $I$ up to natural isomorphism.
    By \autoref{cor:fusionspinchainalgebra-DTcptalgebra-boundedspread}, $\Tube^k(\alpha)$ is an isomorphism.

    Recall that $\DHR(\alpha) (M) = {}^{\alpha^{-1}} M^{\alpha^{-1}}$ where ${}^{\alpha^{-1}} M^{\alpha^{-1}}$ is $M$ as a vector space with right $\caA$ action by $a \triangleright_{\alpha^{-1}} m := \alpha^{-1} (a) m$, $m \triangleleft_{\alpha^{-1}} a := m \alpha^{-1} (a)$ and $\braket*{m}{m'}_{{}^{\alpha^{-1}} M^{\alpha^{-1}}} := \alpha (\braket*{m}{m'}_M)$.
    The functor $(\Tube^k (\alpha^{-1}))^*$ is defined similarly.

    We prove there is a natural isomorphism $\widetilde{(^{\alpha^{-1}} M^{\alpha^{-1}})_I} \cong ^{\Tube^k (\alpha^{-1})} \widetilde{M_{I^{+s}}} {}^{\Tube^k (\alpha^{-1})}$.
    Observe that we have a $\CC$ vector space inclusion $(^{\alpha^{-1}} M^{\alpha^{-1}})_I \hookrightarrow M_{I^{+s}}$.
    Let $\{b_i^{I, \alpha}\}$ be a projective basis of ${}^{\alpha^{-1}} M^{\alpha^{-1}}$ localized in $I$. 
    Then, $\{b_i^{I, \alpha} \boxtimes 1\}$ is a projective basis of $\widetilde{(^{\alpha^{-1}} M^{\alpha^{-1}})_I}$ as a right $\Tube_\caC (X^{\otimes k})$ module.
    Define 
    \begin{align}
        \Omega_M: \widetilde{(^{\alpha^{-1}} M^{\alpha^{-1}})_I} &\to ^{\Tube^k (\alpha^{-1})} \widetilde{M_{I^{+s}}} {}^{\Tube^k (\alpha^{-1})} \\ 
        b_i^{I, \alpha} \boxtimes 1 &\mapsto b_i^{I, \alpha} \boxtimes 1
    \end{align}
    This is an isometry on the basis and therefore extends to a well-defined injective map by right $\Tube_{\caC} (X^{\otimes k})$ linearity. 
    \begin{align}
        &\braket*{\Omega_M (b_i^{I, \alpha} \boxtimes 1)}{\Omega_M (b_j^{I, \alpha} \boxtimes 1)}_{^{\Tube^k (\alpha^{-1})} \widetilde{M_{I^{+s}}} {}^{\Tube^k (\alpha^{-1})}} \\ 
        =& \Tube^k (\alpha) \paren*{\braket*{b_i^{I, \alpha} \boxtimes 1}{b_j^{I, \alpha} \boxtimes 1}_{\widetilde{M_{I^{+s}}}}} \\
        =& \alpha \paren*{\braket*{b_i^{I, \alpha}}{b_j^{I, \alpha}}_{M_{I^{+s}}}} \\
        =& \braket*{b_i^{I, \alpha}}{b_j^{I, \alpha}}_{^{\alpha^{-1}} M^{\alpha^{-1}}} \\
        =& \braket*{b_i^{I, \alpha} \boxtimes 1}{b_j^{I, \alpha} \boxtimes 1}_{\widetilde{(^{\alpha^{-1}} M^{\alpha^{-1}})_I}}
    \end{align}
    For surjectivity, it suffices to observe that $M_{I^{+s}} = \sum_i b_i^{I, \alpha} \caA_{I^{+s}}$. 
    Indeed, if $\sum_i b_i^{I, \alpha} a_i \in \sum_i b_i^{I, \alpha} \caA_{I^{+s}}$, for any $a \in \caA_{(I^{+s})^{c}}$, $\paren*{\sum_i b_i^{I, \alpha} a_i} a = a \paren*{\sum_i b_i^{I, \alpha} a_i}$ since $\alpha^{-1}(\caA_{I^c}) \supseteq \caA_{(I^{+s})^c} \ni a$ has commuting action on $b_i^{I, \alpha}$.
    Conversely, let $m \in M_{I^{+s}}$. Then, $m = \sum_i b_i^{I, \alpha} \triangleleft_{\alpha^{-1}} \braket*{b_i^{I, \alpha}}{m}_{^{\alpha^{-1}} M^{\alpha^{-1}}} = \sum_i b_i^{I, \alpha}\braket*{b_i^{I, \alpha}}{m}_{M} \in \sum_i b_i^{I, \alpha} \caA_{I^{+s}}$.

    It is also a left $\Tube_\caC (X^{\otimes k})$ module map. 
    For any choice of parameters $D, T$ satisfying the conditions in \autoref{thm:fusionspinchainalgebra-DTIcptmod} and $a \in \caA_U$ with $U \in \frF_D (\ZZ)$, 
    \begin{align}
        \Omega_M (a (b_i^{I, \alpha} \boxtimes t))
        =& \Omega_M \paren*{\sum_j b_j^{I,\alpha} \boxtimes \braket*{b_j^{I,\alpha}}{\phi_{+mk} (a) \triangleright_{\alpha^{-1}} b_i^{I, \alpha}}_{^{\alpha^{-1}} M^{\alpha^{-1}}} t} \\
        =& \Omega_M \paren*{\sum_j b_j^{I,\alpha} \boxtimes \alpha\paren*{\braket*{b_j^{I,\alpha}}{\alpha^{-1}(\phi_{+mk} (a)) b_i^{I, \alpha}}_M} t} \\
        =& \sum_j \paren*{b_j^{I,\alpha} \boxtimes 1} \triangleleft_{\Tube^k (\alpha^{-1})} \paren*{\alpha\paren*{\braket*{b_j^{I,\alpha}}{\phi_{+mk} (\alpha^{-1}(a)) b_i^{I, \alpha}}_M} t} \\
        =& \sum_j \paren*{b_j^{I,\alpha} \boxtimes \braket*{b_j^{I,\alpha}}{\phi_{+mk} (\alpha^{-1}(a)) b_i^{I, \alpha}}_M} \triangleleft_{\Tube^k (\alpha^{-1})} t \\
        =& \alpha^{-1}(a) \paren*{b_i^{I, \alpha} \boxtimes 1} \triangleleft_{\Tube^k (\alpha^{-1})} t \\
        =& a \triangleright_{\alpha^{-1}} \Omega_M \paren*{b_i^{I, \alpha} \boxtimes t}.
    \end{align}
\end{proof}

\subsection{Obstruction to Braided Autoequivalence from QCA}\label{subsec:obstruction}
We relate $\DHR (\alpha)$ on  $\caZ (\caC)$ and $\Tube (\alpha^{-1})^*$ (right action twist) on $\Rep (\Tube_\caC (X^{\otimes}))$.
Since $\caZ (\caC) (I(X^{\otimes k}) \to I(X^{\otimes k})) = \Tube_\caC (X^{\otimes k}) \in \Rep (\Tube_\caC (X^{\otimes k}))$ is fixed by $\Tube (\alpha^{-1})^*$ up to isomorphism, this gives a certain stabilization property of a braided autoequivalence of $\caZ (\caC)$ that arises from a bounded-spread automorphism (quantum cellular automaton) of the underlying fusion spin chain.
It is certainly not a complete criterion; nevertheless, it is a strong negative criterion because it requires stabilization for all sufficiently large $k$.

\begin{theorem}\label{thm:fusionspinchainalgebra-DHR-Picard}
    Let $\caA = \caA(\caC, X)$ be a fusion spin chain algebra with $\caC$ a unitary fusion category and $X$ a strong tensor generator with $n \ge 1$.
    Let $\alpha: \caA \to \caA$ a $\ZZ$-equivariant bounded-spread isomorphism with spread $s$, $k > \max \{14 n + 4 s, 6 n + 12 s\}$.
    Let $\DHR(\alpha): \caZ (\caC) \to \caZ (\caC)$ be the braided tensor autoequivalence via canonical equivalence $\caZ (\caC) \cong _\caA \DHR_\caA$.
    Then, there exists $\caL_{k,\alpha} \in \Inv (\caZ (\caC))$ such that the following commutes up to natural isomorphism.
    \begin{equation}\begin{tikzcd}[column sep=huge]
        \caZ (\caC) \arrow [r, "\DHR(\alpha)"] \arrow{d}[swap]{E:= \caZ (\caC)(I(X^{\otimes k}) \to -)} & \caZ (\caC) \arrow[r, "\caL_{k,\alpha} \otimes -"] & \caZ (\caC) \\ 
        \Rep (\Tube_\caC (X^{\otimes k})) \arrow[r, "\Tube^k (\alpha^{-1})^*"] & \Rep (\Tube_\caC (X^{\otimes k})) \arrow{ur}[swap]{-\otimes I(X^{\otimes k}) =: E^{-1}}
    \end{tikzcd}\end{equation}
\end{theorem}
\begin{proof}
    By \autoref{thm:fusionspinchainalgebra-DTIcptmod}, we have a standard right $\caZ (\caC)$ action $\triangleleft$ on $\Rep (\Tube_\caC (X^{\otimes}))$ by $\caZ (\caC) \to _\caA \DHR_\caA \to _{\Tube_\caC (X^{\otimes k})} \pHMod_{\Tube_\caC (X^{\otimes k})}$.
    \begin{align}
        \triangleleft: \Rep (\Tube_\caC (X^{\otimes})) \times \caZ (\caC) \to& \Rep (\Tube_\caC (X^{\otimes})) \\
        (R, Z) \mapsto& R \boxtimes \widetilde{M(Z)_I} \\ 
        (f, g) \mapsto& f \boxtimes \widetilde{M(g)_I} 
    \end{align}
    \begin{equation}
        \tikzmath{
            \begin{scope}[]
                \draw[thick] (0,2) arc (180:0:1);
                \draw[thick, orange, decorate, decoration={snake, segment length=1mm, amplitude=.2mm}] ($ (0,2) + (1,1)$) --node[left]{$R$} ($ (0,2) + (1,2)$);
            \end{scope}

            \draw[thick] (0,0) -- (0,2);
            \draw[thick] (2,0) -- (2,2);
            \draw[very thick, blue] (.4,-.13) -- (0.4,1.87);
            \draw[very thick, blue] (.8,-.15) -- (0.8,1.85);
            \draw[very thick, blue] (1.2,-.15) -- (1.2,1.85);
            \draw[very thick, blue] (1.6,-.13) -- (1.6,1.87);
            \halfDottedEllipse{(0,0)}{1}{.15};
            \halfDottedEllipse[red]{(0,1)}{1}{.15};
            \halfDottedEllipse{(0, 2)}{1}{.15};
            \draw[thick, knot, orange, decorate, decoration={snake, segment length=1mm, amplitude=.2mm}] (1.8,0.8) node[right, xshift=.6cm, yshift=.1cm]{$(Z, \sigma)$} to[out=0,in=270] (3,2);
            \roundNboxEllipse[]{(1,1)}{1}{.15}{-140}{-40}{.5}{};
        }
    \end{equation}
    Let $E := \caZ (\caC)(I(X^{\otimes k}) \to -): \caZ (\caC) \to \Rep (\Tube_\caC (X^{\otimes k}))$ and $E^{-1} := (-) \otimes I(X^{\otimes k}): \Rep (\Tube_\caC (X^{\otimes k})) \to \caZ (\caC)$ be its quasi-inverse.
    Observe that for $Z_1, Z_2 \in \caZ (\caC)$,
    \begin{align} 
        E (Z_1 \otimes Z_2) 
        \cong& \caZ (\caC) (I(X^{\otimes k}) \to Z_1 \otimes Z_2) \\ 
        \cong& \caZ (\caC) (I(X^{\otimes k}) \to Z_1) \boxtimes_{\Tube_\caC (X^{\otimes k})} \caZ (\caC) (I(X^{\otimes k}) \to I(X^{\otimes k}) \otimes Z_2) \\ 
        \cong& E(Z_1) \triangleleft Z_2 
    \end{align}
    Thus, the action $\triangleleft$ is equivalent to the right $\caZ (\caC)$ action on $\caZ (\caC)$ by monoidal product.

    Define $\triangleleft_{\DHR(\alpha)} := \triangleleft \circ (\Cat{id} \times \DHR(\alpha))$ and let $\Rep (\Tube_\caC (X^{\otimes}))^{\DHR(\alpha)}$ denote the right $\caZ (\caC)$ module category with the same underlying category but equipped with action $\triangleleft_{\DHR(\alpha)}$. 
    We prove $\Tube^k (\alpha^{-1})^*: \Rep (\Tube_\caC (X^{\otimes k})) \to \Rep (\Tube_\caC (X^{\otimes k}))^{\DHR(\alpha)}$ is an equivalence of $\caZ(\caC)$ module categories.
    Take $R \in \Rep (\Tube_\caC (X^{\otimes k}))$, $Z \in \caZ (\caC)$.
    Then,
    \begin{align}
        \Tube^k (\alpha^{-1})^* (R \triangleleft Z)
        \cong& \paren*{R \boxtimes \widetilde{M(Z)_I}}^{\Tube^k (\alpha^{-1})} \\
        \cong& R^{\Tube^k (\alpha^{-1})} \boxtimes {}^{\Tube^k (\alpha^{-1})} \widetilde{M(Z)_I} {}^{\Tube^k (\alpha)} \\
        \cong& \Tube^k (\alpha^{-1})^* (R) \boxtimes \widetilde{^{\alpha^{-1}} M(Z)^{\alpha^{-1}}_{I^{+s}}} \\ 
        \cong& \Tube^k (\alpha^{-1})^* (R) \triangleleft \DHR(\alpha) (Z),
    \end{align}
    all the isomorphisms being natural in $R$ and $Z$, and associativity and unit coherence satisfied.

    Consider $E^{-1} \circ \Tube^k (\alpha^{-1})^* \circ E \circ \DHR(\alpha^{-1}): \caZ (\caC) \to \caZ (\caC)$.
    It is an equivalence as $\caZ (\caC)$ module: 
    \begin{align}
        &E^{-1} \circ \Tube^k (\alpha^{-1})^* \circ E \circ \DHR(\alpha^{-1}) (Z_1 \otimes Z_2) \\
        \cong& E^{-1} \circ \Tube^k (\alpha^{-1})^* \circ E (\DHR(\alpha^{-1})(Z_1) \otimes \DHR(\alpha^{-1})(Z_2)) \\ 
        \cong& E^{-1} \circ \Tube^k (\alpha^{-1})^* (E (\DHR(\alpha^{-1})(Z_1)) \triangleleft \DHR(\alpha^{-1})(Z_2)) \\ 
        \cong& E^{-1} (\Tube^k (\alpha^{-1})^* (E (\DHR(\alpha^{-1})(Z_1))) \triangleleft Z_2) \\ 
        \cong& E^{-1} (\Tube^k (\alpha^{-1})^* (E (\DHR(\alpha^{-1})(Z_1)))) \otimes Z_2 
    \end{align}
    Thus, there exists an invertible object $\caL_{\alpha, k} = E^{-1} \circ \Tube^k (\alpha^{-1})^* \circ E \circ \DHR(\alpha^{-1}) (1_{\caZ (\caC)}) \in \Inv (\caZ (\caC))$ such that $E^{-1} \circ \Tube^k (\alpha^{-1})^* \circ E \circ \DHR(\alpha^{-1}) \cong \caL_{\alpha, k} \otimes -$.
\end{proof}

\begin{corollary}\label{cor:fusionspinchainalgebra-obstruction}
    Let $\caA = \caA(\caC, X)$ be a fusion spin chain algebra with $\caC$ a unitary fusion category and $X$ a strong tensor generator with $n \ge 1$.
    Let $\alpha: \caA \to \caA$ a $\ZZ$-equivariant bounded-spread isomorphism with spread $s$.
    Then, for any $k > \max \{14 n + 4 s, 6 n + 12 s\}$, $I(X^{\otimes k}) \in \caZ (\caC)$ is fixed by $\DHR(\alpha)$ up to isomorphism and $\Inv (\caZ (\caC))$ orbit: $\gamma_{\alpha, k}: I(X^{\otimes k}) \stackrel{\sim}{\to} \caL_{k,\alpha} \otimes \DHR(\alpha)(I(X^{\otimes k}))$.
\end{corollary}

Once the unitary fusion category $\caC$, a strong tensor generator $X \in \Ob(\caC)$, a braided autoequivalence of its Drinfeld center $F \in \Aut_{\br} (\caZ (\caC))$ are chosen, the theorem implies that if for all $s$ we can find an integer $k \gg s$ such that there is no invertible element in $\Inv(\caZ (\caC))$, then $F$ is non-isomorphic to any $\DHR (\alpha)$ for a bounded spread automorphism $\alpha: \caA (\caC, X) \to \caA (\caC, X)$. 
Therefore, we obtained an obstruction theory for a bulk TQFT symmetry to be implemented on the lattice model. 

Let us look into specific lattice models built from common fusion categories. 
The representation category of a finite abelian group $A$ is one of the simplest examples. 
It is known that its Drinfeld center is computed by $\Rep (D(A))$, the representation category of the quantum double of $A$, and a nondegenerate bicharacter $\beta: A \times A \to \CC$ defines a braided autoequivalence $F_\beta \in \Aut_\br (\Rep (D(A)))$ \cite{EGNO15}. 
For example, when $A = \ZZ / 2 \ZZ$ and $\beta$ is the unique nondegenerate bicharacter,
\begin{align}
    \beta (x, y) = (-1)^{xy}
\end{align}
this induces $e-m$ swap symmetry of $\ZZ / 2$ gauge theory. 
We prove we cannot implement this type of symmetries on the lattice model built from a generic choice of ``spin'' $X \in \Ob (\Rep (A))$. 

\begin{example}\label{ex:fusionspinchainalgebra-dhr-picard-rep-abelian}
    Let $A$ be a finite abelian group, $\caC = \Rep(\CC[A])$, $X = \bigoplus_{\chi \in \widehat{A}} a_\chi \chi \in \caC$ be a strong tensor generating object with $n \ge 0$ that is not a multiple of the regular representation (i.e., $a_{(\cdot)}: \widehat{A} \to \ZZ_{\ge 0}$ is not constant.) 
    Let $\caA = \caA(\caC, X)$ be the corresponding fusion spin chain algebra.
    Then, any $F_\beta \in \CAut_\br (\caZ (\caC))$ defined by a nondegenerate bicharacter $\beta: A \times A \to \CC$ does not arise as $\DHR (\alpha)$ for any $\ZZ$-equivariant quantum cellular automaton $\alpha: \caA \to \caA$.
    
    This obstruction is stable under adding ancila multiplicity space (i.e., replacing $X$ with $X \otimes \CC^{l} \cong X^{\oplus l}$.)
\end{example}
\begin{proof}    
    The simple objects in $\caZ (\caC) \cong \Rep (D(\CC[A]))$ are labeled by $A \times \widehat{A}$ and $c_{(a, \chi), (b, \psi)} = \chi(b) \id_{(b, \psi) \otimes (a, \chi)}$.
    For $\chi \in \Irr (\caC) \cong \widehat{A}$, we have $I(\chi) = \bigoplus_{a \in A} (a, \chi)$ for all $\chi \in \widehat{A}$. 
    (It follows from $I \dashv F$.)
    It is known that a nondegenerate bicharacter $\beta: A \times A \to \CC$ defines $F_\beta \in \CAut_\br (\caZ (\caC))$ 
    \begin{align}
        F_\beta: \caZ (\caC) &\to \caZ (\caC) \\
        F (a, \chi) &= (\phi_\beta^{-1}(\chi), \phi_\beta (a)) \\ 
        (F_\beta)^2_{(a, \chi), (b, \psi)} &= \psi(a) \id_{F_\beta (a, \chi) \otimes F_\beta (b, \psi)} 
    \end{align}
    where $\phi_\beta: A \to \widehat{A}, g \mapsto \beta(g, \cdot)$ is the musical isomorphism.

    Let $\alpha: \caA \to \caA$ be a $\ZZ$-equivariant bounded-spread isomorphism with spread $s$ and suppose $\DHR (\alpha) \cong F_\beta$. 
    Then, by \autoref{cor:fusionspinchainalgebra-obstruction}, for a sufficiently large $k \gg s, n$, there exists $(g, \lambda) \in \Inv (\caZ (\caC)) \cong A \times \widehat{A}$ such that $(g, \lambda) \otimes I(X^{\otimes k}) \cong F_\beta (I(X^{\otimes k}))$.
    \begin{align}
        X^{\otimes k} \cong& \bigoplus_{\chi \in \widehat{A}} a^{\ast k}_\chi \chi \\ 
        (g, \lambda) \otimes I(X^{\otimes k}) \cong& \bigoplus_{\chi \in \widehat{A}} a^{\ast k}_\chi \bigoplus_{a \in A} (g a, \lambda \chi) \\
        F_\beta (I(X^{\otimes k})) \cong& \bigoplus_{\chi \in \widehat{A}} a^{\ast k}_\chi \bigoplus_{a \in A} (\phi_\beta^{-1}(\chi), \phi_\beta (a)),
    \end{align}
    where $a^{\ast k}_\chi = \sum_{\chi_1 \cdots \chi_k = \chi} a_{\chi_1} \cdots a_{\chi_k}$ is the $k$-fold convolution.
    Comparing the multiplicity of $(b, \psi)$, $a^{\ast k}_{\psi \lambda^{-1}} = a^{\ast k}_{\phi_\beta (b)}$.
    Thus, $a^{\ast k}_{(\cdot)}: \widehat{A} \to \ZZ_{\ge 0}$ is constant on $\widehat{A}$.
    After the Fourier transformation,
    \begin{align}
        \widehat{a}^k_{a} 
        =& \widehat{a^{\ast k}}_a \\
        :=& \sum_{\chi \in \widehat{A}} \overline{\chi(a)} a^{\ast k}_\chi \\
        =& \begin{cases}
            a^{\ast k}_{\mathrm{triv}} |A| & a = 1_A \\ 
            0 & a \neq 1_A
        \end{cases}
    \end{align}
    Thus $a^{\ast 1}_{(\cdot)}: \widehat{A} \to \ZZ_{\ge 0}$ is constant on $\widehat{A}$, contradicting the assumption on $X$.
\end{proof}

We also analyze a nonabelian example. 
Observe that although generalizing this to $\Rep (G)$ for an arbitrary group $G$ is difficult, once a group is fixed one can easily apply a similar procedure. 
The invertible objects in $\caZ (\Rep (G))$ are the invertible Yetter-Drinfeld module over $\CC[G]$, which are exactly labeled by the center $Z(G)$ and $\widehat{G}$.

\begin{example}\label{ex:fusionspinchainalgebra-dhr-picard-rep-nonabelian}
    Let $D_4 = \braket{r, s}{r^4 = s^2 = 1, s r s = r^3}$, $\caC = \Rep (D_4)$ and $[\phi] \in \Out (D_4)$ be the unique nontrivial class represented by, for instance, $\phi: r \mapsto r, s \mapsto r s$.
    Let $\chi_{\epsilon, \delta} \in \Irr (\caC) \cong \widehat{D_4}\cong (\ZZ / 2 \ZZ)^2$ be one dimensional representations with $r \mapsto (-1)^{\epsilon}, s \mapsto (-1)^{\delta}$ and $\rho \in \Irr(\caC)$ be the unique irreducible two dimensional representation. 
    Let $(X, n)$, $X = \bigoplus a_{\epsilon, \delta} \chi_{\epsilon, \delta} \oplus b \rho$ be a strong tensor generating object such that 
    \begin{align}\label{eq:ex-fusionspinchainalgebra-d4-condition}
        a_{(0,0)} \neq a_{(0,1)} \;\text{and}\; a_{(1,0)} \neq a_{(1,1)}
    \end{align}
    Then, $F_\phi \in \CAut_\br (\caZ (\caC))$ induced by $\phi$ does not arise as $\DHR (\alpha)$ for any $\ZZ$-equivariant quantum cellular automaton $\alpha$ on the quasi-local *-algebra $\caA (\caC, X)$.

    This obstruction is stable under adding ancila multiplicity space (i.e., replacing $X$ with $X \otimes \CC^{l} \cong X^{\oplus l}$.)
\end{example}
\begin{proof}
    Recall that simple objects in $\caZ (\caC)$ are labeled by $([r^a s^b], \pi)$, where $[r^a s^b]$ is a conjugacy class of $D_4$ and $\pi$ is an irreducible representation on the centralizer $Z_{D_4} ([r^a s^b])$.
    Concretely, there are 22 simple objects: 
    \begin{itemize}
        \item 5 from $[1] = \{1\}$: $(1, \chi_{\epsilon, \delta}), ((\epsilon, \delta) \in (\ZZ / 2 \ZZ)^2)$ and $(1, \rho)$,
        \item 5 from $[r^2] = \{r^2\}$: $(1, \chi_{\epsilon, \delta}), ((\epsilon, \delta) \in (\ZZ / 2 \ZZ)^2)$ and $(1, \rho)$,
        \item 4 from $[r] = \{r, r^3\}$: $(1, \lambda_j)$, where $\lambda_j: Z_{D_4} (r) = \angles*{r} \to U(1), r \mapsto i^j, j = 0, 1, 2, 3$,
        \item 4 from $[s] = \{s, r^2 s\}$: $(1, \nu_{u,v})$, where $\nu_{u,v}: Z_{D_4} (s) = \{1, s, r^2, r^2 s\} \to U(1), s \mapsto (-1)^u, r^2 \mapsto (-1)^v, (u, v) \in (\ZZ / 2 \ZZ)^2$,
        \item 4 from $[rs] = \{r s, r^3 s\}$: $(1, \nu'_{u',v'})$, where $\nu'_{u',v'}: Z_{D_4} (r s) = \{1, r s, r^2, r^3 s\} \to U(1), r s \mapsto (-1)^{u'}, r^2 \mapsto (-1)^{v'}, (u', v') \in (\ZZ / 2 \ZZ)^2$.
    \end{itemize}

    Recall also that $\caC$ has Tambara-Yamagami type fusion rule:
    \begin{align}
        \chi_{\epsilon, \delta} \otimes \chi_{\epsilon', \delta'} =& \chi_{\epsilon + \epsilon', \delta + \delta'} \\ 
        \chi_{\epsilon, \delta} \otimes \rho =& \rho \otimes \chi_{\epsilon, \delta} = \rho \\ 
        \rho \otimes \rho =& \bigoplus_{(\epsilon, \delta) \in (\ZZ / 2 \ZZ)^2} \chi_{\epsilon, \delta}.
    \end{align}
    One can compute $I: \caC \to \caZ (\caC)$ on simple objects by Frobenius reciprocity. 
    \begin{align}
        I(\chi_{\epsilon, \delta}) =& ([1], \chi_{\epsilon, \delta}) \oplus ([r^2], \chi_{\epsilon, \delta}) \oplus ([r], \lambda_{2 \epsilon}) \oplus ([s], \nu_{0, \delta}) \oplus ([r s], \nu'_{0, \epsilon + \delta}) \\ 
        I(\rho) =& ([1], \rho) \oplus ([r^2], \rho) \oplus (r, \lambda_1 \oplus \lambda_2) \oplus ([s], \nu_{1,0} \oplus \nu_{1,1}) \oplus ([r s], \nu'_{10} \oplus \nu'_{11})
    \end{align}

    $F_\phi$ induced by $\phi: D_4 \to D_4$ is given by
    \begin{align}
        ([r^a s^b], \pi) \mapsto ([\phi(r^a s^b)], \pi \circ \phi^{-1}).
    \end{align}
    The associated tensorator is the identity on the underlying vector space. 
    For instance, $([1], \chi_{\epsilon, \delta}) \mapsto ([1], \chi_{\epsilon, \epsilon+\delta})$, $([r^2], \chi_{\epsilon, \delta}) \mapsto ([r^2], \chi_{\epsilon, \epsilon+\delta})$, and $([1], \rho), ([r^2], \rho)$ are fixed up to isomorphism.

    Let $\alpha: \caA \to \caA$ be a $\ZZ$-equivariant bounded-spread isomorphism with spread $s$ and suppose $\DHR(\alpha) \cong F_\phi$. 
    Then, by \autoref{cor:fusionspinchainalgebra-obstruction}, for a sufficiently large $k \gg s, n$, there exists $([r^a s^b], \psi) \in \Inv (\caZ (\caC))$ such that $([r^a  s^b], \psi) \otimes I(X^{\otimes k}) \cong F_\phi (I(X^{\otimes k}))$. 
    Here, 
    \begin{align}
        \Inv (\caZ (\caC)) = \paren*{([1], \chi_{\epsilon, \delta}), ([1], \rho), ([r^2], \chi_{\epsilon, \delta}), ([r^2], \rho)}
    \end{align}

    Notice that the condition \ref{eq:ex-fusionspinchainalgebra-d4-condition} is stable under taking tensor power.
    If $a_{(0,0)} \neq a_{(0,1)}$, $a_{(1,0)} \neq a_{(1,1)}$, then 
    \begin{align}
        X^{\otimes 2} =& \paren*{a_{(0,0)}^2 + a_{(0,1)}^2 + a_{(1,0)}^2 + a_{(1,1)}^2 + b^2} \chi_{(0,0)} \\ 
        +& \paren*{2 a_{(0,0)} a_{(0,1)} + 2 a_{(1,0)} a_{(1,1)} + b^2} \chi_{(0,1)} \\ 
        +& \paren*{2 a_{(0,0)} a_{(1,0)} + 2 a_{(0,1)} a_{(1,1)} + b^2} \chi_{(1,0)} \\
        +& \paren*{2 a_{(0,0)} a_{(1,1)} + 2 a_{(0,1)} a_{(1,0)} + b^2} \chi_{(1,1)} \\
        +& 2 b \paren*{a_{(0,0)} + a_{(0,1)} + a_{(1,0)} + a_{(1,1)}} \rho. 
    \end{align}
    Comparing the coefficients for $\chi_{(0,0)}, \chi_{(0,1)}$, 
    \begin{align}
        & \paren*{a_{(0,0)}^2 + a_{(0,1)}^2 + a_{(1,0)}^2 + a_{(1,1)}^2} - \paren*{2 a_{(0,0)} a_{(0,1)} + 2 a_{(1,0)} a_{(1,1)}} \\ 
        =& \paren*{a_{(0,0)} - a_{(0,1)}}^2 + \paren*{a_{(1,0)} - a_{(1,1)}}^2 \neq 0,
    \end{align}
    contradicting the condition \ref{eq:ex-fusionspinchainalgebra-d4-condition}.
    Similarly, the coefficients for $\chi_{(1,0)}, \chi_{(1,1)}$ are different.
    Thus, let $a_{\epsilon, \delta}, b$ denote the coefficients of $X^{\otimes k}$ with the condition satisifed. 
    \begin{align}
        I(X^{\otimes k}) =& \paren*{[1], \paren*{\bigoplus a_{\epsilon, \delta} \chi_{\epsilon, \delta}}} \oplus \paren*{[r^2], \paren*{\bigoplus a_{\epsilon, \delta} \chi_{\epsilon, \delta}}} \oplus \cdots \\
        F_\phi (I(X^{\otimes k})) =& \paren*{[1], \paren*{\bigoplus a_{\epsilon, \delta} \chi_{\epsilon, \epsilon + \delta}}} \oplus \paren*{[r^2], \paren*{\bigoplus a_{\epsilon, \delta} \chi_{\epsilon, \epsilon + \delta}}} \oplus \cdots .
    \end{align}
    However,
    \begin{align}
        ([1], \chi_{\epsilon', \delta'}) \otimes I(X^{\otimes k}) =& \paren*{[1], \paren*{\bigoplus a_{\epsilon, \delta} \chi_{\epsilon + \epsilon', \delta + \delta'}}} \oplus \paren*{[r^2], \paren*{\bigoplus a_{\epsilon, \delta} \chi_{\epsilon + \epsilon', \delta + \delta'}}} \oplus \cdots \\
        ([r^2], \chi_{\epsilon', \delta'}) \otimes I(X^{\otimes k}) =& \paren*{[1], \paren*{\bigoplus a_{\epsilon, \delta} \chi_{\epsilon + \epsilon', \delta + \delta'}}} \oplus \paren*{[r^2], \paren*{\bigoplus a_{\epsilon, \delta} \chi_{\epsilon + \epsilon', \delta + \delta'}}} \oplus \cdots .
    \end{align}
    The condition \ref{eq:ex-fusionspinchainalgebra-d4-condition} implies they are non-isomorphic to $F_\phi (I(X^{\otimes k}))$.
\end{proof}

We can generalize the result in another direction, namely to Tambara-Yamagami categories, with a condition on the choice of $A$ and the braided autoequivalence.
\autoref{eq:fusionspinchainalgebra-ty-condition2} below becomes easily testable once the group and the automorphism are specified. 
\begin{example}
    Let $A$ be a finite abelian group, $\chi: A \times A \to \CC^{\times}$ a nondegenerate symmetric bicharacter, $\tau \in \{\pm \sqrt{|A|^{-1}}\}$, $\caC = \TY (A, \chi, \tau)$ be the Tambara-Yamagami category, and $\phi \in \Aut (A, \chi)$ be an isometry that satisfies the following condition:
    \begin{gather}\label{eq:fusionspinchainalgebra-ty-condition2}
        \text{If a function } f: \widehat{A} \to \ZZ \text{ satisfies for any } b \in A \quad \chi_b \cdot f \neq f \circ \phi^* \\ 
        \text{then for any } b \in A \quad \chi_b \cdot f^2 \neq f^2 \circ \phi^*,
    \end{gather}
    where $\chi_b: \widehat{A} \to \CC^{\times}, \psi \mapsto \psi(b)$ and $\phi^*: \widehat{A} \to \widehat{A}, \psi(\cdot) \mapsto \psi(\phi(\cdot))$.
    Let $X = \bigoplus_{a \in A} n_a a \oplus n_m m$ be a strong tensor generating object with $n \ge 0$ such that 
    \begin{align}\label{eq:ex-fusionspinchainalgebra-ty-condition}
        \text{For all } c \in A \quad \chi_c \cdot \widehat{n} \neq \widehat{n} \cdot \phi^*
    \end{align}
    Let $\caA = \caA (\caC, X)$ be the corresponding fusion spin algebra, 
    Then, $F_\phi \in \CAut_{\br} (\caZ (\caC))$ induced by $\phi$ does not arise as $\DHR(\alpha)$ for any $\ZZ$-equivariant quantum cellular automaton $\alpha: \caA \to \caA$.

    This obstruction is stable under adding ancila multiplicity space (i.e., replacing $X$ with $X \otimes \CC^l \cong X^{\oplus l}$.)
\end{example}
\begin{proof}
    Recall the results in \cite{gelaki2009centersgradedfusioncategories}. 
    The simple objects in $\caC$ are labeled by $\{a \in A\} \cup \{m\}$ with fusion rules:
    \begin{align}
        a \otimes b =& (a + b), \\ 
        a \otimes m =& m \otimes a = m, \\
        m \otimes m =& \bigoplus_{a \in A} a
    \end{align}
    and the associators are defined in terms of $\chi, \tau$. 
    The simple objects in $\caZ (\caC)$ is labeled by 
    \begin{enumerate}
        \item $X_{a, \epsilon}$ for $a \in A$, $\epsilon \in \braces*{\pm \sqrt{\chi(a,a)^{-1}}}$.
        \item $Y_{a, b}$ for $a \neq b \in A$.
        \item $Z_{\rho, \Delta}$, where $\rho: A \to \CC^\times$, $\Delta \in \braces*{\pm \sqrt{\tau \prod_{a \in A} \rho(a)^{-1}}}$.
    \end{enumerate}
    The invertible elements are exactly the first family.
    For fusion rules and more detailed analysis, see \cite{gelaki2009centersgradedfusioncategories}.

    \autoref{eq:ex-fusionspinchainalgebra-ty-condition} is equivalent to $n_{(\cdot) - c} \neq n_{\phi^{-1} (\cdot)}$. 
    The assumptions guarantee that this is stable under taking tensor powers. 

    Let $\alpha: \caA \to \caA$ be a $\ZZ$-equivariant bounded-spread automorphism with spread $s$ and suppose $\DHR(\alpha) \cong F_\phi$. 
    Then, by \autoref{cor:fusionspinchainalgebra-obstruction}, for a sufficiently large $k \gg s, n$, there exists $X_{a, \epsilon}$ such that $X_{a, \epsilon} \otimes I(X^{\otimes k} \cong F_{\phi} (I(X^{\otimes k}))$.
    Since the condition \ref{eq:ex-fusionspinchainalgebra-ty-condition} is stable under tensor power, let $n_a$, $n_m$ denote the coefficients of $X^{\otimes k}$ from now. 
    Computing $I$ by reciprocity, 
    \begin{align}
        & X_{a, \epsilon} \otimes I(X^{\otimes k}) \\
        =& X_{a, \epsilon} \otimes \paren*{\bigoplus_{a \in A} n_a \paren*{X_{a, +} \oplus X_{a, -} \oplus \bigoplus_{b \neq a} Y_{a,b}} \oplus n_m \paren*{\bigoplus_{\rho} Z_{\rho, +} \oplus Z_{\rho, -}}} \\ 
        =& \bigoplus_{a \in A} n_{a-c} (X_{a, +} \oplus X_{a, -} \oplus \cdots) \oplus \cdots  \\ 
        & F_\phi (I(X^{\otimes k})) \\ 
        =& \bigoplus_{a \in A} n_{\phi^{-1} (a)} (X_{a, +} \oplus X_{a, -} \oplus \cdots) \oplus \cdots .
    \end{align}
    They are non-isomorphic by assumption.
\end{proof}

More generally, we have a condition on the braided autoequivalence that can be easily checked to test the implementability. 
Indeed, the condition below reduces to linear algebra once we know the fusion ring of $\caZ (\caC)$ and how $I$, $F$ acts on simple objects.
This criterion does not yield a meaningful answer if, for example, $F \in \CAut_\br (\caZ (\caC))$ acts trivially at the object level.

\begin{example}
    Let $\caC$ be a unitary fusion category and let $F \in \CAut_{\br} (\caZ (\caC))$ be a braided autoequivalence such that for all $g \in \Inv (\caZ (\caC))$, there exists $Y \in \Ob (\caC)$ such that $g \otimes F(I(Y)) \not\cong I(Y)$.

    Then, there exists $X \in \Ob (\caC)$ such that $F \not\cong \DHR (\alpha)$ for any quantum cellular automaton $\alpha$ on the quasi-local *-algebra $\caA (\caC, X)$.

    This obstruction is stable under adding ancila multiplicity space (i.e., replacing $X$ with $X \otimes \CC^{l} \cong X^{\oplus l}$.)
\end{example}
\begin{proof}
    Let $S_g := \{Z \in \Ob (\caC) \mid g \otimes F(I(Z)) \cong I(Z)\}$.
    The assumption implies $K_0 (S_g) \subsetneq K_0 (\caC)$, where $K_0$ is the Grothendieck group.
    Turning to vector spaces, $K_0 (S_g) \otimes_\ZZ \QQ \subsetneq K_0 (\caC) \otimes_\ZZ \QQ$ and
    \begin{align}
        \paren*{\bigcup_{g \in \Inv (\caZ (\caC))} K_0 (S_g) \otimes_\ZZ \QQ} \cap {(\QQ^+)}^r \subsetneq \paren*{K_0 (\caC) \otimes_\ZZ \QQ} \cap {(\QQ^+)}^r,
    \end{align}
    where ${(\QQ^+)}^r \subseteq \QQ^r$ is the positive cone.
    Take $[X'] \in K_0 (\caC) \otimes_\ZZ \QQ \setminus \bigcup_{g \in \Inv (\caZ (\caC))} K_0 (S_g) \otimes_\ZZ \QQ$ with nonnegative coefficients and clear denominators to obtain $X \in \caC$.
\end{proof}

\begin{corollary}\label{cor:fusionspinchainalgebra-dhr-picard-2group}
    Let $\caA = \caA(\caC, X)$ be a fusion spin chain algebra with $\caC$ a unitary fusion category and $X$ a strong tensor generator with $n \ge 1$.
    Let $\alpha_{(\cdot)}: G \to \QLstAlg_{\ZZ, \ZZ} (\caA, \caA)$ be a group homomorphism from a finite group and $s < \infty$ be at least the maximum spread of $\{\alpha_g; g \in G\}$, $k > \max \{14 n + 4 s, 6 n + 12 s\}$.
    Then, the 2-group homomorphism $G \to \CAut_{\br} (\caZ (\caC)), g \mapsto \DHR(\alpha_g)$ lifts to 
    \begin{equation}\begin{tikzcd}
        & (\CAut_{\br} (\caZ (\caC)) \ltimes \Inv (\caZ (\caC)))_{I(X^{\otimes k})} \arrow[d, "\Cat{Forget}"] \\ 
        G \arrow{r}[swap]{\DHR(\alpha_{(\cdot)})} \arrow[ru] & \CAut_{\br} (\caZ (\caC))
    \end{tikzcd}\end{equation}
    where $(\CAut_{\br} (\caZ (\caC)) \ltimes \Inv (\caZ (\caC)))_{I(X^{\otimes k})}$ consists of $(F, \caL, \gamma: I (X^{\otimes k}) \stackrel{\sim}{\to} \caL \otimes F(I(X^{\otimes k})))$ and $(\eta, \phi): (F, \caL, \gamma) \to (G, \caM, \gamma')$ with $\eta: F \to G$ a braided tensor natural isomorphism and $\phi: \caL \to \caM$ an isomorphism such that 
    \begin{equation}\begin{tikzcd}
        I(X^{\otimes k}) \arrow[r, "\gamma"] \arrow[rd, "\gamma'"] & \caL \otimes F(I(X^{\otimes k})) \arrow{d}{\phi \otimes \eta_{I(X^{\otimes k})}} \\
        & \caM \otimes G(I(X^{\otimes k})) 
    \end{tikzcd}\end{equation}
\end{corollary}
\begin{proof}
    Recall given $\alpha_{(\cdot)}: G \to \QLstAlg(\caA, \caA)$, we have a strict monoidal functor $G \to \CAut_{\br} (_\caA \DHR_\caA), g \mapsto \DHR(\alpha_g)$. 
    Through braided equivalence $M: \caZ (\caC) \to _\caA \DHR_\caA$ and a quasi-inverse $M^{-1}$, there is a 2-group homomorphism $G \to \CAut_{\br} (\caZ (\caC))$ given by $g \mapsto M^{-1} \circ \DHR(\alpha_g) \circ M$.

    $\CAut_{\br} (\caZ (\caC)) \ltimes \Inv (\caZ (\caC))$ is a 2-group with multiplication $(F, \caL) \otimes (G, \caM) = (F \circ G, \caL \otimes F(\caM))$ and unit $(\Cat{id}_{\caZ (\caC)}, 1_{\caZ (\caC)})$ and associator 
    \begin{align}
        (F, \caL) \otimes ((G, \caM) \otimes (H, \caN)) 
        =& (F \circ G \circ H, \caL \otimes F(\caM \otimes G(\caN))) \\
        \xrightarrow{(\Cat{id}_{F \circ G \circ H}, \alpha_{\caL, \caM, \caN} \circ F^2_{\caM, G(\caN)})}& (F \circ G \circ H, (\caL \otimes F(\caM)) \otimes (F \circ G)(\caN)) \\ 
        =& ((F, \caL) \otimes (G, \caM)) \otimes (H, \caN),
    \end{align}
    where $\alpha$ is the associator of $\caZ (\caC)$ and $F^2_{(\cdot), (\cdot)}: F(\cdot) \otimes F(\cdot) \to F(\cdot \otimes \cdot)$ is the tensorator. 
    The action on $\caZ (\caC)$ is given by $(F, \caL) \triangleright Z = \caL \otimes F(Z)$ with obvious associator and the unit map.    
\end{proof}

\printbibliography

@misc{AFM20,
  author = {David Aasen and Paul Fendley and Roger S. K. Mong},
  title = {Topological defects on the lattice: Dualities and degeneracies},
  year = {2020},
  note = {arXiv:2008.08598},
  eprint = {2008.08598},
  eprinttype = {arXiv},
  doi = {10.48550/arXiv.2008.08598}
}

@misc{EF23,
  author = {Luisa Eck and Paul Fendley},
  title = {From the XXZ chain to the integrable Rydberg-blockade ladder via non-invertible duality defects},
  year = {2023},
  note = {arXiv:230x.xxxxx (preprint)}
}

@book{EGNO15,
  author = {Pavel Etingof and Shlomo Gelaki and Dmitri Nikshych and Victor Ostrik},
  title = {Tensor categories},
  publisher = {American Mathematical Society},
  series = {Mathematical Surveys and Monographs},
  volume = {205},
  year = {2015},
  address = {Providence, RI},
  note = {MR3242743, doi:10.1090/surv/205}
}

@article{FH20,
  author = {Michael Freedman and Matthew B. Hastings},
  title = {Classification of quantum cellular automata},
  journal = {Comm. Math. Phys.},
  volume = {376},
  number = {2},
  pages = {1171--1222},
  year = {2020},
  note = {MR4103966}
}

@misc{FMT22,
  author = {Daniel S. Freed and Gregory W. Moore and Constantin Teleman},
  title = {Topological symmetry in quantum field theory},
  year = {2022},
  note = {preprint}
}

@book{Haa96,
  author       = {Rudolf Haag},
  title        = {Local Quantum Physics: Fields, Particles, Algebras},
  subtitle     = {Fields, Particles, Algebras},
  series       = {Theoretical and Mathematical Physics},
  edition      = {2},
  publisher    = {Springer Berlin, Heidelberg},
  year         = {1996},
  doi          = {10.1007/978-3-642-61458-3},
  isbn         = {978-3-540-61049-6},
  eisbn        = {978-3-642-61458-3},
  pages        = {XV + 390},
  url          = {https://link.springer.com/book/10.1007/978-3-642-61458-3},
}

@misc{LDOV21,
  author = {Laurens Lootens and Clement Delcamp and Gerardo Ortiz and Frank Verstraete},
  title = {Dualities in one-dimensional quantum lattice models: symmetric hamiltonians and matrix product operator intertwiners},
  year = {2021},
  note = {arXiv:2112.09091},
  eprint = {2112.09091},
  eprinttype = {arXiv},
  doi = {10.48550/arXiv.2112.09091}
}

@article{Mug03,
  author = {Michael M\"uger},
  title = {From subfactors to categories and topology. II. The quantum double of tensor categories and subfactors},
  journal = {J. Pure Appl. Algebra},
  volume = {180},
  number = {1-2},
  pages = {159--219},
  year = {2003},
  note = {MR1966525},
  doi = {10.1016/S0022-4049(02)00248-7},
  eprint = {math.CT/0111205}
}

@article{Naa11,
  author = {Pieter Naaijkens},
  title = {Localized endomorphisms in Kitaev's toric code on the plane},
  journal = {Rev. Math. Phys.},
  volume = {23},
  number = {4},
  pages = {347--373},
  year = {2011},
  note = {MR2804555}
}

@incollection{Naa15,
  author = {Pieter Naaijkens},
  title = {Kitaev's quantum double model from a local quantum physics point of view},
  booktitle = {Advances in algebraic quantum field theory},
  series = {Math. Phys. Stud.},
  pages = {365--395},
  publisher = {Springer},
  address = {Cham},
  year = {2015},
  note = {MR3410706}
}

@book{Naa17,
  author = {Pieter Naaijkens},
  title = {Quantum spin systems on infinite lattices},
  series = {Lecture Notes in Physics},
  volume = {933},
  publisher = {Springer},
  address = {Cham},
  year = {2017},
  note = {A concise introduction, MR3617688}
}

@article{NS97,
  author = {Florian Nill and Korn\'el Szlach\'anyi},
  title = {Quantum chains of Hopf algebras with quantum double cosymmetry},
  journal = {Comm. Math. Phys.},
  volume = {187},
  number = {1},
  pages = {159--200},
  year = {1997},
  note = {MR1463825}
}

@misc{SW04,
  author = {B. Schumacher and R. F. Werner},
  title = {Reversible quantum cellular automata},
  year = {2004},
  note = {arXiv:quant-ph/0405174},
  eprint = {quant-ph/0405174},
  eprinttype = {arXiv},
  doi = {10.48550/arXiv.quant-ph/0405174}
}

@book{ZCZW19,
  author = {Bei Zeng and Xie Chen and Duan-Lu Zhou and Xiao-Gang Wen},
  title = {Quantum information meets quantum matter},
  publisher = {Springer},
  series = {Quantum Science and Technology},
  year = {2019},
  address = {New York},
  note = {From quantum entanglement to topological phases of many-body systems; Foreword by John Preskill, MR3929747}
}

@article{LevinWen05,
  author = {Michael A. Levin and Xiao-Gang Wen},
  title = {String-net condensation: A physical mechanism for topological phases},
  journal = {Phys. Rev. B},
  volume = {71},
  pages = {045110},
  year = {2005},
  doi = {10.1103/PhysRevB.71.045110}
}

@article{KITAEV20032,
title = {Fault-tolerant quantum computation by anyons},
journal = {Annals of Physics},
volume = {303},
number = {1},
pages = {2-30},
year = {2003},
issn = {0003-4916},
doi = {https://doi.org/10.1016/S0003-4916(02)00018-0},
url = {https://www.sciencedirect.com/science/article/pii/S0003491602000180},
author = {A.Yu. Kitaev}
}

@misc{kawagoe2024levinwengaugetheoryentanglement,
      title={Levin-Wen is a gauge theory: entanglement from topology}, 
      author={Kyle Kawagoe and Corey Jones and Sean Sanford and David Green and David Penneys},
      year={2024},
      eprint={2401.13838},
      archivePrefix={arXiv},
      primaryClass={cond-mat.str-el},
      url={https://arxiv.org/abs/2401.13838}, 
}

@misc{jones2024dhrbimodulesquasilocalalgebras,
      title={DHR bimodules of quasi-local algebras and symmetric quantum cellular automata}, 
      author={Corey Jones},
      year={2024},
      eprint={2304.00068},
      archivePrefix={arXiv},
      primaryClass={math-ph},
      url={https://arxiv.org/abs/2304.00068}, 
}

@article{KW41,
  title = {Statistics of the Two-Dimensional Ferromagnet. Part I},
  author = {Kramers, H. A. and Wannier, G. H.},
  journal = {Phys. Rev.},
  volume = {60},
  issue = {3},
  pages = {252--262},
  numpages = {0},
  year = {1941},
  month = {Aug},
  publisher = {American Physical Society},
  doi = {10.1103/PhysRev.60.252},
  url = {https://link.aps.org/doi/10.1103/PhysRev.60.252}
}

@book{Arveson1976Invitation,
  author       = {William Arveson},
  title        = {An Invitation to C*-Algebras},
  year         = {1976},
  publisher    = {Springer-Verlag New York, Inc.},
  series       = {Graduate Texts in Mathematics},
  volume       = {39},
  edition      = {1},
  pages        = {X + 108},
  isbn         = {978-0-387-90176-3},
  isbn2        = {978-1-4612-6373-9},
  isbn3        = {978-1-4612-6371-5},
  doi          = {10.1007/978-1-4612-6371-5},
  url          = {https://doi.org/10.1007/978-1-4612-6371-5},
  note         = {Springer Book Archive; Topics: Algebra; Series ISSN: 0072-5285; Series E-ISSN: 2197-5612}
}

@incollection{Day1970Closed,
  author       = {Brian Day},
  title        = {On Closed Categories of Functors},
  booktitle    = {Reports of the Midwest Category Seminar, IV},
  series       = {Lecture Notes in Mathematics},
  volume       = {137},
  pages        = {1--38},
  publisher    = {Springer},
  address      = {Berlin},
  year         = {1970},
  mrnumber     = {0272852}
}

@incollection{Ocneanu1994Chirality,
  author       = {Adrian Ocneanu},
  title        = {Chirality for Operator Algebras},
  booktitle    = {Subfactors: Proceedings of the Taniguchi Symposium on Operator Algebras (Kyuzeso, 1993)},
  pages        = {39--63},
  year         = {1994},
  publisher    = {World Scientific},
  address      = {River Edge, NJ},
  mrnumber     = {1317353}
}

@misc{das2014drinfeldcenterplanaralgebra,
      title={Drinfeld center of planar algebra}, 
      author={Paramita Das and Shamindra Kumar Ghosh and Ved Prakash Gupta},
      year={2014},
      eprint={1203.3958},
      archivePrefix={arXiv},
      primaryClass={math.QA},
      url={https://arxiv.org/abs/1203.3958}, 
}

@article{Jones_2017,
   title={Operator Algebras in Rigid C*-Tensor Categories},
   volume={355},
   ISSN={1432-0916},
   url={http://dx.doi.org/10.1007/s00220-017-2964-0},
   DOI={10.1007/s00220-017-2964-0},
   number={3},
   journal={Communications in Mathematical Physics},
   publisher={Springer Science and Business Media LLC},
   author={Jones, Corey and Penneys, David},
   year={2017},
   month=aug, pages={1121–1188} }

@misc{henriques2015categorifiedtracemoduletensor,
      title={Categorified trace for module tensor categories over braided tensor categories}, 
      author={André Henriques and David Penneys and James Tener},
      year={2015},
      eprint={1509.02937},
      archivePrefix={arXiv},
      primaryClass={math.QA},
      url={https://arxiv.org/abs/1509.02937}, 
}

@misc{gelaki2009centersgradedfusioncategories,
      title={Centers of graded fusion categories}, 
      author={Shlomo Gelaki and Deepak Naidu and Dmitri Nikshych},
      year={2009},
      eprint={0905.3117},
      archivePrefix={arXiv},
      primaryClass={math.QA},
      url={https://arxiv.org/abs/0905.3117}, 
}

@misc{tu2025anomaliesglobalsymmetrieslattice,
      title={Anomalies of global symmetries on the lattice}, 
      author={Yi-Ting Tu and David M. Long and Dominic V. Else},
      year={2025},
      eprint={2507.21209},
      archivePrefix={arXiv},
      primaryClass={cond-mat.str-el},
      url={https://arxiv.org/abs/2507.21209}, 
}

@article{PhysRevLett93070601,
  title = {Kramers-Wannier Duality from Conformal Defects},
  author = {Fr\"ohlich, J\"urg and Fuchs, J\"urgen and Runkel, Ingo and Schweigert, Christoph},
  journal = {Phys. Rev. Lett.},
  volume = {93},
  issue = {7},
  pages = {070601},
  numpages = {4},
  year = {2004},
  month = {Aug},
  publisher = {American Physical Society},
  doi = {10.1103/PhysRevLett.93.070601},
  url = {https://link.aps.org/doi/10.1103/PhysRevLett.93.070601}
}

@misc{evans2025operatoralgebraicapproachfusion,
      title={An operator algebraic approach to fusion category symmetry on the lattice}, 
      author={David E. Evans and Corey Jones},
      year={2025},
      eprint={2507.05185},
      archivePrefix={arXiv},
      primaryClass={math-ph},
      url={https://arxiv.org/abs/2507.05185}, 
}

\end{document}